\documentclass[a4paper,twocolumn,10pt,accepted=2024-02-02]{quantumarticle}
\pdfoutput=1

\usepackage[numbers,sort&compress]{natbib}

\usepackage{graphicx} 
\usepackage{soul,xcolor}

\usepackage{braket}
\usepackage{array}   
\newcolumntype{L}{>{$}l<{$}} 
\usepackage[all]{nowidow}

\usepackage{hyperref}
\usepackage{tikz}
\usetikzlibrary{quantikz}

\usepackage{caption}
\usepackage{subcaption}
\usepackage{adjustbox}
\usepackage{graphicx}
\graphicspath{{figures/}}

\usepackage{empheq}
\usepackage[most]{tcolorbox}
\newtcbox{\mymath}[1][]{%
    nobeforeafter, math upper, tcbox raise base,
    enhanced, colframe=blue!30!black,
    colback=white, boxrule=1pt,
    #1}
	
\usepackage{acronym}
\newacro{QSP}{quantum state preparation}
\newacro{SOCP}{second order cone program}
\newacro{QSVT}{quantum singular value transform}
\newacro{SVD}{singular value decomposition}
\newacro{QRAM}{quantum random access memory}

\newcommand{\shiftright}[2]{\makebox[#1][r]{\makebox[0pt][l]{#2}}}

\newcommand{\numq}[2]{\raisebox{#1}{/${}^{#2}$}}
\newcommand{\rb}[2]{\raisebox{#1}{#2}}

\newcommand{\mltg}[2]{\gate[#1]{#2}}

\usepackage{float}
\usepackage{bm}
\usepackage{algorithm}
\usepackage{algpseudocode}
\usepackage{diagbox}

\usepackage[english]{babel}
\usepackage{amsthm}

\newtheorem{theorem}{Theorem}[section]
\newtheorem{lemma}[theorem]{Lemma}

\newtheorem{corollary}{Corollary}[theorem]

\usepackage{amssymb}

\usepackage{marvosym}

\usepackage{hyperref}
\hypersetup{
    colorlinks=true,
    linkcolor=blue,
    urlcolor=blue,
    }

\newcommand\starfootnotetext[1]{%
  \begingroup
  \renewcommand\thefootnote{*}\footnotetext{#1}%
  \addtocounter{footnote}{-1}%
  \endgroup
}
    
\begin{document}
\title{Spacetime-Efficient Low-Depth Quantum State Preparation with Applications} 
\author{Kaiwen Gui${}^{\,*,}$}
\affiliation{Amazon Web Services, WA, USA}
\affiliation{Pritzker School of Molecular Engineering, University of Chicago, IL, USA}
\affiliation{Department of Computer Science, University of Chicago, IL, USA}
\author{Alexander M.~Dalzell${}^{\,*,}$}
\affiliation{AWS Center for Quantum Computing, Pasadena, CA, USA}
\author{Alessandro Achille}
\affiliation{AWS AI Labs, Pasadena, CA, USA}
\author{Martin Suchara}
\affiliation{Amazon Web Services, WA, USA}
\author{Frederic T. Chong}
\affiliation{Department of Computer Science, University of Chicago, IL, USA}


\begin{abstract}
We propose a novel deterministic method for preparing arbitrary quantum states. When our protocol is compiled into CNOT and arbitrary single-qubit gates, it prepares an $N$-dimensional state in depth $O(\log(N))$ and \textit{spacetime allocation} (a metric that accounts for the fact that oftentimes some ancilla qubits need not be active for the entire circuit) $O(N)$, which are both optimal. When compiled into the $\{\mathrm{H,S,T,CNOT}\}$ gate set, we show that it requires asymptotically fewer quantum resources than previous methods. Specifically, it prepares an arbitrary state up to error $\epsilon$ with optimal depth of $O(\log(N) + \log (1/\epsilon))$ and spacetime allocation $O(N\log(\log(N)/\epsilon))$, improving over $O(\log(N)\log(\log (N)/\epsilon))$ and $O(N\log(N/\epsilon))$, respectively. We illustrate how the reduced spacetime allocation of our protocol enables rapid preparation of many disjoint states with only constant-factor ancilla overhead---$O(N)$ ancilla qubits are reused efficiently to prepare a product state of $w$ $N$-dimensional states in depth $O(w + \log(N))$ rather than $O(w\log(N))$, achieving effectively constant depth per state.  We highlight several applications where this ability would be useful, including quantum machine learning, Hamiltonian simulation, and solving linear systems of equations. We provide quantum circuit descriptions of our protocol, detailed pseudocode, and gate-level implementation examples using Braket.
\end{abstract}

\maketitle
\starfootnotetext{These two authors contributed equally; \href{mailto:kgui@uchicago.edu}{kgui@uchicago.edu}, \href{mailto:dalzel@amazon.com}{dalzel@amazon.com}}
\section{Introduction}\label{sec:intro}
\setstcolor{magenta}
Quantum state preparation (QSP) is a crucial subroutine in many proposed quantum algorithms that claim speedup over their classical counterparts in applications such as quantum machine learning \cite{biamonte2017quantum, lloyd2014quantum, kerenidis2017quantum, rebentrost2018quantum, kerenidis2019q, kerenidis2021quantum, rebentrost2014quantum, schuld2021machine}, simulating quantum systems \cite{berry2015simulating, berry2015hamiltonian, low2017optimal, low2019hamiltonian}, solving linear systems of equations \cite{harrow2009quantum, ambainis2012variable, wossnig2018quantum}, and synthesizing unitary operations \cite{low2018trading, sun2023asymptotically, yuan2023optimal}.

\begin{figure}[!t]
    \centering
    \includegraphics[width=0.48\textwidth]{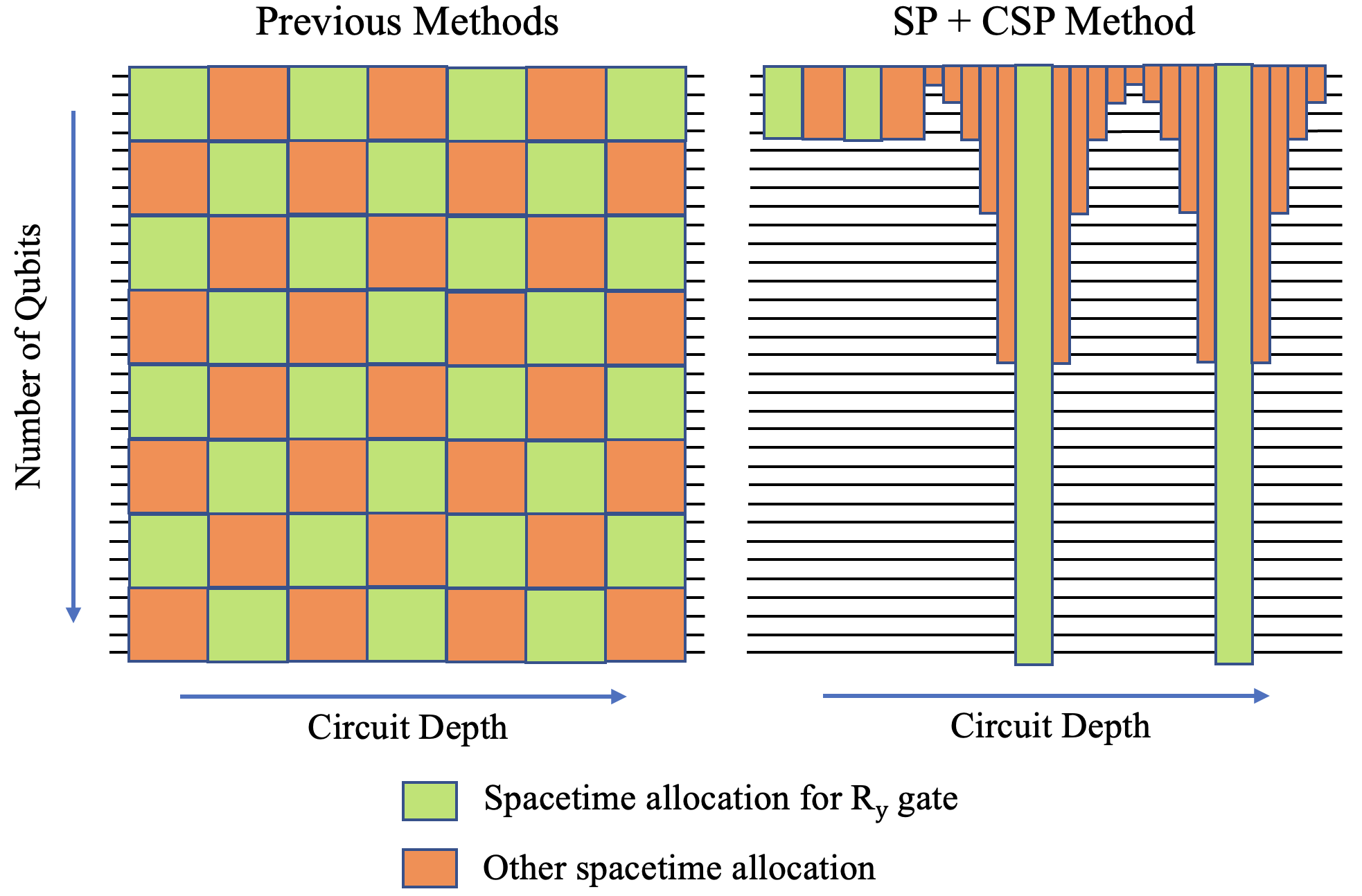}
    \caption{Illustration of spacetime allocation in our state preparation method (SP+CSP) vs. previous methods. Green regions correspond to spacetime allocated for arbitrary single-qubit rotations. Orange regions correspond to spacetime allocated to qubits that are active (i.e., not in $\ket{0}$) without single-qubit rotations; they might experience, e.g., CNOT or Toffoli gates or might be idling while entangled with other qubits. In our method, the spacetime allocation is asymptotically smaller than the product of qubit count and depth.}
    \label{fig:SA_comps}
\end{figure}

The state preparation problem is to create an arbitrary quantum state of the form:
\begin{equation} \label{Eq:quantum_state}
    \ket{\psi} = \frac{1}{\lVert \textbf{x}\rVert}\sum_{i=0}^{2^n-1} x_i \ket{i}
\end{equation}
where $n$ denotes the number of qubits, $\lVert \cdot \rVert$ denotes the standard Euclidean norm, and $\mathbf{x}$ is a $2^n$-dimensional vector with components $x_i \in \mathbb{C}$. In many applications, it is sufficient to let $x_i \in \mathbb{R}^+ \cup \{0\}$, and henceforth, we assume this for simplicity; the additional phase information can be easily incorporated with minimal overhead (see App.~\ref{app:complex-amplitudes}). We also denote $N = 2^n$ as the total number of parameters encoded, the same as the dimension of the vector space.

\begingroup
\begin{table*}[t]
\setlength{\tabcolsep}{2pt} 
\renewcommand{\arraystretch}{1.4} 
\centering
\begin{adjustbox}{max width=\textwidth}
 \begin{tabular}{c|c|c||c|c} 
 Gate Set &  \multicolumn{2}{c}{\{U(2), CNOT\}} & \multicolumn{2}{c}{$\{\mathrm{H, S, T, CNOT}\}$}\\
 \hline\hline 
& Depth & Spacetime Allocation & Depth &  Spacetime Allocation \\
 \hline
 \hline
 Sun et al.~\cite{sun2023asymptotically} & $O(\log(N)\log(\log(N)))$ & $\Theta(N)$ & $O(\log(N)\log(\log(N))\log(N/\epsilon))$ & $O(N\log(N/\epsilon))$ \\
 \hline
 Sun et al.~\cite{sun2023asymptotically} & $\Theta(\log(N))$ & $O(N\log(N))$ & $O(\log(N)\log(N/\epsilon))$ & $O(N\log(N)\log(N/\epsilon))$ \\
 \hline
 Zhang et al. \cite{zhang2022quantum} & $\Theta(\log(N))$ & $O(N\log(N))$ & $O(\log(N)\log(\log(N)/\epsilon))$ & $O(N\log(N)\log(\log(N)/\epsilon))$ \\
 \hline

Yuan et al.~\cite{yuan2023optimal} & $\Theta(\log(N))$ & $\Theta(N)$ & $\Omega(\log(N)\log(\log(N)/\epsilon))^{\dagger}$ & $\Omega(N\log(\log(N)/\epsilon))^{\dagger}$ \\
 \hline
 Clader et al.~\cite{clader2023quantum} & $O(\log^2(N))$ & $O(N\log^2(N))$ & $O(\log^2(N) + \log(1/\epsilon))$ & $O(N(\log^2(N) + \log(1/\epsilon)))$ \\
 \hline
  \textbf{This Work (SP)} & $\bm{\Theta(\log(N))}$ & $\bm{O(N\log(N))}$ & $\bm{\Theta(\log(N) + \log(1/\epsilon))}$ & $\bm{O(N\log(N/\epsilon))}$ \\
 \hline
 \textbf{This Work (SP+CSP)} & $\bm{\Theta(\log(N))}$ & $\bm{\Theta(N)}$ & $\bm{\Theta(\log(N) + \log(1/\epsilon))}$ & $\bm{O(N\log(\log(N)/\epsilon))}$ \\
 \hline
 \hline
 \end{tabular}
 \end{adjustbox}
\caption{Comparison of our results to previous state-of-the-art low-depth QSP methods.  $N = 2^n$ is the total number of basis/parameters, and $\epsilon$ is the error precision parameter. Note that Sun et al.~\cite{sun2023asymptotically}  proposed additional variations of the QSP protocol with larger circuit depth and fewer ancilla qubits that all have spacetime allocation of $\Theta(N)$ and $O(N\log(N/\epsilon))$. The dagger $^{\dagger}$ in the row for Ref.~\cite{yuan2023optimal} indicates lower bounds from our calculation because the paper did not analyze the Clifford + T costs. Note that T depth is optimized in Clader et al.~\cite{clader2023quantum}. We also note that our SP+CSP protocol allows some ancilla registers to be dirty, similar to that in \cite{low2018trading}. See the corresponding theoretical lower bound for spacetime allocation using \{H, S, T, CNOT\} gate set in Sec. \ref{sec:LB}.}
 \label{table:QRE_comps}
\end{table*}
\endgroup

Recent advancement by Sun et al.~\cite{sun2023asymptotically} gave an optimal construction that creates $\ket{\psi}$ with $\Theta(2^n/n)$\footnote{Throughout the paper, we use $O(f(n,1/\epsilon))$ notation to indicate that the cost is upper bounded by a constant number times $f(n,1/\epsilon)$ as $n,1/\epsilon \rightarrow \infty$, $\Omega(f(n,1/\epsilon))$ to denote lower bounds, and $\Theta(f(n,1/\epsilon))$ when there are upper and lower bounds that match up to constant factors.} circuit depth using arbitrary single-qubit gates and two-qubit CNOT gates (henceforth called the $\{\mathrm{U}(2), \mathrm{CNOT}\}$ gate set) and no ancilla qubits. On the other hand, if ancilla qubits are available, one can dramatically reduce the circuit depth. It is preferable to use low-depth protocols when the state-preparation procedure must be completed quickly or repeated many times sequentially. Oftentimes, the quantum algorithm that follows the preparation of $\ket{\psi}$---for instance, making a machine learning inference---runs in depth $\mathrm{poly}(n)$, exponentially faster than the ancilla-free QSP implementations, so low-depth state preparation is vital to make the overall runtime reasonable. 

Toward that end, recent deterministic methods \cite{sun2023asymptotically, zhang2022quantum, rosenthal2021query} have achieved optimal $\Theta(n)$ quantum circuit depth in the $\{\mathrm{U(2),CNOT}\}$ gate set. However, this exponential depth reduction comes with an exponential overhead in space. In other words, one would need an exponential number of ancilla qubits. 
Recent work \cite{yuan2023optimal} is able to achieve $\Theta(n)$ depth with only $\Theta(2^n/n)$ ancilla qubits, which is optimal.

When considering the total required quantum resource tradeoffs, metrics such as circuit depth, circuit size, and qubit count are typically considered. Assuming the physical architecture allows one to perform gates in parallel, the circuit depth is a proxy for the overall runtime of the computation, and the qubit count represents the overall amount of space that must be allocated to the computation. We now propose another metric---\textbf{spacetime allocation}---the total time that each individual qubit must be active (i.e., not in the $\ket{0}$ state), summed over all qubits. The spacetime allocation is bounded below by the circuit size and bounded above by the product of the qubit count and the circuit depth, but it carries a distinct operational meaning. In a model where one wishes to perform many distinct ancilla-intensive jobs (such as rapidly preparing many independent $n$-qubit states) with a fixed number of ancilla qubits, the availability of fresh ancillae in the state $\ket{0}$ becomes the algorithmic bottleneck. Assuming ancillae can be reallocated from one job to another as soon as they are returned to the $\ket{0}$ state, the overall runtime to complete the batch of jobs is determined by the spacetime allocation of the jobs rather than their depth or size. State-preparation algorithms that are optimal in terms of depth or size are not necessarily also optimal in terms of spacetime allocation. We will show that when compiled into the $\{\mathrm{U(2),CNOT}\}$ gate set, our state preparation protocol is simultaneously optimal in depth, size, and spacetime allocation up to constant factors. This feat has already been achieved by the state preparation method of Ref.~\cite{yuan2023optimal}, which has $\Theta(n)$ depth and $\Theta(2^n/n)$ qubit count, along with $\Theta(2^n)$ size and spacetime allocation.

However, in practice, it may not be possible to perform arbitrary single-qubit gates to exact precision, and thus the $\{\mathrm{U}(2), \mathrm{CNOT}\}$ gate set may not be applicable.  In this case, we cannot hope to prepare the state $\ket{\psi}$ exactly; rather, given an error parameter $\epsilon$, we seek to prepare a state $\ket{\tilde{\psi}}$ such that
\begin{equation}
    \left\lVert \ket{\psi} - \ket{\tilde{\psi}} \right\rVert \leq \epsilon\,.
\end{equation}
For example, this is the case in most proposals for fault-tolerant quantum computation based on error-correcting codes, where logical single-qubit gates are approximately performed using a sequence of logical gates drawn from a discrete gate set (see, e.g.,~\cite{ross2016optimal}). 
A common choice of discrete gate set is $\{\mathrm{H, S, T, CNOT}\}$ (defined in the next section); we quote the scaling of the resource cost of our protocol for approximate QSP when compiled into this gate set. 
We use the terminology \textbf{approximate spacetime allocation} and \textbf{exact spacetime allocation}, denoted as SA$_{\mathrm{approx}}$ and SA$_{\mathrm{exact}}$, to differentiate the cost in the two models. Similarly, we denote the circuit depth in the two models by $\mathrm{D}_{\mathrm{approx}}$ and $\mathrm{D}_{\mathrm{exact}}$.

Crucially, in the $\{\mathrm{H,S,T, CNOT}\}$ gate set, for constant target error $\epsilon$, the method of Ref.~\cite{yuan2023optimal} does not achieve $\Theta(n)$ depth. While they do not explicitly report the depth in this gate set, this can be seen by the fact that their circuit has $O(2^n)$ arbitrary single-qubit gates spread out over $O(n)$ distinct layers; to convert to the discrete gate set while incurring constant overall error, each single-qubit gate must be approximately decomposed into a sequence of depth at least $\Omega(\log(n))$,
yielding total depth at least $\Omega(n\log(n))$. Our state preparation method will have depth $\Theta(n)$ even in the discrete gate set, when $\epsilon$ is taken as a constant.

The value of minimizing with respect to spacetime allocation is also apparent in the context of realistic hardware considerations. For example, in gate-based near-term devices without error correction, reducing the spacetime allocation (even as circuit size remains the same) amounts to reducing the time that qubits are left idling. By deallocating and reinitializing ancillas in $\ket{0}$, this kind of idling time can be minimized, and the impacts of noise can be made minimal. Moreover, idling of logical qubits is also costly in fault-tolerant architectures, which require continuous error correction even when no gates are being performed. Each round of error correction requires a set of parity check circuits and measurements, followed by a nontrivial classical decoding calculation, all of which contribute to the overall energy expenditure of the computation.

Another feature of our constructions is that they are garbage-free. This contrasts with some previous work (e.g., \cite{babbush2018encoding, araujo2021divide}), which perform a relaxed version of the QSP task, where, instead of preparing the $n$-qubit state $\ket{\psi}$ from Eq.~\eqref{Eq:quantum_state}, one aims to prepare an $(n+a)$-qubit state $\ket{\Psi}$ given by
\begin{equation}
    \ket{\Psi} = \frac{1}{\lVert \textbf{x} \rVert}\sum_{i=0}^{2^n-1} (x_i \ket{i} \otimes \ket{\text{garbage}_i})\,.
\end{equation}
Typically, $a$ would scale at least linearly in $n$; these $a$ qubits are left entangled with the $n$ data qubits. 

Allowing garbage makes the QSP task easier, and in some applications, garbage is tolerable (since the prepared state only acts as control qubits). However, as long as the $a$ ancilla qubits are entangled with the data, they cannot be used as fresh $\ket{0}$ ancillae for other tasks, and they continue to contribute toward the spacetime allocation of the QSP protocol. Moreover, it is important to emphasize that this is a fundamentally different task. After tracing out the garbage, the $n$-qubit state that results is a mixed state, which lacks the coherence of the pure state in Eq.~\eqref{Eq:quantum_state}. This is especially true when one wants to directly manipulate the prepared states (e.g., in quantum linear system solvers and quantum machine learning applications) rather than only using them as control bits.

Lastly, we point out an interesting feature of our state preparation protocol: a constant fraction of the $O(N)$ ancilla qubits needed to achieve logarithmic depth can be \emph{dirty}, that is, they can start in an arbitrary state, and they will be returned to the state they began in at the end of the circuit. Dirty qubits are appealing because they can be prepared without spending resources to fault tolerantly prepare high-fidelity $\ket{0}$ states. Additionally, dirty qubits can be qubits from another computation experiencing a long period of idling, provided they are returned to their original state before the end of the idling period. This situation is especially salient when using the $\{\mathrm{H,S,T,CNOT}\}$ gate set, since our state preparation circuits involve a layer where some qubits are idling and others experience a single-qubit gate sequence approximating a single-qubit rotation---the idling qubits can be used as dirty ancillae for other state preparation instances. The possibility of state preparation ancillae being dirty has been previously explored in Ref.~\cite{low2018trading}, which studied tradeoffs between the number of T gates and the number of dirty ancillae. We believe with additional innovations, it may be possible for the number of clean qubits to be further reduced to be asymptotically better than $O(N)$. 

The contributions of this work are the following:
\begin{enumerate} 
    \item We propose a novel, deterministic, garbage-free quantum state preparation method, which we call SP+CSP (Section \ref{sec:SP+CSP}), and we show that in the $\{\mathrm{U(2),CNOT}\}$ gate set it simultaneously achieves optimal depth and spacetime allocation.  We also show that in the $\{\mathrm{H,S,T,CNOT}\}$ gate set, it achieves depth and spacetime allocation that are both asymptotically superior to previous methods (see Table \ref{table:QRE_comps}). A matching lower bound shows that the depth-scaling is optimal in this gate set, whereas the best lower bound on spacetime allocation leaves open the possibility of improvements by logarithmic factors.
    \item We show how the optimal spacetime allocation of our protocol allows it to prepare multiple copies of a state more rapidly than other methods 
    (Section \ref{sec:multiple_copy}).
    \item We discuss several applications of the optimal spacetime allocation quantum state preparation protocol (Section \ref{sec:applications}), including
        \begin{enumerate}
            \item Quantum Machine Learning
            \item Hamiltonian simulation
            \item Solving linear systems with HHL-style quantum algorithms
        \end{enumerate}
    \item We provide a circuit-level implementation for the SP+CSP method (Sec.~\ref{sec:SP+CSP} \& App.~\ref{sec:appendix}) with detailed pseudocode.
    \item We provide two gate-level implementation examples using Braket (Sec.~\ref{sec:code}).
\end{enumerate}

We summarize our result compared to other state-of-the-art results in Table \ref{table:QRE_comps}, with a pictorial depiction in Fig. \ref{fig:SA_comps}.

To achieve these results, we build from the work of Clader et al.~\cite{clader2023quantum}, who gave a QSP method with depth $O(n)$ and qubit count $O(2^n)$ under the assumption that the FANOUT-CNOT operation---that is, the product of up to $O(2^n)$ CNOT gates sharing the same control qubit---could be performed in a single time step.  Our protocol (we call it the \textbf{SP} protocol) eliminates the need for FANOUT-CNOT at the expense of only constant-factor overhead in the number of ancilla qubits and the depth. The high-level idea is to use a tree-like data copying circuit consisting only of CNOT gates to copy the control bit into many ancillas, so it can then be used to control many operations in parallel (see App.~\ref{sec:data_copy_circ}). However, this observation is not alone sufficient---to achieve $O(n)$ overall depth, we require a delicate method that alternates between layers of control-bit copying and layers of controlled operations (see App.~\ref{sec:SPF_circ}). Like that of Clader et al., this circuit has the feature that only a constant number of layers involve single-qubit rotations (which incur depth $O(\log(n/\epsilon))$ when approximately decomposed into a finite gate set), allowing the overall depth to scale as $O(n+\log(1/\epsilon))$. 
In comparison, previous methods (e.g., \cite{sun2023asymptotically, zhang2022quantum}) assemble these single-qubit rotations across at least $n$ layers, which would then make the total depth at least $\Omega(n\log (n/\epsilon))$.

The SP protocol (and protocols designed in previous works such as \cite{sun2023asymptotically} and \cite{zhang2022quantum}) does not achieve optimal spacetime allocation since it fundamentally requires $O(2^n)$ ancilla qubits to be entangled with the $O(n)$ data qubits for a constant fraction of the $O(n)$ circuit depth, leading to $O(n2^n)$ spacetime allocation (similar to what is illustrated on the left side of Fig.~\ref{fig:SA_comps}). To circumvent this, we pursue an additional idea: prepare roughly half of the qubits using the SP method, and then perform \emph{controlled state preparation} (\textbf{CSP}) from those qubits into the rest of the qubits.  Both SP and CSP require only $O(1)$ layers of single-qubit rotations, preserving the $O(n+\log(1/\epsilon))$ depth for the approximate compilation. Moreover, the full $O(2^n)$ ancilla qubits are only needed very briefly during the CSP procedure, and most can be freed up after only $O(1)$ depth (illustrated as the right side of Fig.~\ref{fig:SA_comps}). Ultimately, we show that $O(2^n)$ spacetime allocation can be achieved. We also give a detailed circuit implementation that shows the difference in spacetime allocation requirements in App.~\ref{sec:app_LOADF}. 


\section{Background}
\label{sec:quantum_computing_background}
This section will introduce the required quantum gates and provide more details about the spacetime allocation metric, as well as relevant lower bounds and a summary of the relevant prior work. 

\subsection{Quantum Gates}
In this work, we will use the following quantum gates: X gate, S gate, T gate, Hadamard (H) gate, $\mathrm{R}_y$ gate, CNOT gate, SWAP gate, controlled-$\mathrm{R}_y$ gate, Fredkin gate, and Toffoli gate.
\begingroup
\renewcommand{\arraystretch}{1.5} 
\begin{table}[h]
    \centering
    \begin{tabular}{ c  c }
    \toprule \\[0.1cm]
    \textbf{Gate} & \textbf{Unitary}\\[0.1cm]
    X & $\begin{pmatrix} 0 & 1 \\[2pt] 1 & 0 \end{pmatrix}$ \\[10pt]
    \vspace{0.2cm}
    Hadamard (H) & $\begin{pmatrix} \frac{1}{\sqrt{2}} & \frac{1}{\sqrt{2}} \\[2pt]  \frac{1}{\sqrt{2}} & -\frac{1}{\sqrt{2}} \end{pmatrix}$ \\
    \vspace{0.2cm}
    Y-Rotation ($\mathrm{R}_y(\theta)$) & $\begin{pmatrix} \cos{\frac{\theta}{2}} & -\sin{\frac{\theta}{2}} \\[2pt]  \sin{\frac{\theta}{2}} & \cos{\frac{\theta}{2}} \end{pmatrix}$ \\
    \vspace{0.2cm}
    CNOT & $\begin{pmatrix} 1 & 0 & 0 & 0 \\[2pt]  0 & 1 & 0 & 0 \\[2pt]  0 & 0 & 0 & 1 \\[2pt]  0 & 0 & 1 & 0 \end{pmatrix}$ \\
    \vspace{0.2cm}
S & $\begin{pmatrix} 1 & 0 \\[2pt]  0 & i \end{pmatrix}$ \\
    T & $\begin{pmatrix} 1 & 0 \\[2pt]  0 & e^{\frac{i\pi}{4}} \end{pmatrix}$ \\[0.4cm]
    \hline
    \end{tabular}

\caption{Elementary Quantum Gates. Note that there is some redundancy as $\mathrm{S} = \mathrm{T}^2$ and $\mathrm{X} = \mathrm{H}\mathrm{S}^2\mathrm{H}$. The $\mathrm{R}_y$ gate can be approximated to arbitrary precision as a sequence of $\mathrm{H}$, $\mathrm{S}$, and $\mathrm{T}$ gates (see, e.g.,~\cite{ross2016optimal})}
\label{tab:gate_matrices}
\end{table}
\endgroup

Table~\ref{tab:gate_matrices} shows the elementary quantum gates required for the $\{\mathrm{U}(2), \mathrm{CNOT}\}$ gate set and the $\{\mathrm{H, S, T, CNOT}\}$ gate set. We also summarize all other required quantum gates that can be constructed in Fig.~\ref{fig:swap_decomposition}, \ref{fig:toffoli_decomposition}, \ref{fig:cswap_decomposition}, and \ref{fig:CRy_decomposition}.
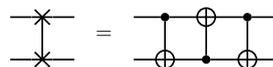
\begin{figure}[h!]
\centering
\scalebox{0.8}{
\begin{quantikz}[row sep={2em,between origins}, column sep=1em, align equals at=1.5]
    &\swap{1} &\qw  & \rb{-0.7cm}{=} & &\ctrl{1} &\targ{} &\ctrl{1} &\qw \\
    &\targX{} &\qw  & & &\targ{} &\ctrl{-1} &\targ{} &\qw
\end{quantikz}
}
\caption{SWAP Gate Decomposition Using CNOT}\label{fig:swap_decomposition}
\end{figure}


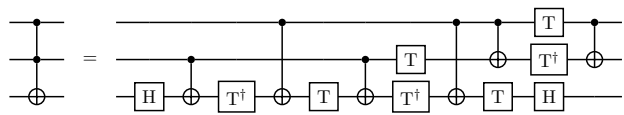
\begin{figure}[h!]
\centering
\scalebox{0.7}{
\begin{quantikz}[row sep={2em,between origins}, column sep=1em]
 &\ctrl{1} &\qw & & & \qw &\qw &\qw &\ctrl{2} &\qw &\qw &\qw &\ctrl{2} &\ctrl{1} &\gate{\mathrm{T}} &\ctrl{1} &\qw \\
 &\ctrl{1}  &\qw &= & & \qw  &\ctrl{1} &\qw &\qw &\qw &\ctrl{1} &\gate{\mathrm{T}} &\qw &\targ{} &\gate{\mathrm{T}^{\dagger}} &\targ{}& \qw \\
 &\targ{} &\qw & & &\gate{\mathrm{H}} &\targ{} &\gate{\mathrm{T}^{\dagger}} &\targ{} &\gate{\mathrm{T}} &\targ{} &\gate{\mathrm{T}^{\dagger}} &\targ{} &\gate{\mathrm{T}} &\gate{\mathrm{H}} &\qw & \qw \\
\end{quantikz}
}
\caption{Toffoli (CCNOT) Gate Decomposition Using H, T, CNOT \cite{shende2008cnot}}\label{fig:toffoli_decomposition}
\end{figure}

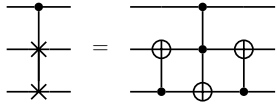
\begin{figure}[h!]
\centering
\scalebox{0.8}{
\begin{quantikz}[row sep={2em,between origins}, column sep=1em]
 &\ctrl{2} &\qw & & &\qw &\ctrl{1} &\qw &\qw \\
 &\targX{} &\qw &= & &\targ{} &\ctrl{1}& \targ{} &\qw \\
 &\swap{-1} &\qw & & &\ctrl{-1} &\targ{} &\ctrl{-1} &\qw \\
\end{quantikz}
}
\caption{Fredkin (CSWAP) Gate Decomposition Using CNOT, Toffoli \cite{smolin1996five}}\label{fig:cswap_decomposition}
\end{figure}


\begin{figure}[h!]
     \centering
     \begin{subfigure}[b]{0.5\textwidth}
     \scalebox{0.75}{
         \centering
\begin{quantikz}[row sep={2em,between origins}, column sep=1em,align equals at=1.5]
         &\ctrl{1} &\qw &\rb{-0.8cm}{=}&  & \ctrl{1} & \qw &\ctrl{1} &\qw &\qw\\
         &\gate{\mathrm{R}_y(\theta)} &\qw & &  &\targ{} &\gate{\mathrm{R}_y(-\theta/2)} &\targ{} & \gate{\mathrm{R}_y(\theta/2)} &\qw \\
\end{quantikz}
}
     \end{subfigure}
     \\
     \begin{subfigure}[b]{0.5\textwidth}
     \scalebox{0.8}{
         \centering
         \begin{quantikz}[row sep={2em,between origins}, column sep=1em]
         &\ctrl{1} &\qw & & & \ctrl{1} & \qw &\ctrl{1} &\qw &\qw\\
         &\ctrl{1} &\qw & = & & \ctrl{1} & \qw &\ctrl{1} &\qw &\qw\\
         &\gate{\mathrm{R}_y(\theta)} &\qw & & &\targ{} &\gate{\mathrm{R}_y(-\theta/2)} &\targ{} & \gate{\mathrm{R}_y(\theta/2)} &\qw \\
\end{quantikz}
}
     \end{subfigure}
        \caption{Controlled Rotation Gates (C$\mathrm{R}_y$ and CC$\mathrm{R}_y$) Decomposition Using $\mathrm{R}_y$, CNOT, Toffoli}
        \label{fig:CRy_decomposition}
\end{figure}
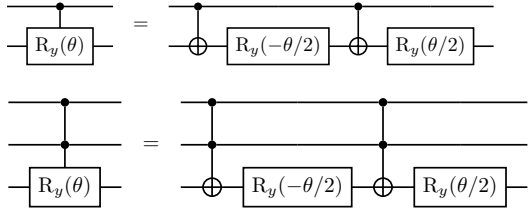

\subsection{Spacetime Allocation --- QPU ``Core Hours"}
In classical computing, ``core hours" \cite{walker2009real} refers to the number of CPUs used to run a certain computing task multiplied by the duration of the job in hours. We now define the a quantum analogue of the classical ``core hours" as the quantum spacetime allocation cost, equivalently defined in either of the following ways
\begin{itemize}
    \item Sum of the individual duration (depth) that each logical qubit is active (i.e., not in the $\ket{0}$ state)
    \item Sum of the number of active qubits in each layer
\end{itemize}
Or quantitatively: 
\begin{equation}
    \mathrm{SA} := \sum_{i=0}^{q-1} d_i = \sum_{t=0}^{d-1} q_t
\end{equation}
where $d$ is the total depth, $q$ is the total number of qubits, $d_i$ is the active time (depth) for the $i^{\mathrm{th}}$ qubit, and $q_t$ is the number of active qubits at layer $t$.

In the case that all qubits are active for the entirety of the computation, the spacetime allocation is simply the product of the circuit depth and the total number of qubits. However, if most of the ancilla qubits are needed only briefly and can be reset to the $\ket{0}$ state before the end of the computation, the spacetime allocation can be significantly less \cite{ding2020square}. In this work, we will show that freeing up the ancilla qubits at early times can bring \textbf{asymptotic} spacetime allocation advantage for quantum state preparation.

\subsection{Lower Bounds}\label{sec:LB}
Previous work has shown QSP lower bounds for circuit depth, size, number of qubits, and the corresponding tradeoffs therein \cite{sun2023asymptotically,plesch2012quantum} when using only arbitrary single-qubit gates and CNOT gates. Since each gate in the circuit acts on at least one qubit, at least one qubit must be active per gate in each layer. Thus, the number of gates provides a lower bound for the spacetime allocation. 

\begin{lemma} \label{le:single_q_SA}
    If a quantum circuit has circuit size $\mathrm{C}$, its spacetime allocation must be at least $\mathrm{C}$.   
\end{lemma}

Consequently, if we wish to lower bound the spacetime allocation, it suffices to bound the circuit size. An arbitrary quantum state is described by $O(2^n)$ real parameters. Each single-qubit gate can introduce at most $O(1)$ real parameters, so there must be $\Omega(2^n)$ single-qubit gates. Additionally, any consecutive single-qubit gates on the same qubit that do not have a CNOT in between can be combined into a single gate, so the number of CNOTs must be in the same order as the number of single-qubit gates. This line of reasoning gives rise to a lower bound on circuit size \cite{plesch2012quantum}. 
\begin{lemma}[Section II of \cite{plesch2012quantum}] \label{le:NT_gate_LB}
    A quantum circuit consisting of gates drawn from $\{\mathrm{U(2), CNOT}\}$ that prepares an arbitrary $n$-qubit quantum state must have at least  $\Omega(2^n)$ $\mathrm{CNOT}$ gates and at least $\Omega(2^n)$ $\mathrm{U}(2)$ gates.
\end{lemma}

\begin{corollary}
    A quantum circuit consisting of gates drawn from $\{\mathrm{U(2), CNOT}\}$ that prepares an arbitrary $n$-qubit quantum state must have at least  $\Omega(2^n)$ spacetime allocation.
\label{le:N_spacetime}
\end{corollary}

Similar lower bounds are possible in the $\{\mathrm{H,S,T,CNOT}\}$ gate set. 

\begin{lemma}\label{le:FT_gate_LB}
    A quantum circuit consisting of gates drawn from $\{\mathrm{H, S, T, CNOT}\}$ that prepares an arbitrary $n$-qubit state to error $\epsilon$ must have at least
    $\Omega\left(\frac{2^n\log(1/\epsilon)}{n+\log(\log(1/\epsilon))}\right)$ size.
\end{lemma}
\begin{proof}
    We follow the logic similar to Theorem 9 of \cite{zhang2023circuit}. Following Lemma 1 in \cite{zhang2023circuit}, the circuit size $C$ must satisfy $C \log(C) \geq \Omega(2^n\log(1/\epsilon))$. Letting $U = C \log(C)$ and taking the logarithm on both sides, we have $\log(U) = \log(C\log(C)) \geq \log(C)$, and hence $C = U/\log(C) \geq U / \log(U)$. Thus, $U \geq \Omega(2^n\log(1/\epsilon))$ implies $C \geq \Omega(2^n\log(1/\epsilon)/(n+\log(\log(1/\epsilon))))$.
\end{proof}


Note that Ref.~\cite{low2018trading} has shown that the number of $\mathrm{T}$ gates can be significantly smaller than this lower bound, an interesting fact since $\mathrm{T}$ gates are significantly more expensive than the other gates in many approaches to fault-tolerant quantum computation. 
\begin{corollary}\label{cor:SA_LB_HSTCNOT}
        A quantum circuit consisting of gates drawn from $\{\mathrm{H, S, T, CNOT}\}$ for preparing an arbitrary $n$-qubit state to error $\epsilon$ must have at least $\Omega\left(\frac{2^n\log(1/\epsilon)}{n+\log(\log(1/\epsilon))}\right)$ spacetime allocation.
\end{corollary}
It is unknown whether the lower bound in Lemma \ref{le:FT_gate_LB} and Corollary \ref{cor:SA_LB_HSTCNOT} is tight. 
\\\\
Lower bounds on depth can also be shown in both gate sets. 

\begin{lemma}[Theorem 3 of \cite{sun2023asymptotically}]\label{lem:LB_depth_U2CNOT}
A quantum circuit consisting of gates drawn from $\{U(2), \mathrm{CNOT}\}$ that prepares an arbitrary $n$-qubit state using $q$ ancilla qubits must have depth at least
   $\Omega(\max(n, \frac{2^n}{q+n}))$.
\end{lemma}

\begin{lemma}\label{lem:LB_depth_discrete}
A quantum circuit consisting of gates drawn from $\{\mathrm{H},\mathrm{S},\mathrm{T}, \mathrm{CNOT}\}$ that prepares an arbitrary $n$-qubit state using $O(2^n)$ ancilla qubits must have depth at least $\Omega (n + \log(1/\epsilon))$.
\end{lemma}
\begin{proof}
First we examine the $n$ dependence for any constant $\epsilon < 1$. Following a similar argument as \cite[Theorem 3]{sun2023asymptotically}, at depth $D$, the number of gates that are in the ``lightcone'' of the $n$ data qubits is upper bounded by $n2^D$. The rest of the gates can be ignored. Since there are a constant number of gates in the gate set, the total number of unique output states is upper bounded by $e^{O(n2^{D})}$. Meanwhile, Eq.~(4.85) of Ref.~\cite{nielsen2002quantum} asserts that the number of circuits needed to cover the entire set of $n$-qubit states up to $\epsilon=O(1)$ precision is at least $e^{\Omega(2^{n})}$. Together, these imply that $D = \Omega(n)$. Separately, we can lower bound the $\epsilon$ dependence as $D = \Omega(\log(1/\epsilon))$ directly from Theorem 9 of \cite{zhang2023circuit}, once we set $n_{anc} = O(2^n)$. These two bounds hold independently for sufficiently large $n$ and sufficiently small $\epsilon$, thus the sum of the bounds must also hold (up to a constant factor that can be absorbed into big-$\Omega$), implying the overall stated bound of $\Omega(n + \log(1/\epsilon))$.
\end{proof}

The SP+CSP circuit construction illustrated in Sec.~\ref{sec:SP+CSP} will give an upper bound that matches the lower bound of Lemma \ref{lem:LB_depth_discrete}.

\subsection{Upper Bounds in Prior Work}
\subsubsection{\texorpdfstring{$\{\mathrm{U(2), CNOT}\}$}{TEXT} gate set}

To the best of our knowledge, the state-of-the-art quantum state preparation \cite{sun2023asymptotically, yuan2023optimal} methods achieve the following depth, where $a$ is the number of ancillae available:

\begin{equation}
\nonumber
\mathrm{D_{exact}} =\begin{cases}
    \Theta(\frac{2^n}{n}), & \text{if $a = 0 $ (need }n \text{ data qubits) } \\
    \Theta(\frac{2^n}{n+a}), & \text{if $a = O(2^n/n)$}\\
    \Theta(n), & \text{if $a = \Omega(\frac{2^n}{n})$}
    \end{cases}
\end{equation}

Here $\mathrm{D_{exact}}$ labels the depth of the exact quantum state preparation circuit using the $\{$U(2), CNOT$\}$ gate set, $a$ denotes the number of ancillae, and $n$ denotes the number of qubits of the desired arbitrary quantum state to be prepared.

We can compute the spacetime allocation upper bound by simply multiplying the depth with the number of qubits required; we find that the lower bound of $\Omega(2^n)$ is saturated for any choice of ancilla qubits $a < O(2^n/n)$, and in particular:

\begin{equation}
\nonumber
\mathrm{SA_{exact}} =\begin{cases}
    \Theta(2^n), & \text{if $\mathrm{D_{exact}} =\Theta(\frac{2^n}{n})$ } \\
    \Theta(2^n), & \text{if $\mathrm{D_{exact}} =\Theta(n)$}
    \end{cases}
\end{equation}

Interestingly, Sun et al.~\cite{sun2023asymptotically} also proved a depth lower bound of $\Omega(n)$ for circuits that use arbitrary single-qubit and two-qubit gates, regardless of the number of ancillary qubits from a graph theory perspective. In Sec.~\ref{sec:SP+CSP}, we show that our SP+CSP state preparation protocol can also achieve $\Theta(n)$ depth and $\Theta(2^n)$ spacetime allocation. This provides an alternative construction to that of Ref.~\cite{yuan2023optimal}, which is also optimal.

\subsubsection{\texorpdfstring{$\{\mathrm{H, S, T, CNOT}\}$}{TEXT} gate set}

When compiled into gates from $\{\mathrm{H, S, T, CNOT}\}$ previous methods have achieved \cite{sun2023asymptotically, zhang2022quantum} the following depth, where $a$ is the number of ancillae available:

\begin{equation}
\nonumber
\mathrm{D_{approx}} =\begin{cases}
    O(\frac{2^n\log(2^n/\epsilon)}{n}), & \text{if $a = 0 $ } \\
    O(n\log(n)\log(2^n/\epsilon)), & \text{if $a = \Theta(\frac{2^n}{n\log(n)})$ }\\
    O(n\log(n/\epsilon)), & \text{if $a = \Omega(2^n)$}
    \end{cases}
\end{equation}

We can also compute the spacetime allocation cost by multiplying the circuit depth by the number of qubits involved:

\begin{equation}
\nonumber
\mathrm{SA_{approx}} =\begin{cases}
    O(2^n\log(2^n/\epsilon)), & \text{if $a = 0 $ } \\
    O(2^n\log(2^n/\epsilon)), & \text{if $a = \Theta(\frac{2^n}{n\log(n)})$ }\\
    O(n2^n\log(n/\epsilon)), & \text{if $a = \Omega(2^n)$}
    \end{cases}
\end{equation}


We also note that the method proposed by Yuan et al. \cite{yuan2023optimal} will have at least $\Omega(n\log(n/\epsilon))$ depth and $\Omega(2^n\log(n/\epsilon))$ spacetime allocation. This was not stated explicitly in the paper but can be deduced from the fact that single-qubit rotations appear in $O(n)$ different layers and each requires depth at least $\Omega(\log(n/\epsilon))$ in the discrete gate set.

In addition, Clader et al.~\cite{clader2023quantum} gave a method with depth $O(\log(2^n/\epsilon))$ and spacetime allocation $O(2^n\log(2^n/\epsilon))$, but under the assumption of unit time FANOUT-CNOT. Decomposing FANOUT-CNOT with $2^n$ targets into two-qubit CNOT gates incurs a multiplicative overhead of $O(n)$ in the depth of the protocol, which also induces a similar overhead to the spacetime allocation (see Table \ref{table:QRE_comps}).

In this work, we will show that our method achieves the depth scaling of Ref.~\cite{clader2023quantum} in the $\{\mathrm{H, S, T, CNOT}\}$ gate set (i.e.~not requiring FANOUT-CNOT) while achieving superior spacetime allocation of $O(2^n\log (n/\epsilon))$.
\section{SP+CSP State Preparation}\label{sec:SP+CSP}
In this section, we will walk through the details of our SP+CSP method that achieves $\Theta(N)$ spacetime allocation while keeping the $\Theta(\log(N))$ depth using the $\{$U(2), CNOT$\}$ gate set.

Recall that we wish to create the $n$-qubit state $\ket{\psi}$ in Eq.~\eqref{Eq:quantum_state}, which has known coefficients $x_i/\lVert \textbf{x} \rVert$ in the computational basis. We propose to use the following method that first uses a \textit{state preparation} step (\textbf{SP}) on a subset of the data qubits, followed by a \textit{controlled state preparation} step (\textbf{CSP}). The rationale for doing this rather than a direct SP is that the SP+CSP protocol can harness both the advantages of SP and CSP steps while avoiding their disadvantages if set up correctly. We explain the details of the complexity advantages and disadvantages later in this section at Sec.~\ref{sec:overall_depth_and_SA}.

The general idea of the SP+CSP protocol is to partition the $N = 2^n$ basis states into $M=2^m$ non-overlapping sets of size $\frac{N}{M}$. 
Computational basis states for which the first $m$ bits agree (when written in binary) are placed into the same set. Denote the $M$ sets by $J_i$ with $i=0,\ldots, M-1$. For each $i$, we let $y_i = \sqrt{\sum_{j \in J_i} |x_j|^2}$.

We first prepare the state of $m$ qubits ($m < n$):
\begin{equation}\label{eq:phi_def}
    U_{\mathrm{SP}}\ket{0^m} = \ket{\phi} = \frac{1}{\lVert \mathbf{y} \rVert}\sum_{i=0}^{2^m-1} y_i \ket{i}
\end{equation}
The total state is now $\ket{\phi} \otimes \ket{0^{n-m}}$.

Next, we perform the controlled state preparation operation $U_{\mathrm{CSP}}$, defined by
\begin{equation}
U_{\mathrm{CSP}}(\ket{i} \otimes \ket{0^{n-m}}) = y_i^{-1}\sum_{j \in J_i} x_j \ket{j}
\end{equation}
for a particular $\ket{i}$ state being controlled on, so that 
\begin{equation}
    U_{\mathrm{CSP}}(\ket{\phi} \otimes \ket{0^{n-m}}) = \ket{\psi} = \frac{1}{\lVert \mathbf{x}\rVert}\sum_{j=0}^{2^n - 1}x_j \ket{j}
\end{equation}
The circuit diagram is shown in Fig.~\ref{fig:structure}. Implementation details of $U_{\mathrm{SP}}$ and $U_{\mathrm{CSP}}$ are shown in Fig.~\ref{fig:SP} and Fig.~\ref{fig:CSP}. 
\begin{figure}[h!]
\centering
\begin{quantikz}[row sep={2em,between origins}, column sep=1em]
\lstick{$\ket{0^m}$}     & \gate{U_{\mathrm{SP}}} & \ctrlslash{1} & \qw \rstick[wires=2]{$\ket{\psi}$} \\
\lstick{$\ket{0^{n-m}}$} & \qw      & \gate{U_{\mathrm{CSP}}}      & \qw \\
\end{quantikz}
\caption{Circuit for the SP+CSP protocol that implements state preparation of an arbitrary $n$-qubit state $\ket{\psi}$ by first creating an $m$-qubit state $\ket{\phi}$ and then performing controlled state preparation into the final $n-m$ qubits. The square control $\text{\Squarepipe}$ indicates a different operation is performed on the final $n-m$ qubits for each setting of the first $m$ qubits.}
\label{fig:structure}
\end{figure}
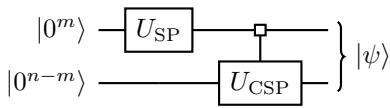

We will now describe how to perform the $U_{\mathrm{SP}}$ and $U_{\mathrm{CSP}}$ circuits.

\subsection{SP Circuit Structure}

To perform $U_{\mathrm{SP}}$, we give a method that has depth $O(m)$ and uses $O(2^m)$ ancilla qubits. The spacetime allocation is thus at most $O(m2^m)$. If we require a discrete gate set and have approximation error $\epsilon$, this method achieves $O(m+\log(1/\epsilon))$ depth and spacetime allocation $O(m2^m + 2^m\log(m/\epsilon))$.

The idea behind the SP protocol follows previous literature and begins by defining $M-1$ angles that can each be efficiently computed classically from the list of amplitudes $\{y_i\}_{i=0}^{M-1}$. For convenience and consistency with previous literature \cite{schuld2021machine}, we use a 2-index angle definition to define the rotation angles:
\begin{equation}\label{eq:theta_r}
    \theta_{s,p} = 2\cos^{-1} \left(\frac{\sqrt{\sum_{l=0}^{2^{m-s-1}-1}|y_{p\cdot 2^{m-s}+l}|^2}}{\sqrt{\sum_{l=0}^{2^{m-s}-1}|y_{p\cdot 2^{m-s}+l}|^2}}\right)\,,
\end{equation}
where $s=0,\ldots,m-1$, and $p = 0,\ldots, 2^{s}-1$, for a total of $1+2+4+\ldots+2^{m-1} = M-1$ angles. Naively, computing each angle can require querying as many as $O(2^m)$ entries of the vector $\mathbf{y}$, resulting in $O(2^{2m})$ total classical runtime. However, we can reuse some of the computations by storing the quantities $S_{s,p} = \sum_{l=0}^{2^{m-s}}|y_{p \cdot 2^{m-s}+l}|^2$ in a binary-tree data structure with $m$ levels \cite{kerenidis2017quantum}. The tree has $2^{m-1}$ leaves; we store the quantities $S_{m-1,p}$ in the leaves for $p=0,\ldots, 2^{m-1}-1$. The tree is then constructed recursively by the rule that a parent stores the sum of its children. Since $S_{s,p} = S_{s+1,2p} + S_{s+1,2p+1}$, we can verify that the value of the $p^{\mathrm{th}}$ node at level $s$ will store the value $S_{s,p}$. The angle $\theta_{s,p}$ can then be determined only from the values stored at this node along with its two children. The total work to construct the entire tree is $O(2^m)$, and the tree has the added benefit that if an individual entry of $\mathbf{y}$ is modified or if the entries arrive in an online fashion, updating the tree data structure can be classically done in time $O(m)$ by following a single path from a leaf back to the root. Additionally, although we do not utilize this property, in the case that $\mathbf{y}$ is sparse and has only $d$ nonzero entries,  the tree structure will also be sparse and have at most $dm$ nonzero entries.

This 2-index angle definition will also give us more convenience when labeling the ancilla registers holding the computed angles shown in Eq.~\eqref{eq:ket_theta_sp} and the corresponding data qubits. In particular, $s$ corresponds to the index of the ancilla register $A_s$ (referenced in, e.g., Algorithm \ref{algo:SPF}) as well as the index of the data qubit being processed. Meanwhile, $p$ corresponds to the index of the qubit within the particular ancilla register $A_s$.

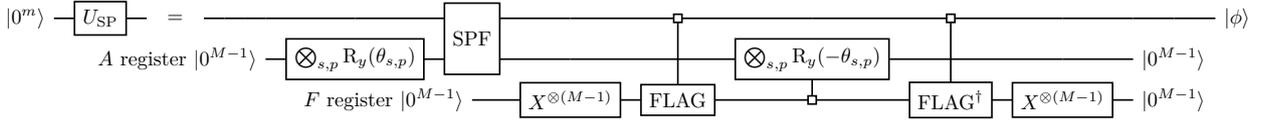
\begin{figure*}[t]
\makebox[\textwidth][c]{
\scalebox{0.77}{
\begin{quantikz}[row sep={2em,between origins}, column sep=1em]
\lstick{$\ket{0^m}$} & \gate{U_\mathrm{SP}} & \qw &=& &\qw & \qw &\qw &\qw                                         &\mltg{2}{\text{SPF}} & \qw & \ctrlslash{2}      & \qw            & \ctrlslash{2}              & \qw   & \qw & \qw &\qw & \qw  & \rstick{$\ket{\phi}$} \qw \\
  &  & & & & && \lstick{$A$ register $\ket{0^{M-1}}$} &\gate{\bigotimes_{s,p} \mathrm{R}_y(\theta_{s,p})} &    & \qw & \qw                & \gate{\bigotimes_{s,p} \mathrm{R}_y(-\theta_{s,p})}       & \qw                        & \qw      & \rstick{$\ket{0^{M-1}}$}    \qw \\
 & &  & &&&                                     &&&\lstick{$F$ register $\ket{0^{M-1}}$}                & \gate{X^{\otimes(M-1)}} & \gate{\text{FLAG}} & \ctrlslash{-1} & \gate{\text{FLAG}^\dagger} & \gate{X^{\otimes(M-1)}} & \rstick{$\ket{0^{M-1}}$}    \qw \\
\end{quantikz}
}
}
\caption{Circuit that implements $U_{\mathrm{SP}}$, which prepares an arbitrary $m$-qubit state from the state $\ket{0^m}$ with the assistance of $O(2^m)$ ancillae that begin and end in $\ket{0}$. We let $M=2^m$ and clarify that the controlled rotation gate denotes $M-1$ controlled rotations occurring in parallel, by different angles $\theta_{s,p}$ with $s=0,\ldots,m-1$ and $p=0,\ldots,2^s-1$. Note that the implementation of SPF and FLAG, given in the appendix, involves $O(M)$ additional unshown ancilla qubits. }\label{fig:SP}
\end{figure*}

Araujo et al.~\cite{araujo2021divide} proposed a low-depth strategy for implementing $U_{\mathrm{SP}}$ using $O(M)$ ancilla qubits. The general idea is to ``pre-rotate'' many angles by preparing the states
\begin{equation}\label{eq:ket_theta_sp}
\ket{\theta_{s,p}} = \mathrm{R}_y(\theta_{s,p})\ket{0} = \cos\left(\frac{\theta_{s,p}}{2}\right)\ket{0} + \sin\left(\frac{\theta_{s,p}}{2}\right)\ket{1}
\end{equation}
in parallel for each $s=0,\ldots,m-1$, $p = 0,\ldots, 2^s-1$, and then efficiently inject a subset of these states into the data qubits. Which states are injected is controlled by the data qubits themselves, leaving the data qubits entangled with a large garbage register. Clader et al.~\cite{clader2023quantum} showed how to uncompute the garbage using a flag mechanism. It was shown that the overall $\mathrm{T}$ depth of the injection steps (e.g., SPF and FLAG) was $O(m)$, but the Clifford depth was not studied, and a naive calculation would suggest the Clifford depth is at least $\Omega(m^2)$ when we insist on using only two-qubit Clifford gates, due to the existence of the FANOUT-CNOT gate. In this work, we give an adapted version of Ref.~\cite{clader2023quantum} by utilizing a bit-copying mechanism to guarantee that the Clifford depth is also $O(m)$. 

The high-level circuit that accomplishes the $U_{\mathrm{SP}}$ routine is shown in Fig.~\ref{fig:SP}. The circuit describes the gate sequence that prepares an arbitrary $m$-qubit state $\ket{\phi}$ with the assistance of $2M-2 = 2^{m+1}-2$ ancilla qubits that begin and end\footnote{\label{footnote:end_not_0}Note that the ancillas in the $A$ register of Fig.~\ref{fig:SP} are only guaranteed to end in $\ket{0}$ when the input to $U_{\mathrm{SP}}$ is $\ket{0^m}$. On other input states, this register may end up entangled with the data qubits. See App.~\ref{sec:input_not_0} for a discussion on why this feature does not ruin the optimality of our protocol in applications where the input is not always $\ket{0^m}$.} in $\ket{0}$. Note that implementations of the SPF circuit and FLAG circuit, discussed in App.~\ref{sec:appendix}, require an additional $O(M)$ ancillae not shown in the figure to implement the bit copying mechanism mentioned above.

We now follow Ref.~\cite{clader2023quantum} and define the action of circuits SPF and FLAG that appear in Fig.~\ref{fig:SP}. First, define the product states
\begin{equation}\label{eq:ket_Theta}
\begin{split}
    \ket{\Theta_s} &= \bigotimes_{p=0}^{2^{s}-1} \ket{\theta_{s,p}} \qquad \text{for } s = 0,1,\ldots,m-1 \\
    \ket{\Theta} &= \bigotimes_{s=0}^{m-1}\ket{\Theta_s}\,.
\end{split}
\end{equation}
Next, for each $j = 0, \ldots, M-1$, $s=0,\ldots,m-1$, $p=0,\ldots,2^s-1$, define
\begin{equation}
    f_{(s,p)|j} = \begin{cases}
        1 & \text{if } p = j \bmod 2^s \\
        0 & \text{otherwise}
    \end{cases}\,.
\end{equation}
Thus, if we fix a value of $j$, there are $m$ pairs $(s,p)$ for which $f_{(s,p)|j}=1$: in fact, for each value of $s$, there is exactly one $p \in [0, 2^{s}-1]$ for which $f_{(s,p)|j} = 1$. The values of $(s,p)$ for which $f_{(s,p)|j} =1$ correspond to the angle states $\ket{\theta_{s,p}}$ that are injected into the data by the SPF circuit for the data qubit setting $\ket{j}$. 

The $U_{\mathrm{SP}}$ procedure begins by preparing $\ket{\Theta}$ into an ancilla ``angle'' register (register $A$ in Fig.~\ref{fig:SP}) of size $M-1$ using $M-1$ parallel $\mathrm{R}_y$ gates. The SPF circuit is then applied, which injects some of the angle states into the data and produces the action
\begin{align}
\mathrm{SPF}\left(\ket{0^m}\ket{\Theta}\right) = \frac{1}{\lVert \mathbf{y} \rVert}\sum_{j=0}^{M-1} y_j\ket{j} \ket{g_j}\,.
\end{align}
where
\begin{equation}
\ket{g_j} = \bigotimes_{s=0}^{m-1}\bigotimes_{p=0}^{2^s-1} \ket{\theta_{s,p}(1- f_{(s,p)|j})}\,.
\end{equation}
Here, the notation $\ket{\theta(1-f)}$ is used to denote that the state is $\ket{\theta}$ when $f=0$ and $\ket{0}$ when $f=1$. In other words, the registers holding $\ket{\theta_{s,p}}$ for which $f_{(s,p)|j} = 1$ are replaced with $\ket{0}$. We see that the output of SPF has the correct amplitudes as the target state $\ket{\phi}$, but leaves the data entangled with an $(M-1)$-qubit garbage register.

The FLAG operation computes the bit $f_{(s,p)|j}$ into a fresh ancilla ``flag'' register\footnote{The need for a flag mechanism following the SPF circuit is the origin of the letter F in the name of the SPF circuit (SP stand for ``state preparation'') and also in the name of the LOADF circuit (presented in the next section).} (register $F$ in Fig.~\ref{fig:SP}) for each pair $(s,p)$, i.e.,
\begin{equation}
\text{FLAG}(\ket{j} \ket{1^{M-1}}) = \ket{j}\bigotimes_{s=0}^{m-1}\bigotimes_{p=0}^{2^s-1}\ket{1-f_{(s,p)|j}}
\end{equation}
while leaving the garbage states on the angle qubits untouched. Thus, by applying a controlled-$\mathrm{R}_y(-\theta_{s,p})$ gate controlled by the flag bit $(1-f_{(s,p)|j})$ onto the angle qubit in the state $\ket{\theta_{s,p}(1-f_{(s,p)|j})}$ after the FLAG, we reset all of the angle qubits back to $\ket{0}$. This is done in parallel for each pair $(s,p)$. Finally, the adjoint of the FLAG operation (FLAG$^{\dagger}$) can be applied to uncompute the flag bits and bring the flag register back to $\ket{0^{M-1}}$. The output of this process is the state $\ket{\phi}$ without any garbage. 

The construction for SPF and FLAG described in Ref.~\cite{clader2023quantum} was only optimized for the $\mathrm{T}$-depth and $\mathrm{T}$-count. The Clifford count and depth were not explicitly studied. Moreover, it was assumed that a FANOUT-CNOT gate with an arbitrary number of targets could be performed in a single time step. In the appendix, we detail our constructions for the SPF (Sec.~\ref{sec:SPF_circ}) and FLAG (Sec.~\ref{sec:FLAG_circ}) subroutines. There, we illustrate how both subroutines can be performed in $O(m)$ depth and $O(m2^m)$ spacetime allocation using $\{\mathrm{U(2)}, \mathrm{CNOT}\}$ gate set and near-optimal costs for $\{\mathrm{H, S, T, CNOT}\}$ gate set.

We document the pseudocode of the SP circuit implementation below in Algorithm \ref{algo:sp_circuit}. 

\begin{algorithm}[H]
    \caption{SP Circuit (see Fig.~\ref{fig:SP})}\label{algo:sp_circuit}
    \begin{algorithmic}[1]
    \Procedure{SP}{$m$, $D$, $A$, $F$}

    \Comment{$D$: data register of size $m$}
    
    \Comment{$A$: angle register of size $M - 1$}
    
    \Comment{$F$: flag register of size $M - 1$}
        \For{$s$ \textbf{in} range($m$) \textbf{and} $p$ \textbf{in} range($2^s$)}
            \State Classically compute $\theta_{s,p}$ from $\mathbf{y}$
        \Comment{Eq.~\eqref{eq:theta_r}}
        \EndFor
        \For{$s$ \textbf{in} range($m$) \textbf{and} $p$ \textbf{in} range($2^s$)}
            \State $\mathrm{R}_y(\theta_{s,p}, A_{s,p})$ onto fresh ``angle'' ancilla
        \EndFor
        \State SPF($D$, $A$, $m$) subroutine \Comment{$O(m)$}
        \For{$s$ \textbf{in} range($m$) \textbf{and} $p$ \textbf{in} range($2^s$)}
            \State $\mathrm{X}(F_{s,p})$ onto fresh ``flag'' ancilla
        \EndFor
        \State FLAG$(D, F, m)$ subroutine \Comment{$O(m)$}
        \For{$s$ \textbf{in} range($m$) \textbf{and} $p$ \textbf{in} range($2^s$)}
        \State $\mathrm{CR}_y(-\theta_{s,p}, F_{s,p}, A_{s,p})$
        \EndFor 
        \State FLAG$^\dagger(D, F, m)$ subroutine \Comment{$O(m)$}
       \For{$s$ \textbf{in} range($m$) \textbf{and} $p$ \textbf{in} range($2^s$)}
            \State $\mathrm{X}(F_{s,p})$ to reset ``flag'' ancilla to $\ket{0}$
        \EndFor
    \EndProcedure

    \end{algorithmic}
    \Comment{Total D$_{\mathrm{exact}}$: $O(m)$}

    \Comment{Total D$_{\mathrm{approx}}$: $O(m + \log(m/\epsilon))$}
    
    \Comment{Total SA$_{\mathrm{exact}}$: $O(m2^m)$}

    \Comment{Total SA$_{\mathrm{approx}}$: $O(m2^m + 2^m\log(m/\epsilon))$}
\end{algorithm}



\subsection{CSP Circuit Structure}

\begin{figure*}[t]
\makebox[\textwidth][c]{
\scalebox{0.95}{
\begin{quantikz}[row sep={2.2em,between origins}, column sep=1em, align equals at=1.5]
\lstick{$\ket{\phi}$} &\ctrlslash{1}  & \qw &\rb{-1.0cm}{=} &  &\qw &\qw &\qw       &\ctrlslash{2}                                            &
 \qw                  & \qw                & \ctrlslash{2}                                                    & \qw                        & \qw      &     \qw & \qw &\qw&\qw &\qw& \qw\rstick[wires=2]{$\ket{\psi}$} \\
\lstick{$\ket{0^{n-m}}$} &\gate{U_{\mathrm{CSP}}} & \qw && &\qw &\qw &\qw      &\qw                                                        &
 \mltg{2}{\text{SPF}} & \ctrlslash{2}      & \qw                                                               & \ctrlslash{2}              & \qw      & \qw & \qw &\qw&\qw &\qw& \qw  \\
&&& &&&\lstick{$B$ register $\ket{0^{\frac{N}{M}-1}}$} &\qw      &\gate{\substack{\text{LOADF}\\\\\{\theta^{(k)}\}}} &
    & \qw                & \gate{\substack{\text{LOADF}^\dagger\\\\\{\theta^{(k)}\}}} & \qw                        & \qw      & \rstick{$\ket{0^{\frac{N}{M}-1}}$}      \qw \\
&&& &&&\lstick{$F$ register $\ket{0^{\frac{N}{M}-1}}$} &\gate{X^{\otimes(\frac{N}{M}-1)}} &\ctrlslash{-1}                                            &
 \qw                  & \gate{\text{FLAG}} & \ctrlslash{-1}                                                    & \gate{\text{FLAG}^\dagger} & \gate{X^{\otimes(\frac{N}{M}-1)}} & \rstick{$\ket{0^{\frac{N}{M}-1}}$}      \qw \\
\end{quantikz}
}
}
\caption{Circuit that implements the $U_{\mathrm{CSP}}$ operation, which prepares an arbitrary $n-m$ qubit state for each setting of an $m$-qubit control register, with the assistance of $O(2^{n-m})$ ancillas that begin and end in $\ket{0}$. We let $M=2^m$ and $N=2^n$. Note that implementations of LOADF, SPF, and FLAG, given in the appendix, involve $O(N)$ additional unshown ancillae. We also note that in the actual implementation, one can reduce the first LOADF operation so that it is controlled only by the B register (not on the F register, which is guaranteed to be in the $\ket{1^{N/M-1}}$ state at this stage in the circuit) in order to save a constant depth. We choose to define the LOADF operator in this way only to avoid another definition of the later LOADF$^{\dagger}$.}
\label{fig:CSP}
\end{figure*}
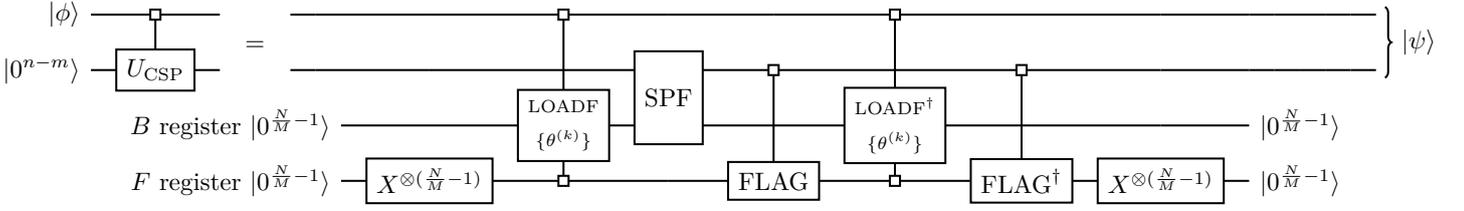

The controlled state preparation circuit prepares a different $n-m$ qubit state for each of $M$ possible settings $\ket{k} \in \{\ket{0},\ldots,\ket{M-1}\}$ of an $m$-qubit control register. Thus, for each $k$, there are $2^{n-m} -1 = \frac{N}{M} -1$ angles, which we denote by $\theta^{(k)}_{s,p}$, for $s = 0,\ldots, n-m-1$ and $p = 0,\ldots, 2^{s}-1$. These angles can be computed from the amplitudes of $\mathbf{x}$ by the equation (c.f.~Eq.~\eqref{eq:theta_r})
\begin{equation}\label{eq:theta_ksp}
    \theta^{(k)}_{s,p} = 2\cos^{-1} \left(\frac{\sqrt{\sum_{l=0}^{2^{n-m-s-1}-1}|x_{k2^{n-m}+p\cdot 2^{n-m-s}+l}|^2}}{\sqrt{\sum_{l=0}^{2^{n-m-s}-1}|x_{k2^{n-m}+p\cdot 2^{n-m-s}+l}|^2}}\right)\,.
\end{equation}
The total number of angles is then $N-M$. The states $\ket{\theta_{s,p}^{(k)}}$ are then defined as in Eq.~\eqref{eq:ket_theta_sp}, and the product states $\ket{\Theta^{(k)}_s}$ and $\ket{\Theta^{(k)}}$ as in Eq.~\eqref{eq:ket_Theta}. We document the pseudocode of the CSP circuit implementation in Algorithm \ref{algo:csp_circuit}.

\begin{algorithm}[H]
    \caption{CSP Circuit (see Fig.~\ref{fig:CSP})}\label{algo:csp_circuit}
    \begin{algorithmic}[1]
    \Procedure{CSP}{$m$, $n$, $D$, $B$, $F$}

    \Comment{$D$: data register of size $m - n$}
    
    \Comment{$B$: buffer register of size $\frac{N}{M} - 1$}
    
    \Comment{$F$: flag register of size $\frac{N}{M} - 1$}
    
        \For{$k$ \textbf{in} range($2^m$)}
            \For{$s$ \textbf{in} range($n - m$) \textbf{and} $p$ \textbf{in} range($2^s$)}
            \State Classically compute $\theta_{s,p}^{(k)}$ from $\mathbf{x}$ \Comment{Use Eq.~\eqref{eq:theta_ksp}}
            \EndFor
        \EndFor
        \For{$s$ \textbf{in} range($n - m$) \textbf{and} $p$ \textbf{in} range($2^s$)}
            \State $\mathrm{X}(F_{s,p})$ onto fresh ``flag'' ancilla
        \EndFor
        \State LOADF$(D, B, F)[\theta^{(k)}_{s,p}]$ load angles into buffer ancillae \Comment{$O(n)$}
        \State SPF$(D, B, n-m)$ to prepare $n-m$ qubit state with garbage \Comment{$O(n-m)$}
        \State FLAG$(D, F, n-m)$ \Comment{$O(n-m)$}
        \State $\mathrm{LOADF}^{\dagger}$$(D, B, F)[\theta^{(k)}_{s,p}]$ to unload angles in buffer ancillae \Comment{$O(n)$}
        \State $\mathrm{FLAG}^{\dagger}(D, F, n-m)$ \Comment{$O(n-m)$}
        \For{$s$ \textbf{in} range($n - m$) \textbf{and} $p$ \textbf{in} range($2^s$)}
            \State $\mathrm{X}(F_{s,p})$ to reset ``flag'' ancilla
        \EndFor
    \EndProcedure
    \end{algorithmic}
    \Comment{Total D$_{\mathrm{exact}}$: $O(n)$}

    \Comment{Total D$_{\mathrm{approx}}$: $O(n + \log(n/\epsilon))$}
    
    \Comment{Total SA$_{\mathrm{exact}}$: $O(2^n)$}

    \Comment{Total SA$_{\mathrm{approx}}$: $O((n-m)2^{n-m} + 2^n\log(n/\epsilon))$}
\end{algorithm}

The controlled state preparation circuit, shown in Fig.~\ref{fig:CSP}, is a generalized version of the controlled state preparation circuit proposed in Clader et al.~\cite{clader2023quantum}. The general idea of the CSP circuit is to first, controlled on the control register being $\ket{k}$, load in the correct state $\ket{\Theta^{(k)}}$. Once this has been done, the same SPF circuit as was used for the $U_{\mathrm{SP}}$ protocol is applied to inject the correct angles into the $n-m$ data qubits. A FLAG mechanism similar to that of the $U_{\mathrm{SP}}$ protocol is then employed to disentangle the data from the angle register and reset\footnote{\label{footnote:letter_f_definition}Note that the angle register ($B$ register from Fig.~\ref{fig:CSP}) is only guaranteed to end in $\ket{0^{M-1}}$ when the input to $U_{\mathrm{CSP}}$ is $\ket{0^{n-m}}$. See Footnote \ref{footnote:end_not_0} and App.~\ref{sec:input_not_0}.} the angle register to $\ket{0^{M-1}}$. 


The loading and unloading is accomplished with a circuit we call LOADF, similar to (but not the same as) the LOADF in Ref.~\cite{clader2023quantum}, which is defined by the action
\begin{equation}\label{eq:LOADF}
\begin{split}
    &\mathrm{LOADF}\left(\ket{k}\ket{0^{N/M-1}}\left[\bigotimes_{s=0}^{m-n-1}\bigotimes_{p=0}^{2^s-1} \ket{f_{s,p}}\right]\right) \\
    ={}& \ket{k}\left[\bigotimes_{s=0}^{m-n-1}\bigotimes_{p=0}^{2^s-1} \ket{f_{s,p}\theta^{(k)}_{s,p}}\right]\left[\bigotimes_{s=0}^{m-n-1}\bigotimes_{p=0}^{2^s-1} \ket{f_{s,p}}\right]
\end{split}
\end{equation}
Note that if all the flag bits $f_{s,p}$ are set to 1, the state $\ket{\Theta^{(k)}}$ is prepared into the second register. 
Our implementation of LOADF is shown in Fig.~\ref{fig:LOADF} in the appendix, and involves an additional $O(2^n)$ ancilla qubits.


\subsection{Overall depth and spacetime allocation}\label{sec:overall_depth_and_SA}

\begingroup
\begin{table*}[t]
\setlength{\tabcolsep}{2pt} 
\renewcommand{\arraystretch}{1.4} 
\makebox[\textwidth][c]{
\centering
\begin{adjustbox}{max width=1.0\textwidth}
\newcolumntype{C}[1]{ >{\centering\arraybackslash} m{#1} }
  \begin{tabular}{C{100pt}|C{40pt}|C{90pt}|C{45pt}|C{150pt}|C{200pt}} 
 & $\mathrm{D_{exact}}$ & $\mathrm{D_{approx}}$ & \# Qubits & $\mathrm{SA_{exact}}$ & $\mathrm{SA_{approx}}$ \\ [0.05ex]
 \hline\hline
 Paralleled Ry & $O(1)$ & $O(\log(m/\epsilon))$ & $O(2^m)$ & $O(2^m)$ & $O(2^m\log(m/\epsilon))$ \\ [0.05ex]
 \hline
 SPF ($U_{\mathrm{SP}}$) & $O(m)$ & $O(m)$ & $O(2^m)$ & $O(m2^m)$ & $O(m2^m)$ \\ [0.05ex]
 \hline
 Paralleled X ($U_{\mathrm{SP}}$) & $O(1)$ & $O(1)$ & $O(2^m)$ & $O(2^m)$ & $O(2^m)$ \\ [0.05ex]
 \hline
 FLAG ($U_{\mathrm{SP}}$) & $O(m)$ & $O(m)$ & $O(2^m)$ & $O(m2^m)$ & $O(m2^m)$ \\ [0.05ex]
 \hline
 Paralleled Control-Ry & $O(1)$ & $O(\log(m/\epsilon))$ & $O(2^m)$ & $O(2^m)$ & $O(2^m\log(m/\epsilon))$ \\ [0.05ex]
 \hline
\textbf{Total for }\bm{$U_{\mathrm{SP}}$} & $O(m)$ & $O(\log(2^m/\epsilon))$ & $O(2^m)$ & $O(m2^m)$ & $O(2^m\log(2^m/\epsilon))$ \\ [0.05ex]
\hline\hline
 Paralleled X ($U_{\mathrm{CSP}}$) & $O(1)$ & $O(1)$ & $O(2^{n-m})$ & $O(2^{n-m})$ & $O(2^{n-m})$ \\ [0.05ex]
 \hline
 LOADF & $O(n)$ & $O(n+\log(1/\epsilon))$ & $O(2^n)$ & $O(2^n)$ & $O(2^n\log(n/\epsilon))$ \\ [0.05ex]
 \hline
 SPF ($U_{\mathrm{CSP}}$) & $O(n - m)$ & $O(n - m)$ & $O(2^{n-m})$ & $O((n-m)2^{n-m})$ & $O((n-m)2^{n-m})$ \\ [0.05ex]
 \hline
 FLAG ($U_{\mathrm{CSP}}$) & $O(n - m)$ & $O(n - m)$ & $O(2^{n-m})$ & $O((n-m)2^{n-m})$ & $O((n-m)2^{n-m})$ \\ [0.05ex]
\hline
\textbf{Total for }\bm{$U_{\mathrm{CSP}}$} & $O(n)$ & $O(n+\log(1/\epsilon))$ & $O(2^n)$ & $O(2^n + (n-m)2^{n-m})$ & $O(2^n\log(n/\epsilon) + (n-m)2^{n-m})$ \\ [0.05ex]
\hline\hline
 \textbf{Total} & $O(n)$ & $O(n+\log(n/\epsilon))$ & $O(2^n)$ & $O(m2^m + 2^n + (n-m)2^{n-m})$ & $O(m2^m + 2^n\log(n/\epsilon) + (n-m)2^{n-m})$\\ [0.05ex]
 \hline
  \textbf{If Eq.~(\ref{eq:mn_relation}) Satisfied} & \bm{$O(n)$} & \bm{$O(n+\log(n/\epsilon))$} & \bm{$O(2^n)$} & \bm{$O(2^n)$} & \bm{$O(2^n\log(n/\epsilon))$}\\ [0.05ex]
 \hline
\end{tabular}
 \end{adjustbox}
 }
\caption{Individual Complexity Summary for Each Step of State Preparation Protocol. ($M = 2^m$,  $N = 2^n$)}
\label{table:complexity_summary}
\end{table*}
\endgroup

In this subsection, we compute the overall depth and ancilla allocation of our protocol. To verify the stated complexities of subroutines, we refer the reader to the sections where those subroutines are implemented. We state complexities in the $\{\mathrm{U(2)},\mathrm{CNOT}\}$ gate set, and then quote the associated complexity in the $\{\mathrm{H,S,T}, \mathrm{CNOT}\}$ gate set in \textit{(parentheses)} when it is different. We summarize all the individual complexity contributions and the resulted total complexities in Table \ref{table:complexity_summary}. 

Note that to achieve overall error $\epsilon$ in state preparation or controlled state preparation, it suffices to prepare each angle state $\ket{\theta_{s,p}}$ to precision $\epsilon/n$. To see this, note that the protocol is equivalent to a sequence of $n$ multi-controlled rotations on the data qubits. If every angle is accurate up to error $\epsilon/n$, each of these $n$ unitaries is enacted up to spectral norm error $\epsilon/n$, yielding total error at most $\epsilon$. More thorough justifications of this appear in \cite[Section V B]{clader2023quantum} and \cite[Section VIIIA of Supplementary Material]{zhang2022quantum}.

We begin with $U_{\mathrm{SP}}$, implemented as in Fig.~\ref{fig:SP}. The parallel rotation gates and parallel controlled rotation gates are accomplished in depth $O(1)$ ($\mathrm{D_{approx}}$ $= O(\log(m/\epsilon))$), while the SPF and FLAG circuits have depth $O(m)$, yielding a total depth $O(m)$ (total\footnote{Recall that $O(m+\log(m/\epsilon)) = O(m+\log(1/\epsilon))$, since $\log(m/\epsilon) = \log(m) + \log(1/\epsilon)$, and $O(m + \log(m)) = O(m)$.} $\mathrm{D_{approx}}$ $= O(m+\log(1/\epsilon))=O(\log(2^m/\epsilon))$). The required total qubits are the $m + 2M-2$ that appear in Fig.~\ref{fig:SP}, plus an additional $O(M)$ needed to perform the SPF and FLAG operations. The product of space and depth gives an immediate upper bound of spacetime allocation $O(m2^m)$ ($\mathrm{SA_{approx}} = $ $O(2^m\log(2^m/\epsilon))$, and indeed for $U_{\mathrm{SP}}$ there is a matching lower bound (up to constant factor). 

We now consider $U_{\mathrm{CSP}}$. Here, we again note that to perform the operation to $\epsilon$ precision, it suffices to perform each single-qubit rotation gate to error $\epsilon/n$.  The LOADF operation has depth $O(n)$ ($\mathrm{D_{approx}} = $ $O(n+\log(1/\epsilon))$. Meanwhile, in this context, the SPF and FLAG operations each have depth $O(n-m)$, yielding total depth $O(n)$ (total depth $O(n+\log(1/\epsilon)$). The required total qubits are the $n + 2N/M$ that appear in Fig.~\ref{fig:CSP}, plus the additional $O(N)$ needed to perform the LOADF, SPF, and FLAG circuits. The product of space and depth gives an immediate upper bound of spacetime allocation $O(n2^n)$ ($\mathrm{SA_{approx}} = $ $O(2^n\log(2^n/\epsilon))$). However, in this case, the actual spacetime allocation is better than this upper bound. The LOADF subroutine achieves spacetime allocation $O(2^n)$ ($\mathrm{SA_{approx}} = $ $O(2^n\log(n/\epsilon))$) even though it acts on $O(2^n)$ qubits  depth $O(n)$ ($\mathrm{D_{approx}} = $ $O(n+\log(1/\epsilon))$). Meanwhile, FLAG and SPF have spacetime allocation $O((n-m)2^{n-m})$. In total, the procedure has spacetime allocation $O(2^n+(n-m)2^{n-m})$ ($\mathrm{SA_{approx}} = $ $O(2^n\log(n/\epsilon)+ (n-m)2^{n-m})$). 

Adding both parts together, the SP+CSP protocol achieves depth $O(n)$ ($\mathrm{D_{approx}} = $ $O(n+\log(1/\epsilon))$ and has spacetime allocation $O(m2^m + 2^n + (n-m)2^{n-m})$ ($\mathrm{SA_{approx}} = $ $O(m2^m + 2^n\log(n/\epsilon) + (n-m)2^{n-m})$).

If $m$ is chosen such that
\begin{equation}\label{eq:mn_relation}
    \lceil \log_2(n) \rceil \le m \le \lfloor n - \log_2(n) \rfloor,
\end{equation} we see that the stated result of spacetime allocation $O(2^n)$ ($\mathrm{SA_{approx}} = $ $O(2^n \log(n/\epsilon))$) is true.  

In retrospect, we can pinpoint the rationale for pursuing the SP+CSP approach to state preparation. The issue with simply performing the $U_{\mathrm{SP}}$ protocol with $n=m$ is that it would require $2^n-1$ ancilla qubits to store the $2^n-1$ angles, and each of these ancillae must remain allocated for the entire $O(n)$ depth of the circuit. The SP+CSP protocol gets around this by first preparing the $m$-qubit state $\ket{\phi}$, requiring only $2^m-1$ angles, and for each of the $2^m$ basis states, preparing a different $n-m$ qubit state in the remaining qubits. The latter operation requires $2^{n-m}-1$ angles (stored in the buffer ancilla register B), so we avoid the need for $O(2^n)$ angle states to be allocated for the entire $O(n)$ depth. The trick comes from the ability to load the correct set of $2^{n-m}-1$ angles among all $2^n-2^m$ angles with only $O(2^n)$ spacetime allocation, which is accomplished by our LOADF implementation presented in the appendix.

In the appendix, we also comment on parts of the LOADF subroutine where clean ancillae in the $\ket{0}$ state can be replaced with dirty ancillae in an arbitrary state. In total, a constant fraction of the ancillae in our construction can be dirty, but a constant fraction must remain clean. Future work aims to investigate if a fraction of ancillae that asymptotically approaches $1$ can be made dirty. If that is the case, it may be meaningful to differentiate between the portion of spacetime allocation that corresponds to clean ancillae and the fraction that corresponds to dirty ancillae. 

\section{Preparing Many Copies of Independent Quantum States}
\label{sec:multiple_copy}

This section will show how we manage to utilize the novel encoding method we proposed in Sec. \ref{sec:SP+CSP} to more efficiently prepare many copies of arbitrary quantum states (same or different ones) using the $\{\mathrm{U(2), CNOT}\}$ gate set. Specifically, we consider the task of preparing a product state of $w$ separate $N$-dimensional states as quickly as possible. If we had $O(wN)$ ancilla qubits, we could perform state preparation on each of the $w$ copies in parallel for total depth $O(\log(N))$. However, in applications, $N$ is likely to be large, and we may not have so many ancillae. Suppose we have only $O(N)$ ancilla qubits---enough to prepare one state in $O(\log(N))$ depth, but not $w$ states. Naively, we could prepare the $w$ states in series, incurring depth $O(w\log(N))$. However, since the spacetime allocation required for a single state preparation is $O(N)$, the spacetime allocation for preparing $w$ states is $O(wN)$, suggesting that if the $O(N)$ ancillae are used reused optimally, we might achieve $O(w)$ total depth, rather than $O(w\log(N))$. Indeed, we now illustrate that using our SP+CSP protocol one can accomplish the task in depth $O(w + \log(N))$, which is constant $O(1)$ depth per copy when $w \geq \Omega(\log(N))$. 

There are two essential properties in the single-copy arbitrary state preparation method we described in section \ref{sec:SP+CSP}:
\begin{enumerate}
    \item All ancilla qubits are disentangled from the data qubits and returned to the $\ket{0}$ state. The previously used ancillae are ready to act as either new data qubits or new ancilla qubits.
    \item The amount of spacetime occupied by active ancilla qubits is only $O(N)$, even though the depth is $O(\log(N))$ and there are $O(N)$ ancilla qubits.
\end{enumerate}
Thus, once the ancillae are returned to $\ket{0}$, they can begin to assist with the next state preparation, even if the previous state preparation has not been completed. We can concatenate and stack many of the state preparation circuits one after another.

Quantitatively speaking, we can prepare many unentangled copies of the same quantum state
\begin{align}
    \ket{\psi}^{\otimes w} &= \left[\prod_{s=0}^{w-1} U_{s}\right] \ket{0^n}^{\otimes w} \nonumber\\
    &= \left(\frac{1}{\lVert \textbf{x}\rVert}\sum_{i=0}^{2^n-1} x_i \ket{i}\right)^{\otimes w}
\end{align}
Here $U_s$'s are the same quantum state preparation unitary but acting on different $\ket{0^n}$ states. The above equation does not include ancillas that assist the implementation of the $U_s$ unitaries: even though the $U_s$'s act on different $\ket{0^n}$ data registers, they can share many of the same ancilla qubits.

We can also prepare many copies of different unentangled quantum states
\begin{align}\label{eq:diff_joint_state}
    \bigotimes_{d = 0}^{w-1} \ket{\psi_j} &= \left[\prod_{d=0}^{w-1} U_{d} \right]\ket{0^n}^{\otimes w} \nonumber\\
    &= \bigotimes_{d = 0}^{w-1} \left(\frac{1}{\lVert \textbf{x}^d\rVert} \sum_{i=0}^{2^n-1} x_{i}^d \ket{i}\right)
\end{align}
Here $U_d$'s are \textit{different} quantum state preparation unitaries acting on different $\ket{0^n}$ states and preparing different $N$-dimensional states described by the vector of coefficients $\mathbf{x}^d$. Even though each $U_d$ is mathematically different, the exact quantum gate scheduling is the same. The only difference is the single-qubit rotation angles. 

\subsection{SP+CSP Multi-Copy}\label{sec:sp+csp_multi_copy}
In order to prepare the joint state in Eq. \eqref{eq:diff_joint_state}, we can now let $U_d = U_{d(\mathrm{CSP})}U_{d(\mathrm{SP})}$. We first execute all the $U_{d(\mathrm{SP})}$ circuits in parallel and let $m = n - \lceil \log_2(n) \rceil$. That is, we first prepare the state
\begin{equation}\label{eq:parallelized_sp}
    \left[\prod_{d=0}^{w-1} U_{d(\mathrm{SP})} \right] \ket{0^n}^{\otimes w} = \bigotimes_{d = 0}^{w-1} (\ket{\phi_d} \otimes \ket{0^{n-m}})\,,
\end{equation}
where $\ket{\phi_d}$ is defined from $\mathbf{x}^d$ as in Eq.~\eqref{eq:phi_def}. 
Then, we jointly execute all the $U_{d(\mathrm{CSP})}$ circuits. 
\begin{equation}\label{eq:parallelized_csp}
    \prod_{d=0}^{w-1} U_{d(\mathrm{CSP})} \left(\bigotimes_{d = 0}^{w-1} (\ket{\phi_d} \otimes \ket{0^{n-m}})\right) = \bigotimes_{d = 0}^{w-1} \ket{\psi_d}
\end{equation}
We do not have sufficient ancillae to perform these circuits in parallel, so we perform them with some ``indentation'' $k$ between each other. That is, we begin the $d=0$ protocol at layer 0, the $d=1$ protocol at layer $k$, the $d=2$ protocol at layer $2k$, etc.

We illustrate the above operation as a quantum circuit in Fig.~\ref{fig:sp+csp_multiple}.
\begin{figure}[h!]
    \centering
    \includegraphics[width=0.3\textwidth]{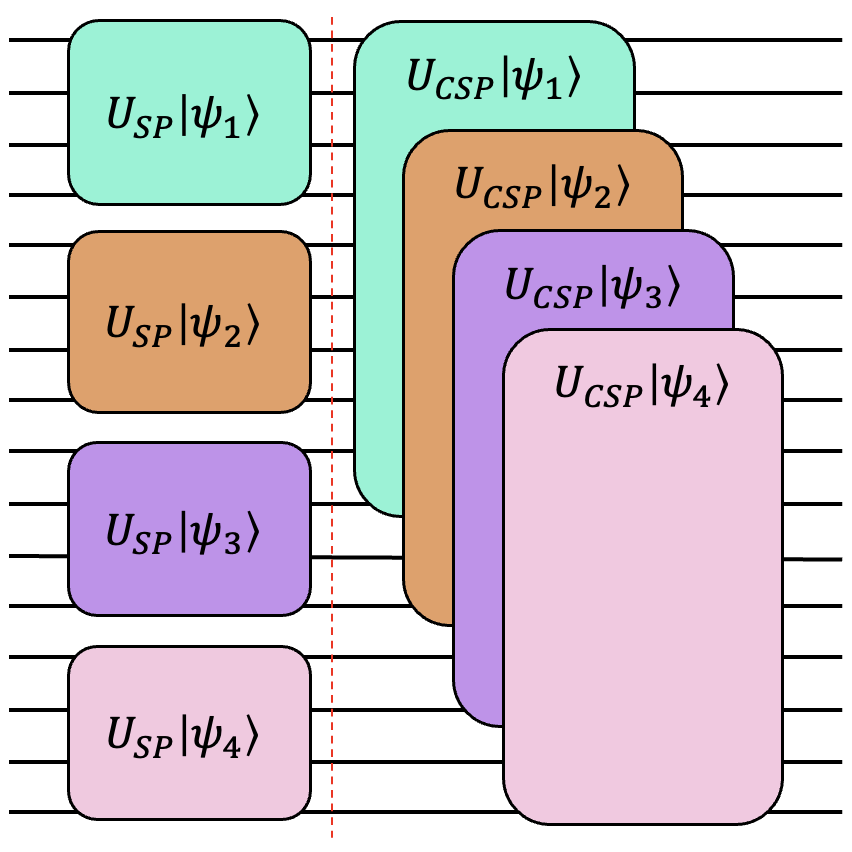}
    \caption{Illustration of stacking different SP+CSP circuit together ($w$ = 4)}
    \label{fig:sp+csp_multiple}
\end{figure}
If we choose $m = n - \lceil \log_2(n) \rceil$, the first operation described in Eq.~\eqref{eq:parallelized_sp} can be achieved in $O(\log(N))$ depth and $O(N)$ spacetime allocation. Since each of the $U_{d(\mathrm{SP})}$ is using $\frac{N}{\log(N)}$ qubits, up to $O(\log(N))$ joint states can be parallelized in this way without introducing additional ancilla qubits.

Now let's walk through the second operation described in Eq.~\eqref{eq:parallelized_csp}. The most ancilla-consuming steps of the CSP circuit are the LOADF operations, and in particular, the layer of doubly-controlled rotation (CCR$_y$) operations that appear in Fig. \ref{fig:LOADF} of the appendix, which requires $O(N)$ ancillae in the $A$ register. Note that these ancillae can be quickly freed up after the parallelized CCR$_{y}$ operations by the $\overline{\mathrm{CopySwap}}$ operation under $O(N)$ spacetime allocation, as shown in Fig.~\ref{fig:LOADF}. Thus we can indent the $U_{\mathrm{CSP}}$ part of the later $\mathrm{U}_d$'s by some constant number of layers $k$.

Besides the $A$ register that uses $O(N)$ ancillae, we also have the $C$, $B$, and $F$ registers (as shown in Fig.~\ref{fig:LOADF}) in the $U_{\mathrm{CSP}}$ operation. By choosing $m = n - \lceil \log_2(n)\rceil $, we would have negligible ancilla counts for the $C$ register ($m\cdot \frac{N}{M} = m\log(N)$ ancillae), the $B$ register ($M = \frac{N}{\log(N)}$ ancillae), and the $F$ register ($\frac{N}{M} = n$ ancillae). Therefore, they will not contribute much to the overall spacetime allocation cost.

We offer a toy model gate-level illustration in Fig.~\ref{fig:csp_multiple_details_comps}. We can see that the circuit depth will be $2w + \log(N)$ instead of $w\log(N)$. The number of ancillae used is upper bounded by $8N$.

\begin{figure}[h!]
    \centering
    \includegraphics[width=0.48\textwidth]{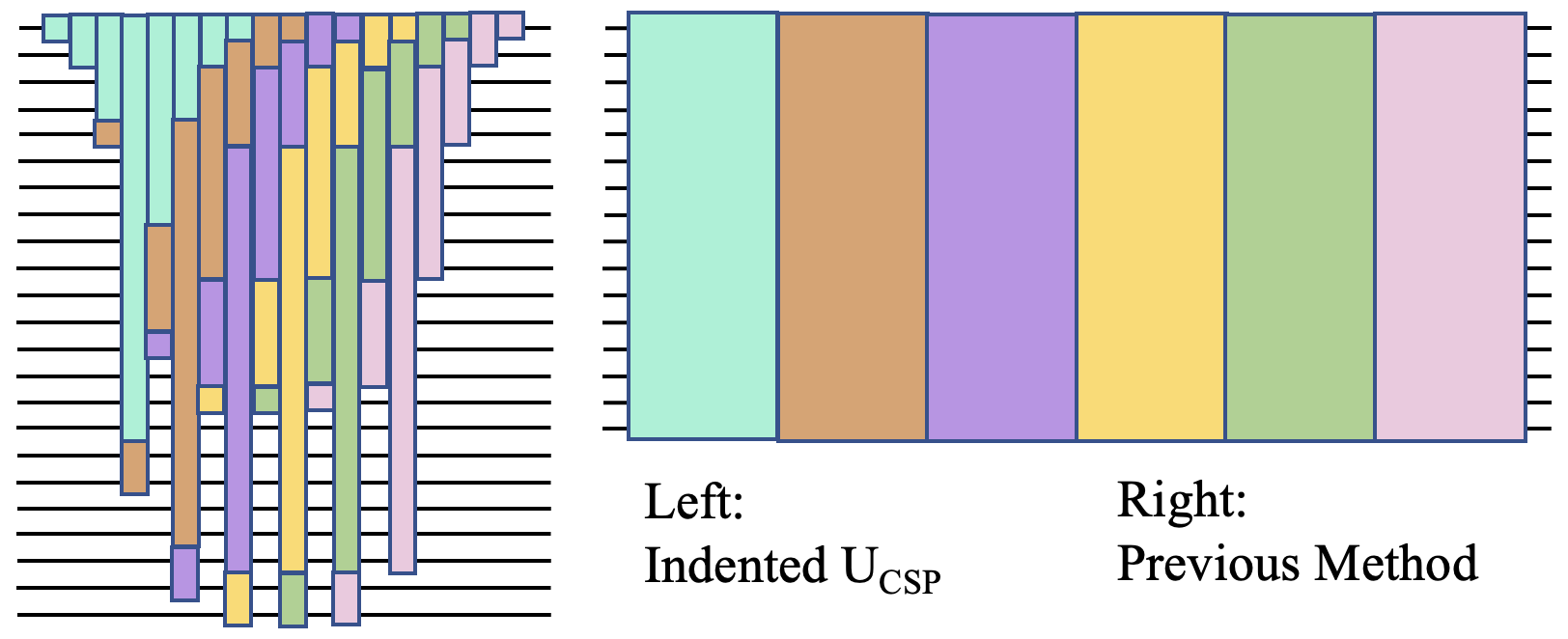}
    \caption{Illustration of stacking \textit{a portion of} the $U_{\mathrm{CSP}}$ circuit (LOADF) with indentation $k = 2$.}
    \label{fig:csp_multiple_details_comps}
\end{figure}

When the ratio between $w$ and $\log(N)$ grows larger, which is often the case in near-term applications such as amplitude encoding for quantum machine learning (as illustrated in Sec.~\ref{sec:qml}), we would effectively have state preparation circuit depth of $O(w)$, same as the depth using the protocol proposed in \cite{yuan2023optimal}.

When using the $\{\mathrm{H, S, T, CNOT}\}$ gate set, one can develop specific compilation strategies that utilize this early-ancilla-free-up structure based on the particular $\epsilon$ values. Asymptotically speaking, since each state preparation occupies spacetime allocation $O(N\log(\log(N)/\epsilon))$, with $O(N)$ ancilla qubits it would be possible to create $w$ states in depth $O(w\log(\log(N)/\epsilon)+\log(N)+\log(1/\epsilon))$. 
\section{Applications}
\label{sec:applications}

In general, quantum measurements are probabilistic. Therefore, for many quantum algorithms, we typically need to run many shots with the same quantum circuit setup. That is, we will need to execute the entire circuit (state preparation, algorithm circuit, and measurement) many times to get enough \textbf{\textit{measurement precision}}. Therefore, for many quantum algorithms that require multiple shots/executions and arbitrary quantum state preparation, our novel state preparation method can reduce the average time-per-shot by as much as a factor linear to the number of qubits.

However, there exist more concrete examples where our novel encoding methods can provide advantages beyond the need to improve measurement precision, including but not limited to: quantum machine learning, Hamiltonian simulation, and solving linear systems with algorithms in the style of the Harrow--Hassidim--Lloyd (HHL) algorithm \cite{harrow2009quantum}.

\subsection{Quantum Machine Learning --- Batching}\label{sec:qml}

One proposed quantum advantage for machine learning tasks is utilizing quantum computers’ exponentially large Hilbert space to encode and process data in parallel, providing potential speedup over classical machine learning methods. One of the bottlenecks to realizing such an advantage is that encoding the classical data into the quantum machine in an exponentially compact form (\textit{amplitude encoding} \cite{schuld2021machine}) is costly.
In amplitude encoding, we wish to encode the feature vector $\mathbf{x} = [x_0, x_1, ..., x_{N-1}]$ into the amplitude of the quantum state such that
\begin{equation}
    \ket{\psi} = \frac{1}{\lVert \textbf{x}\rVert}\sum_{i=0}^{N-1} x_i \ket{i}
\end{equation}
Notice that $N = 2^n$, where $n$ is the number of required qubits. This indicates we can encode the feature vector in an exponentially compact way.

For instance, if we have a 2 $\times$ 2 image that has a pixel value array of $\mathbf{x}$ = [232, 31, 62, 137] (each pixel has an integer color value ranging from 0 to 255), we only need $\log_2 4 = 2$ qubits to encode the values on the amplitudes: 
\begin{equation}
    \ket{\psi} = 0.834\ket{00} + 0.111\ket{01} + 0.223\ket{10} + 0.492\ket{11}
\end{equation}

In typical machine learning tasks, a model is trained by feeding it many copies of training data and adjusting the parameters in the model to minimize a loss function evaluated based on that data. One can update the parameters after computing the loss sum of many data points (e.g., batch gradient descent \cite{ruder2016overview}). 

In the quantum computing setting, one could execute a QML iteration as follows (also depicted in Fig.~\ref{fig:1_shot}):
\begin{enumerate}
    \item encoding many data points into separate quantum states simultaneously;
    \item processing all states with, e.g., separate quantum neural network circuits simultaneously; and
    \item measuring the output of each circuit simultaneously (disjoint or joint measurements) and computing the total loss.
\end{enumerate}

\begin{figure}[h!]
\centering
\includegraphics[width=0.45\textwidth]{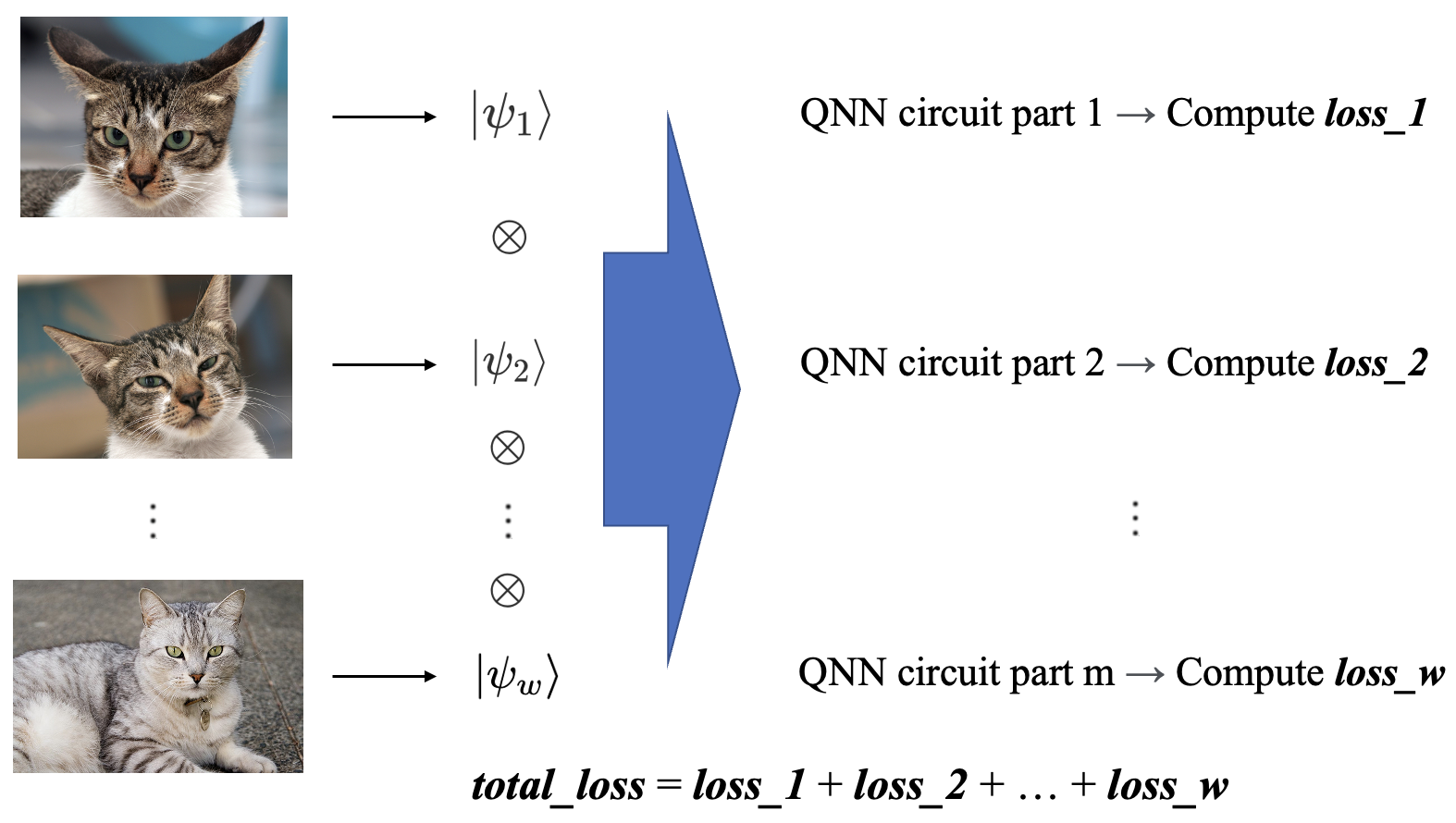}
\caption{simultaneous execution method example}
\label{fig:1_shot}
\end{figure}

Since the encoding process can take significant portions of the total required quantum resources, one must consider the best strategy for state preparations. 

Using our SP+CSP encoding method, we would be able to achieve $O(w+\log(N))$ circuit depth using only $O(N)$ qubits, using the indentation method described in Sec.~\ref{sec:multiple_copy}. In principle, one could encode many data points in effectively \textit{constant} depth, given $O(N)$ ancilla qubits. When running on near-term machines, this effectively matches the best available \textit{theoretical} performance using the protocol proposed by Yuan et al. \cite{yuan2023optimal} using the $\{$U(2), CNOT$\}$ gate set. The \textit{actual} best strategy using near-term machines will depend on the machine properties, such as additional available native gates (e.g., CSWAP) and connectivity allowance. For fault-tolerant level quantum machine learning algorithms (e.g., using the $\{\mathrm{H,S,T,CNOT}\}$ gate set), our SP+CSP protocol has clear advantages in terms of both depth and spacetime allocation.

\subsection{Hamiltonian Simulation --- Linear combination of unitaries}
Hamiltonian simulation is the task of synthesizing the time evolution unitary $U(t)$ given a length of time $t$ and a description of some Hamiltonian $H$. For example, Hamiltonian simulation allows one to measure the properties of a time-evolved state
\begin{equation}\label{eq:hs_def}
     \ket{\psi(t)}=U(t)\ket{\psi(0)}
\end{equation}
where $\ket{\psi(0)}$ is the initial state and 
\begin{equation}
    U(t) = e^{-itH}\,.
\end{equation}

One method to perform Hamiltonian simulation is to use linear combination of unitaries (LCU) \cite{berry2015simulating, childs2018toward}. The idea of this method is to approximate the Hamiltonian evolution operator with a Taylor expansion truncated to order $K$ and then expanded as a linear combination of (unitary) Pauli matrices
\begin{align}
    U(t) &\approx \sum_{k=0}^K \frac{(-itH)^k}{k!}\\
         &= \sum_{j=0}^{\Gamma-1} \beta_j V_j
\end{align}
where $V_j$ are Pauli matrices and the number of Pauli matrices in the linear combination is $\Gamma = \sum^K_{k=0} L^k$, where $L$ is the number of Pauli terms in the Hamiltonian $H$. Let $a = \lceil \log_2(\Gamma) \rceil$.

Implementing the linear combination of unitaries is then done by constructing a SELECT operator and a state-preparation operator $B$, defined as
\begin{align}
    \mathrm{SELECT} &= \sum_{j=0}^{\Gamma-1} \ket{j}\bra{j} \otimes V_j \\
    B \ket{0^{a}}&= \frac{\sum_{j=0}^{\Gamma-1}\sqrt{\beta_j} \ket{j}}{\sqrt{\sum_{j=0}^{\Gamma-1} \beta_j}}\,.
\end{align}
Combining these to form
\begin{equation}
    W =  (B^\dagger \otimes I) \mathrm{SELECT}(B \otimes I)\,,
\end{equation}
we see that
\begin{equation}
   \left(\bra{0^a} \otimes I\right) W \left(\ket{0^a} \otimes I\right) = \frac{\sum_{j=0}^{\Gamma-1} \beta_j V_j}{\sum_{j=0}^{\Gamma-1} \beta_j} \approx \frac{U(t)}{\sum_{j=0}^{\Gamma-1} \beta_j}
\end{equation}
The above equation is equivalent to the statement that $W$ is an approximate ``block-encoding'' of $U(t)$ \cite{chakraborty2018BlockMatrixPowers,gilyen2018QSingValTransf}, with block-encoding factor $\sum_{j=0}^{\Gamma-1} \beta_j$. Thus, $U(t)$ can be applied by application of $W$ and postselection onto the ancilla state $\ket{0^a}$, which occurs with probability $(\sum_{j=0}^{\Gamma-1}\beta_j)^{-2}$. The success probability can be boosted to 1 by dividing the evolution time $t$ into $r$ segments of length $t/r$ for a specific choice of $r$, and applying oblivious amplitude amplification \cite{berry2015hamiltonian,berry2015simulating}. 

In any case, the state-preparation operator $B$ is a key subroutine of the algorithm, which could be implemented in shallow depth with our state-preparation method. 
In a typical quantum chemistry simulation setting, the second-quantized Hamiltonian considered includes two-electron and sometimes three-electron integrals, which would make $\Gamma$ scale as $O(n^4)$ \cite{helgaker2013molecular} and sometimes $O(n^6)$ \cite{motta2020quantum, mcardle2020improving} (where $n$ indicates the number of orbitals). The number of terms in the LCU will be even larger (depending on the truncation parameter); however, our low-depth state preparation method would have depth scaling as $O(\log(n))$. 

Recent works \cite{bubeck2020entanglement, chen2022exponential, huang2022quantum} have shown the advantage of performing \textbf{\textit{joint measurements}} on multiple copies of quantum states (e.g., $\ket{\psi(t)}$ in Eq.~\eqref{eq:hs_def}) compared to separate measurements of the end states. 
Doing this would give us a better sample complexity scaling, which, in the end, will result in better total circuit complexity by saving the total number of circuit executions in order to reach certain measurement thresholds. 

\subsection{Quantum linear system solving --- Repeat until Success}

The classical linear system problem is defined as follows: given an $N \times N$ invertible matrix $A$, and a $N \times 1$ vector $\mathbf{b}$, find the solution vector $\mathbf{x}$  such that $A\mathbf{x}$ = $\mathbf{b}$. The \textit{quantum} linear system problem is to perform the following related task: prepare a state $\ket{\mathbf{x}} = \frac{1}{\lVert \mathbf{x} \rVert} \sum_{i} x_i \ket{i}$ whose amplitudes are proportional to the solution vector $\mathbf{x}$. While the state $\ket{\mathbf{x}}$ does not give access to all $N$ entries of the vector $\mathbf{x}$, multiple copies of $\ket{\mathbf{x}}$ can be used, for example, to estimate expectation values ${\bra{\mathbf{x}} M \ket{\mathbf{x}}}$ of observables $M$.
The first quantum algorithm for solving the quantum linear system problem was by Harrow, Hassidim, and Lloyd (HHL) \cite{harrow2009quantum}. Their solution begins by preparing the state $\ket{\mathbf{b}} = \lVert \mathbf{b} \rVert^{-1} \sum b_i\ket{i}$, where $b_i$ are the entries of the vector $\mathbf{b}$. Then, if $A$ is sparse and, assuming coherent query access to the entries of $A$ in $\mathrm{polylog}(N)$ time, the state $\ket{\mathbf{x}}$ can be prepared to high precision in time $\mathrm{poly}(\kappa,\log(N))$, where $\kappa$ is the \textit{condition number} of $A$, that is, the ratio of the largest to smallest singular values. When $\kappa = O(1)$, this is exponentially faster than classical methods that manipulate vectors of length-$N$, such as Gaussian elimination or conjugate gradient descent \cite{shewchuk1994introduction}). However, reading out useful information from $\ket{\mathbf{x}}$ can often negate this exponential speedup (leaving the possibility of a polynomial speedup) in specific applications, such as solving differential equations \cite{montanaro2016quantumFEM}. Even in these cases, there is still a possible exponential advantage in space complexity, due to the exponentially compact representation of the data involved as quantum states. Relatedly, this compactness is leveraged to yield an unconditional exponential quantum advantage for the linear system problem in the communication complexity setting \cite{montanaro2023quantum}.

It is important to note that even if $A$ is a sparse matrix, the vector $\mathbf{b}$ can often be dense. This is the case, for example, when solving differential equations via mapping to linear systems \cite{montanaro2016quantumFEM}, where $\mathbf{b}$ encodes arbitrary boundary conditions of the problem.  Thus, preparing the state $\ket{\mathbf{b}}$ as the first step of the quantum algorithm is precisely an application of our state-preparation protocol. 

There are a few reasons that our low-depth state preparation protocol is well suited for usage within algorithms for the quantum linear systems problem. When $\mathbf{b}$ is dense, preparing $\mathbf{b}$ will dominate the algorithm's runtime unless a low-depth state-preparation method is used. Furthermore, extracting information from the solution vector $\ket{\mathbf{x}}$ may require repeating the algorithm many times, and thus preparing many copies of $\ket{\mathbf{b}}$; as outlined in Sec.~\ref{sec:multiple_copy}, the fact that our method has optimal spacetime allocation allows many copies to be rapidly prepared with efficient ancilla usage. Furthermore, depending on which quantum approach is taken, the algorithm itself may require preparing many copies of $\ket{\mathbf{b}}$ to generate a single copy of $\ket{\mathbf{x}}$. For example, in modern approaches such as \cite{chakraborty2018BlockMatrixPowers}, one can access the matrix $A$, via a unitary ``block-encoding'' $U_A$. Using signal processing techniques, one can create a block-encoding $U_{A^{-1}}$ of $A^{-1}$. The unitary $U_{A^{-1}}$ uses $\tilde{O}(\kappa)$ calls to the unitary $U_A$, and performs the operation (up to small error)
\begin{equation}
    U_{A^{-1}} \ket{\mathbf{b}}\ket{0^\ell} \approx c\ket{\mathbf{x}}\ket{0^\ell} + \ket{\perp}
\end{equation}
where $\ket{\perp}$ is in the null space of the projector $I \otimes \ket{0^{\ell}}\bra{0^\ell}$, and $\lvert c \rvert = O(1/\kappa)$. A measurement of the ancilla register then produces the desired state $\ket{\mathbf{x}}$ with probability $|c|^2 = O(1/\kappa^2)$. Thus, to produce $\ket{\mathbf{x}}$, it suffices to create $O(\kappa^2)$ copies of $\ket{\mathbf{b}}$ and make $\tilde{O}(\kappa)$ queries to $U_A$ for each copy, for a total of $\tilde{O}(\kappa^3)$ queries. In principle, these copies could be processed in parallel, and, as described above, our method offers a minimal-depth approach for preparing $\ket{\mathbf{b}}^{\otimes(O(\kappa^2))}$ with only $O(N)$ ancillae. Using amplitude amplification, the number of queries to $U_A$ can be reduced to $\tilde{O}(\kappa^2)$, which would involve $O(\kappa)$ \textit{serial} reflections about the state $\ket{\mathbf{b}}$, an operation that can be performed optimally by our protocol (see App.~\ref{sec:input_not_0}). Finally, using variable-time amplitude amplification \cite{ambainis2012variable,chakraborty2018BlockMatrixPowers} or ideas from adiabatic quantum computing \cite{subasi2019adiabatic,costa2022optimal}, the $U_A$ query complexity can be further reduced to $O(\kappa)$. This process also requires $O(\kappa)$ serial applications of our state-preparation protocol.

\subsection{Block-encodings}
We have already seen above how block-encodings are useful as a framework for giving a quantum algorithm, such as Hamiltonian simulation or quantum linear systems solvers, access to some underlying data \cite{chakraborty2018BlockMatrixPowers, gilyen2018QSingValTransf,martyn2021grandUnification}. Quantum state preparation is also a useful ingredient for constructing block-encodings.

The definition of a block-encoding of a matrix $A$ is a unitary $U_A$ where $A$ appears in the upper left block, scaled by some constant $\alpha$ such that $\lVert A/\alpha \rVert \leq 1$:
\begin{equation}
    U_A = \left[\begin{matrix}
        A/\alpha & \cdot \\
        \cdot & \cdot 
    \end{matrix}\right]\,.
\end{equation}
For example, given a Hamiltonian $H$ representing some physical system, one may seek a block-encoding $U_H$ of $H$, which can then be used to build algorithms that produce the ground state or thermal state of $H$. It can also be used to perform Hamiltonian simulation via qubitization \cite{low2019hamiltonian,gilyen2018QSingValTransf}, a technique that can have superior performance to the LCU method discussed previously.

How the block-encoding $U_A$ is constructed depends on the contents of $A$ and the context of the algorithm. Many of the most common methods explicitly involve the need for QSP or controlled QSP (see Sec.~4.2 of \cite{gilyen2018QSingValTransf}).

For example, when $A$ is a Gram matrix---that is, a matrix for which entries $A_{ij} = \braket{\psi_i|\phi_j}$ for some sets $\{\ket{\psi_i}\}_i$, $\{\ket{\phi_j}\}_j$---one can provide a block encoding of $A$ as $U_A = U_L^\dagger U_R$, where $U_L$ and $U_R$ are unitaries that prepare the arbitrary states $\ket{\psi_i}$ and $\ket{\phi_j}$ \cite[Lemma 47]{gilyen2018QSingValTransf}). A similar approach can be used to construct block-encodings of arbitrary matrices \cite{chakraborty2018BlockMatrixPowers,clader2023quantum}.

Both $U_L$ and $U_R$ can be constructed using our controlled state preparation protocol with the optimized depth of $O(\log (N))$ and $O(N)$ ancilla qubits.


\section{Code Availability}\label{sec:code}
We present two gate-level implementation examples at \url{https://github.com/guikaiwen/QSP_Paper_Artifact}, including:
\begin{itemize}
    \item a 2 + 2 case ($m = 2$, $n = 4$) case without COPY
    \item a 1 + 2 case ($m = 1$, $n = 3$) case with COPY
\end{itemize}
\section*{Acknowledgments}
We thank Connor Hann, Yunong Shi, Erhan Bas, Pei Zeng, Ming Yuan, Siteng Kang, Kyle Brubaker, Martin Schuetz, Fei Chen, Liang Jiang, Eric Kessler, Senrui Chen, Sisi Zhou, and Gideon Lee for their helpful discussions. This work is funded in part by EPiQC, an NSF Expedition in Computing, under award CCF-1730449. FTC is Chief Scientist for Quantum Software at Infleqtion and an advisor to Quantum Circuits, Inc.\\

\bibliographystyle{quantum}
\bibliography{reference.bib}

\newpage
\appendix

\onecolumngrid
\newpage

\section{Circuit implementations}\label{sec:appendix}

\subsection{Qubit Data Copying}\label{sec:data_copy_circ}

Our implementation of SPF and FLAG circuits will require additional ancillae that are used as part of coherent data-copying subroutines. That is, we perform the isometry, which was also utilized, for example, in Refs.~\cite{sun2023asymptotically, yuan2023optimal}
\begin{equation}\label{eq:copy_isometry}
    \alpha \ket{0} +\beta \ket{1} \mapsto \alpha \ket{0^{c}} + \beta \ket{1^{c}}
\end{equation}
using only CNOT gates and in depth $\lceil \log_2(c) \rceil $. We will need this subroutine because we often want to perform many CSWAP gates with disjoint targets but controlled by the same register. We can accomplish this in shallower depth by copying the control register and applying the CSWAPs in parallel. Clader et al.~\cite{clader2023quantum} avoided this issue by assuming the ability to do FANOUT-CNOT with an arbitrary number of targets, which is a Clifford gate, in a single time step.

The protocol for copying consists of $\lceil \log_2(c)\rceil $ sequential depth-1 operations, labeled by $\bigoplus_0,\ldots,\bigoplus_{\lceil \log_2(c)\rceil -1}$. We refer to these as \textit{copying layers}. In particular, $\bigoplus_j$ consists of a single layer of $2^j$ parallel CNOTs where the targets of the CNOTs are fresh ancillae. This is illustrated in Fig.~\ref{fig:State_Copying_Subroutine} and described in Algorithm~\ref{algo:state_copy}. Notably, while the total number of qubits is $c$ and the depth is $\lceil \log_2(c)\rceil $, most of the ancillae need not be allocated until close to the end of the protocol, and thus the spacetime allocation is $\Theta(c)$, rather than $O(c \log(c))$. 

\begin{figure}[h!]
\centering
\scalebox{0.8}{
\begin{quantikz}[row sep={2em,between origins}, column sep=1em, align equals at=1.5]
\lstick{$\alpha\ket{0} + \beta\ket{1}$} & \gate[8,nwires={2,3,4,5,6,7,8}]{\mathrm{COPY}_8} & \qw && &\gate[2,nwires={2}]{\bigoplus_0}&\gate[4,nwires={3,4}]{\bigoplus_1}&\gate[8,nwires={5,6,7,8}]{\bigoplus_2}&\qw&&&&& \ctrl{1} \slice[]{} & \qw  &[1em]\qw & \ctrl{2} & \qw \slice[]{} &[1em]  \qw &\qw & \ctrl{4} &\qw&\qw&\qw& \qw\rstick[wires=8]{\Large $\substack{\alpha\ket{0^8} \\+ \\\beta\ket{1^8}}$} \\
&&\qw&&&&&&\qw&&&& \lstick{$\ket{0}$}&\targ{}                                              & \qw &\qw & \qw &\ctrl{2} & \qw &\qw &\qw& \ctrl{4}  &\qw&\qw& \qw \\
 &&\qw&&                                             &&&&\qw&&&&& & &\lstick{$\ket{0}$} &\targ{} &\qw & \qw &\qw &\qw&\qw& \ctrl{4}&\qw& \qw \\
 &&\qw&&                                             &&&&\qw&&&&& & &\lstick{$\ket{0}$} &\qw &\targ{} & \qw &\qw &\qw&\qw&\qw& \ctrl{4} & \qw \\
&&\qw&\rb{0.5cm}{=}&                                          &&&&\qw&\rb{0.5cm}{=} &&& && & && & &\lstick{$\ket{0}$} &  \targ{} &\qw&\qw&\qw& \qw \\
%
 & &\qw &&& &                                             &&\qw&& &  &&&& && &  &\lstick{$\ket{0}$} &\qw& \targ{} &\qw &\qw & \qw  \\
%
   & &\qw &&& &                                             &&\qw&&& &&&&&&&&\lstick{$\ket{0}$}&\qw&\qw& \targ{} &\qw & \qw \\
  & &\qw &&& &                                            &&\qw&&&  &&&&&& &&\lstick{$\ket{0}$}&\qw&\qw&\qw& \targ{}  & \qw   
\end{quantikz}
}
\caption{Example circuit implementing state copy subroutine for $c=8$ qubits. An arbitrary single-qubit state $\alpha \ket{0}+\beta \ket{1}$ is copied into $c=2^t$ qubits (in the sense of mapping to $\alpha \ket{0^c}+\beta\ket{1^c}$) via a series of $t$ layers of CNOTs, denoted $\bigoplus_0,\ldots,\bigoplus_{t-1}$. The $c-1$ qubits are fresh ancilla qubits that begin in the state $\ket{0}$. The depth is $\log_2(c)$, and the spacetime allocation is only $c$, as seen in the right diagram, by waiting to allocate ancillae until the final moment that they are needed.}
\label{fig:State_Copying_Subroutine}
\end{figure}
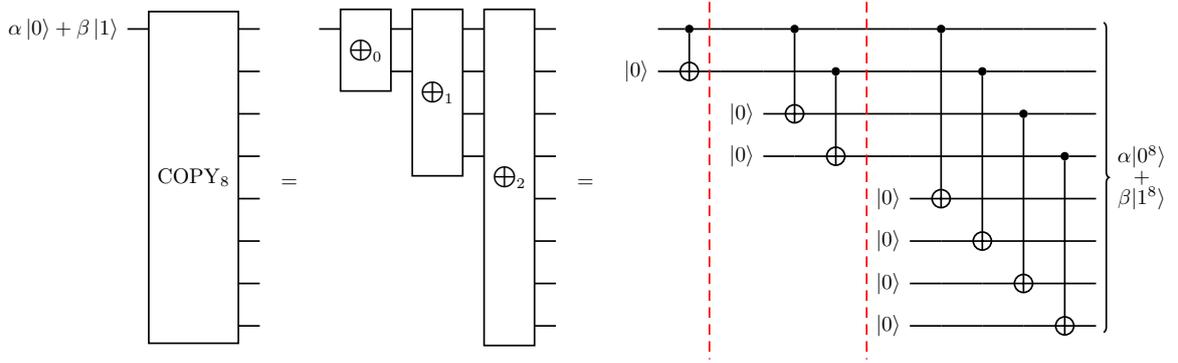

\begin{algorithm}[H]
    \caption{State-Copy Subroutine}\label{algo:state_copy}
    \begin{algorithmic}[1]
    \Procedure{Copy${}_c$}{$R$}
    
    \Comment{$R$ is a register of size $c$, where $c$ is assumed to be the power of 2 for simplicity}
    \For{i \textbf{in} range($\log_2(c)$)} \Comment{Each value of $i$ occupies $O(1)$ depth}
        \State $\bigoplus_i(R)$
    \EndFor
    \EndProcedure

    \Procedure{$\bigoplus_t$}{$R$}
        \For{j in range($2^{t}$)} \Comment{All values of $j$ performed in parallel}
            \State CNOT($R_{jc2^{-t}}$, $R_{jc2^{-t}+c2^{-t- 1}}$)
        \EndFor
    \EndProcedure
    
    \end{algorithmic}
    \Comment{Total D$_{\mathrm{exact}}$: $\log_2(c)$}

    \Comment{Total D$_{\mathrm{approx}}$: $\log_2(c)$}
    
    \Comment{Total SA$_{\mathrm{exact}}$: $2c-2$}

    \Comment{Total SA$_{\mathrm{approx}}$: $2c-2$}
\end{algorithm}

\subsection{Parallel CSWAP}
With the help of the State-Copy subroutine, we can now effectively perform the parallel CSWAP operations using single-qubit gates and two-qubit CNOT gates, without using the FANOUT-CNOT gate that was required in Ref.~\cite{clader2023quantum} to perform the parallel CSWAP when sharing the same control bits.

We denote a layer of $2^t$ parallel CSWAPs by $\mathrm{CS}_t$, as depicted in Fig.~\ref{fig:CS} and Algorithm \ref{algo:CS}. 

\begin{figure}[h!]
\centering
\begin{quantikz}[row sep={1.0em,between origins}, column sep=1em, align equals at=2]
    \qw & \gate{C_0} \vqw{1} & \qw  \\ [4em]
    \qw & \mltg{2}{S_0}      & \qw  \\
    \qw &                    & \qw 
\end{quantikz} \;\;= 
\begin{quantikz}[row sep={1em,between origins}, column sep=1em, align equals at=2]
\qw &\ctrl{2} & \qw \\[4em]
\qw & \targX{} & \qw \\
\qw & \swap{-1} & \qw
\end{quantikz}\qquad \;\;
\begin{quantikz}[row sep={1em,between origins}, column sep=1em, align equals at=3]
    \qw & \mltg{2}{C_1} \vqw{2} & \qw  \\ [-0.3em]
    \qw &                       & \qw  \\ [3.3em]
    \qw & \mltg{4}{S_1}      & \qw  \\
    \qw &                    & \qw  \\
    \qw &                    & \qw  \\
    \qw &                    & \qw  
\end{quantikz} \;\;= 
\begin{quantikz}[row sep={1em,between origins}, column sep=1em, align equals at=3]
\qw &\ctrl{4} & \qw & \qw \\ [-0.3em]
\qw &\qw & \ctrl{4} & \qw \\[3.3em]
\qw & \targX{} & \qw & \qw \\
\qw & \qw       & \targX{} & \qw \\
\qw & \swap{-2} & \qw & \qw \\
\qw & \qw       & \swap{-2} & \qw 
\end{quantikz} \qquad \;\;
\begin{quantikz}[row sep={1em,between origins}, column sep=1em, align equals at=5]
    \qw & \mltg{4}{C_2} \vqw{4} & \qw  \\ [-0.3em]
    \qw &                       & \qw  \\[-0.3em]
    \qw &                       & \qw  \\ [-0.3em]
    \qw &                       & \qw  \\ [1.9em]
    \qw & \mltg{8}{S_2}      & \qw  \\
    \qw &                    & \qw  \\
    \qw &                    & \qw  \\
    \qw &                    & \qw  \\
    \qw &                    & \qw  \\
    \qw &                    & \qw  \\
    \qw &                    & \qw  \\
    \qw &                    & \qw  
\end{quantikz} \;\;= 
\begin{quantikz}[row sep={1em,between origins}, column sep=1em, align equals at=5]
\qw &\ctrl{8} & \qw      & \qw      & \qw        & \qw \\[-0.3em]
\qw &\qw      & \ctrl{8} & \qw      & \qw        & \qw \\[-0.3em]
\qw &\qw      & \qw      & \ctrl{8} & \qw        & \qw \\[-0.3em]
\qw &\qw      & \qw      & \qw      & \ctrl{8}  & \qw \\[1.9em]
\qw & \targX{}& \qw      & \qw      & \qw        & \qw \\
\qw & \qw     & \targX{} & \qw      & \qw        & \qw  \\
\qw & \qw     & \qw      & \targX{} & \qw        & \qw  \\
\qw & \qw     & \qw      & \qw      & \targX{}   & \qw  \\
\qw &\swap{-4}& \qw      & \qw      & \qw        & \qw  \\
\qw & \qw     &\swap{-4} & \qw      & \qw        & \qw  \\
\qw & \qw     & \qw      & \swap{-4}& \qw        & \qw  \\
\qw & \qw     & \qw      & \qw      & \swap{-4}  & \qw   
\end{quantikz}
\caption{\label{fig:CS} Implementation of $\mathrm{CS}_t$ for $t=0,1,2$, which can be accomplished in one layer of parallel CSWAP gates.}
\end{figure}
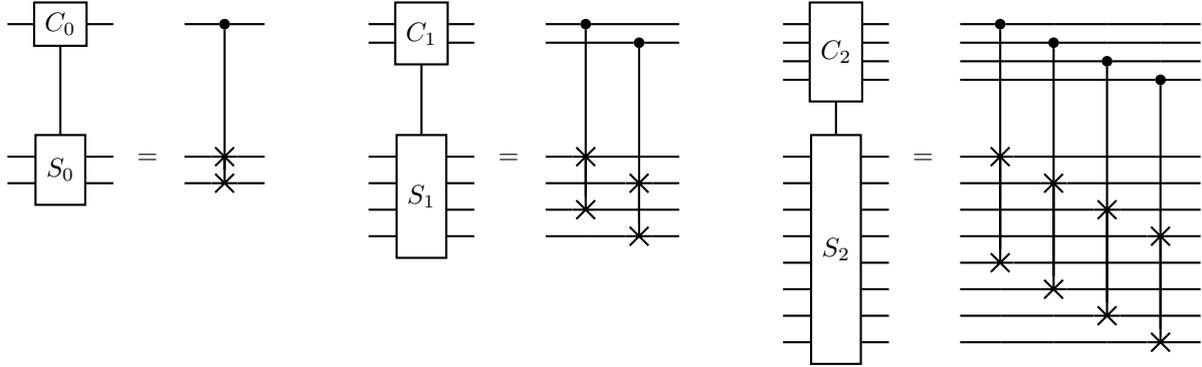

\begin{algorithm}[H]
    \caption{Parallel CSWAP}\label{algo:CS}
    \begin{algorithmic}[1]
    \Procedure{$\mathrm{CS}_t$}{$R, S$}
    
    \Comment{$t$: log number of parallel control-swaps, $R$: control bit data register with at least $2^t$ qubits, $S$: target bit angle register with at least $2^{t+1}$ qubits (note that the subscript here labels the qubit indices of every \textit{single} register)}
    \For{$i$ in range($2^t$)} \Comment{All values of $i$ performed in parallel}
        \State CSWAP$(R_{i}, S_{i}, S_{i + 2^t})$
    \EndFor
    \EndProcedure
    \end{algorithmic}

    

\end{algorithm}

\subsection{COPY Layer Application Example}
In Fig.~\ref{fig:COPY_eq}, we illustrate how we can effectively parallelize many CSWAP gates with the same control bits and different target bits with \textit{constant} ancilla and total depth overhead. The same logic can also apply to Toffoli gates.
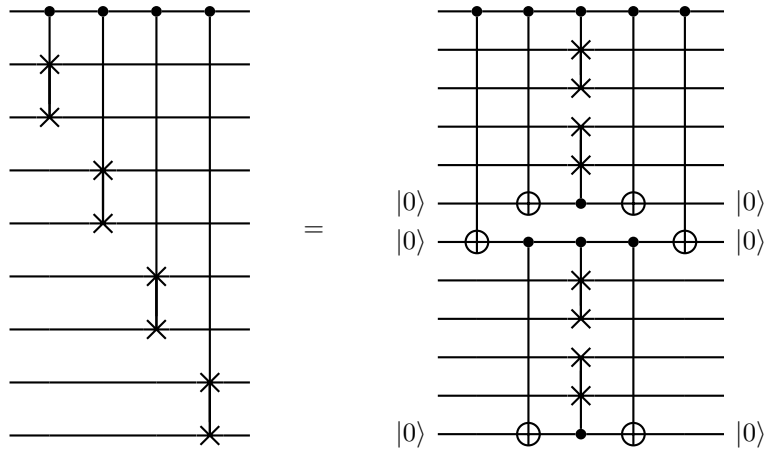
\begin{figure}[h!]
\centering
\begin{quantikz}[row sep={2em,between origins}, column sep=1em, align equals at=1]
    \qw & \ctrl{2} & \ctrl{4} & \ctrl{6} & \ctrl{8} &\qw \\
    \qw & \swap{1} &\qw & \qw & \qw &\qw\\
    \qw & \swap{-1} & \qw & \qw & \qw &\qw\\
    \qw & \qw & \swap{1} &\qw & \qw &\qw\\
    \qw & \qw & \swap{-1} &\qw & \qw &\qw \\
    \qw & \qw & \qw &\swap{1} & \qw &\qw &&\rb{1.2cm}{=}&&\\
    \qw & \qw & \qw & \swap{-1} & \qw&\qw\\
    \qw & \qw &\qw & \qw & \swap{1}&\qw\\
    \qw & \qw &\qw & \qw & \swap{-1} &\qw\\
\end{quantikz}
\begin{quantikz}[row sep={1.45em,between origins}, column sep=1em, align equals at=1]
&&\ctrl{6} &\ctrl{5} & \ctrl{2} & \ctrl{5} & \ctrl{6} & \qw\\
&& \qw& \qw & \swap{1} & \qw & \qw & \qw\\
&& \qw& \qw & \swap{-1} & \qw & \qw & \qw\\
&& \qw& \qw & \swap{1} & \qw & \qw & \qw\\
&& \qw& \qw & \swap{-1} & \qw & \qw & \qw\\
&\lstick{$\ket{0}$} &\qw& \targ{} & \ctrl{-2} & \targ{} & \qw & \qw \rstick{$\ket{0}$}\\
&\lstick{$\ket{0}$} &\targ{} &\ctrl{5} & \ctrl{2} & \ctrl{5} & \targ{} &\qw \rstick{$\ket{0}$}\\
&& \qw& \qw & \swap{1} & \qw & \qw & \qw\\
&& \qw& \qw & \swap{-1} & \qw & \qw & \qw\\
&& \qw& \qw & \swap{1} & \qw & \qw & \qw\\
&& \qw& \qw & \swap{-1} & \qw & \qw & \qw\\
&\lstick{$\ket{0}$} &\qw& \targ{} & \ctrl{-2} & \targ{} & \qw & \qw \rstick{$\ket{0}$}\\
\end{quantikz}
\caption{\label{fig:COPY_eq} Circuit equivalence with and without COPY layers: note that we are only introducing an additional $O(\log(N))$ CNOT layers and $O(N)$ ancilla qubits in the $\ket{0}$ state, same scale as the total circuit depth and total number of ancilla qubits.}
\end{figure}

\newpage

\subsection{SPF}\label{sec:SPF_circ}

Now we describe the implementation of the SPF circuit. The circuit acts on $m + M-1$ qubits, and our implementation uses additional $M/2-m$ ancillae that begin in and end in $\ket{0}$. The circuit is similar to that described in Ref.~\cite{clader2023quantum}, but in our implementation, all gates are two-qubit gates. The idea is to enact $\mathrm{R}_y(\theta_{s,p})$ rotations onto the $m$ data qubit registers in sequence, but where which pair $(s,p)$ is used for the $s$th rotation depends on the value of the first $s-1$ registers. We assume that each of the $\ket{\theta_{s,p}}$ states have already been prepared. The rotations are enacted by, for each $s=0,\ldots,m-1$, swapping (``injecting'') the correct $\ket{\theta_{s,p}}$ state with the $s$th data qubit. Exactly one $\theta_{s,p}$ will be injected for each $s$. Since $(0,0)$ is the only $(s,p)$ with $s=0$, the first step is to swap $\ket{\theta_{0,0}}$ with the first data qubit. The second step is to, conditioned on the first data register being $\ket{1}$, swap the registers holding $\ket{\theta_{1,0}}$ and $\ket{\theta_{1,1}}$, and then inject the register originally holding $\ket{\theta_{1,0}}$ into the second data register (using a swap gate). In general, the $s$th rotation is enacted by first swapping the correct $\ket{\theta_{s,p}}$ to the top of the size-$2^s$ register initially holding the states $\ket{\theta_{s,p}}$, and then swapping the top register into the $s$th data qubit. Naively, this would seem to require at least $O(m^2)$ depth, as there are $m$ rotations, and deciding which qubit to inject for each requires $O(m)$ rounds of controlled-swapping. Ref.~\cite{clader2023quantum} observed that this can be reduced to $O(m)$ depth, assuming access to FANOUT-CNOT with $O(M)$ targets that act in unit time. To decompose such a FANOUT-CNOT gate into two-qubit gates, at least $O(m)$ depth would typically be required; thus, the depth of the Clader et al. \cite{clader2023quantum} SPF implementation in the $\{U(2), \mathrm{CNOT}\}$ gate set would be $O(m^2)$, not $O(m)$. 

We give a circuit compilation method for reducing the depth back to $O(m)$ using a few extra ancillae. We do so by interspersing the CSWAPs with copying layers. In particular, for each $u = s+1,s+2,\ldots,m-1$, data qubit $s$ is needed to control $2^{u-1-s}$ CSWAPs on the angle register that starts in state $\ket{\Theta_{u}}$. Thus, once we have applied $u-1-s$ copying layers to it, we are free to perform the CSWAPs in parallel; we denote the operation that performs $2^t$ CSWAPs in parallel by $\mathrm{CS}_t$, following Fig.~\ref{fig:CS}. While we apply CSWAPs controlled on data qubit $s$, we are free to copy other data qubits in parallel. Overall, we manage to perform the operation in only $O(m)$ depth. 

We describe the SPF circuit in several ways. In Fig.~\ref{fig:SPF}, we give a complete example of SPF for $m=4$, illustrating how ancilla qubits are used to host copies of the data qubits for the purpose of controlling swaps. However, it is hard to fully see the pattern for small values of $m$. In Fig.~\ref{fig:first_part_SPF}, we show the first half of the SPF circuit for a larger $m=7$ example and label the various registers. In this figure, one can easily verify that the depth is $O(m)$. Finally, in Algorithm \ref{algo:SPF}, we give pseudocode for the implementation of SPF, consistent with the labels appearing in Fig.~\ref{fig:first_part_SPF}. To understand the idea behind the circuits, we now describe several rules that must be obeyed (and the reasoning why):
\begin{enumerate}
    \item $\bigoplus_j(D_q)$ needs to happen after $\bigoplus_i(D_q)$ if $j < i$, $\forall q$: 
    \begin{itemize}
        \item This is true by construction, shown in Fig.~\ref{fig:State_Copying_Subroutine}.
    \end{itemize}
    \item $\bigoplus_j(D_q)$ needs to happen after SWAP($D_{q,0}, A_{q,0}$), $\forall j$, $\forall q$:
    \begin{itemize}
        \item This is because the injection of the angle state into data qubit $q$ must occur before we can begin to copy data qubit $q$.
    \end{itemize}
    \item $\mathrm{CS}_k(D_q, \cdot)$ needs to happen after $\bigoplus_{k-1}(D_q)$, $\forall q$: 
    \begin{itemize}
        \item This makes sure we can have enough qubits for the CSWAP sequence to control on in order to guarantee maximal parallelization.
    \end{itemize}
    \item $\mathrm{CS}_j(A_q)$ needs to iterate through all $j$ values in the order of $q-1$ to $0$ before SWAP($D_{q,0}, A_{q,0}$), $\forall q$: 
    \begin{itemize}
        \item This makes sure all the values are swapped into the correct place for the current qubit based on the previous qubits' amplitudes before the current qubit is pumped into the data register.
    \end{itemize}
\end{enumerate}

Because of the existence of criteria $\#2$ and $\#4$, we cannot simply execute the STATE-COPY subroutine at one time from each $D_q$ as we do in the FLAG circuit compilation presented in the next subsection (algorithm~\ref{algo:FLAG}). Doing that would result in $O(m^2)$ depth. Therefore, we have to intersperse the CSWAP sequences with the copying layers.

The total number of copying layers of data qubits $j$ that we need is $m-2-j$ (for $j=0,\ldots,m-3$). Data qubits $m-2$ and $m-1$ need not be copied. Thus, the total number of ancillae needed is $\sum_{j=0}^{m-3}(2^{m-2-j}-1) = 2^{m-1}-m$.

\begin{figure}[h!]
\makebox[\textwidth][c]{
\scalebox{0.67}{
\begin{quantikz}[row sep={2.5em,between origins}, column sep=1em, align equals at=1.5]
\lstick{$\ket{0}$} & \qw & \qw & \gate[13,nwires={2,4}]{\text{SPF}} & \qw && & \qw & \qw & \swap{6} & \gate{C_0}  \vqw{7} & \mltg{2}{\bigoplus_0} & \mltg{2}{C_1}  & \mltg{2}{\bigoplus_1} & \mltg{2}{C_2}  &\qw & \qw & \qw & \qw & \qw & \qw & \qw & \mltg{2}{C_2} & \mltg{2}{\bigoplus_1} & \mltg{2}{C_1} & \mltg{2}{\bigoplus_0} & \gate{C_0}  \vqw{7} & \qw& \qw& \qw\\
                 & & & \linethrough  & &               && &\lstick{$\ket{0^3}$}&\qw            & \qw           & \qw & \qw  \vqw{8} & \qw            & \qw   \vqw{10} & \qw & \qw & \qw           & \qw            & \qw      & \qw & \qw & \qw  \vqw{10} & \qw & \qw  \vqw{8} & \qw & \qw & \rstick{$\ket{0^3}$} \qw\\
\lstick{$\ket{0}$} & \qw                                              & \qw &       & \qw &               & & \qw & \qw & \qw            & \qw           & \swap{5}            & \qw           & \gate{C_0}  \vqw{7}] & \qw &
\mltg{2}{\bigoplus_0} &\mltg{2}{C_1}& \qw           & \qw            & \qw       & \mltg{2}{C_1} & \mltg{2}{\bigoplus_0}  & \qw & \gate{C_0}  \vqw{7} & \qw & \qw & \qw & \qw & \qw& \qw\\
&&&&&   &&&\lstick{$\ket{0}$}& \qw            & \qw           & \qw            & \qw           & \qw            & \qw &
\qw &\qw  \vqw{8} & \qw & \qw            & \qw       & \qw  \vqw{8} & \qw & \qw & \qw & \qw & \qw & \qw & \rstick{$\ket{0}$} \qw \\
\lstick{$\ket{0}$} & \qw                                              & \qw &    & \qw &              & & \qw & \qw &\qw            & \qw           & \qw            & \qw           & \qw            & \qw & \swap{5} & \qw & \gate{C_0} \vqw{7}           & \qw & \gate{C_0}  \vqw{7}                & \qw & \qw  & \qw & \qw & \qw & \qw & \qw & \qw & \qw& \qw\\
\lstick{$\ket{0}$} & \qw                                              & \qw &     & \qw & \rb{-1.5cm}{=}              & & \qw & \qw & \qw            & \qw           & \qw            & \qw           & \qw            & \qw &
\qw & \qw & \qw           & \swap{6} & \qw                 & \qw & \qw  & \qw & \qw & \qw & \qw & \qw & \qw& \qw& \qw\\
\lstick{$\ket{\Theta_0}$} & \qw                                              & \qw & \qw   & \qw &               & & \qw & \qw & \swap{}         & \qw & \qw        & \qw & \qw         & \qw &
\qw & \qw & \qw & \qw        & \qw & \qw & \qw  & \qw & \qw & \qw & \qw & \qw & \qw & \qw& \qw\\
 \lstick[wires=2]{$\ket{\Theta_1}$}              & \qw                                              & \qw & \qw   & \qw &               & & \qw & \qw &\qw         & \mltg{2}{S_0} & \swap{}        & \qw  & \qw         & \qw &
\qw & \qw & \qw & \qw         & \qw & \qw & \qw  & \qw & \qw & \qw & \qw & \mltg{2}{S_0} & \qw & \qw& \qw\\
                 & \qw & \qw & \qw   & \qw &               & & \qw & \qw & \qw            & \qw   & \qw            & \qw   & \qw            & \qw &
\qw & \qw & \qw   & \qw            & \qw   & \qw & \qw  & \qw & \qw & \qw & \qw & \qw & \qw & \qw& \qw\\
  \lstick[wires=2]{$\ket{\Theta_2}$}               & \qw                                              & \qw & \qw   & \qw &               &&\qw  & \qw & \qw         & \qw & \qw         & \mltg{2}{S_1} & \mltg{2}{S_0}        & \qw &
\swap{} & \qw & \qw & \qw         & \qw & \qw & \qw  & \qw & \mltg{2}{S_0} & \mltg{2}{S_1} & \qw & \qw & \qw & \qw& \qw\\
               & {\rb{0.1cm}{\shiftright{-0.5cm}{$/^{3}$}}} \qw & \qw & \qw   & \qw &               & & \qw & \qw & \qw            & \qw   & \qw            & \qw   & \qw            & \qw &
\qw & \qw & \qw  & \qw            & \qw   & \qw & \qw  & \qw & \qw & \qw & \qw & \qw & \qw & \qw& \qw\\
     \lstick[wires=2]{$\ket{\Theta_3}$}              & \qw                                              & \qw & \qw  & \qw &               &&\qw & \qw & \qw & \qw & \qw          & \qw & \qw         & \mltg{2}{S_2} &
\qw & \mltg{2}{S_1}& \mltg{2}{S_0} & \swap{}         & \mltg{2}{S_0} & \mltg{2}{S_1} & \qw  & \mltg{2}{S_2} & \qw & \qw & \qw & \qw & \qw & \qw& \qw\\
                 & {\rb{0.13cm}{\shiftright{-0.5cm}{$/^{7}$}}} \qw & \qw & \qw   & \qw &               && \qw & \qw &\qw            & \qw   & \qw            & \qw   & \qw            & \qw &
\qw &\qw& \qw   & \qw            & \qw  & \qw & \qw  & \qw & \qw & \qw & \qw & \qw & \qw & \qw& \qw\\
\end{quantikz}
}
}
\caption{SPF circuit for preparing a state with $m=4$ qubits, which are left entangled with $2^q-1$ qubits (garbage). The circuit has depth $O(m)$ and uses an additional $O(2^m)$ ancillae that begin and end in $\ket{0}$. Thus, the total spacetime allocation is upper bounded by $O(m2^m)$. The second half of the SPF circuit is very similar to the reverse of the first half, except that SWAP gates are not present.}
\label{fig:SPF}
\end{figure}
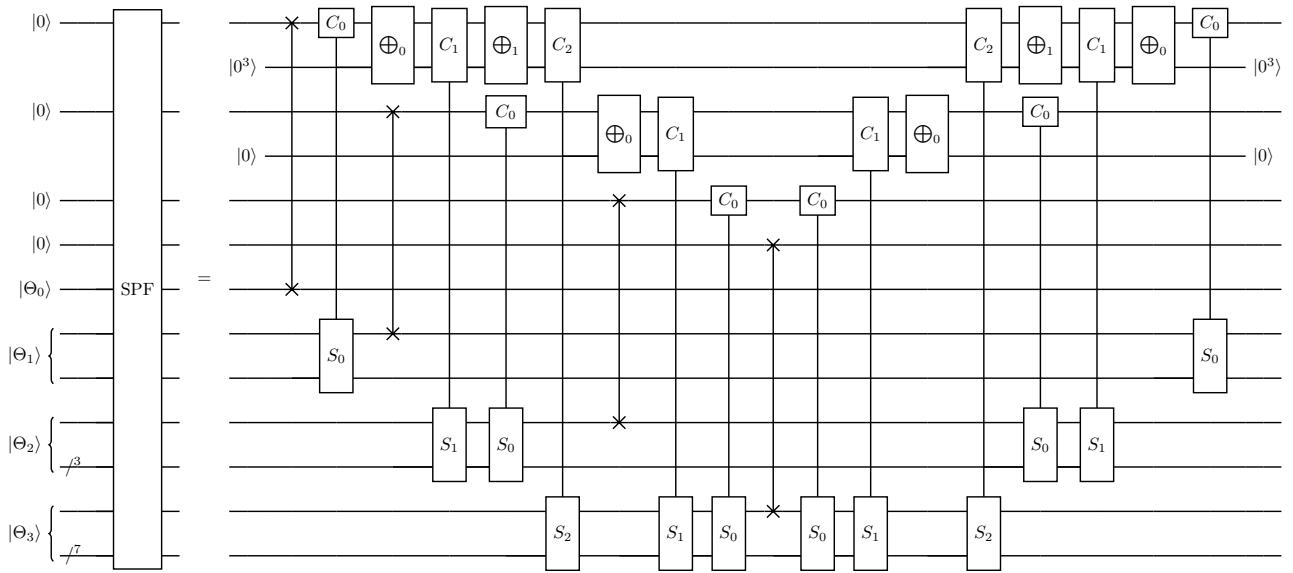

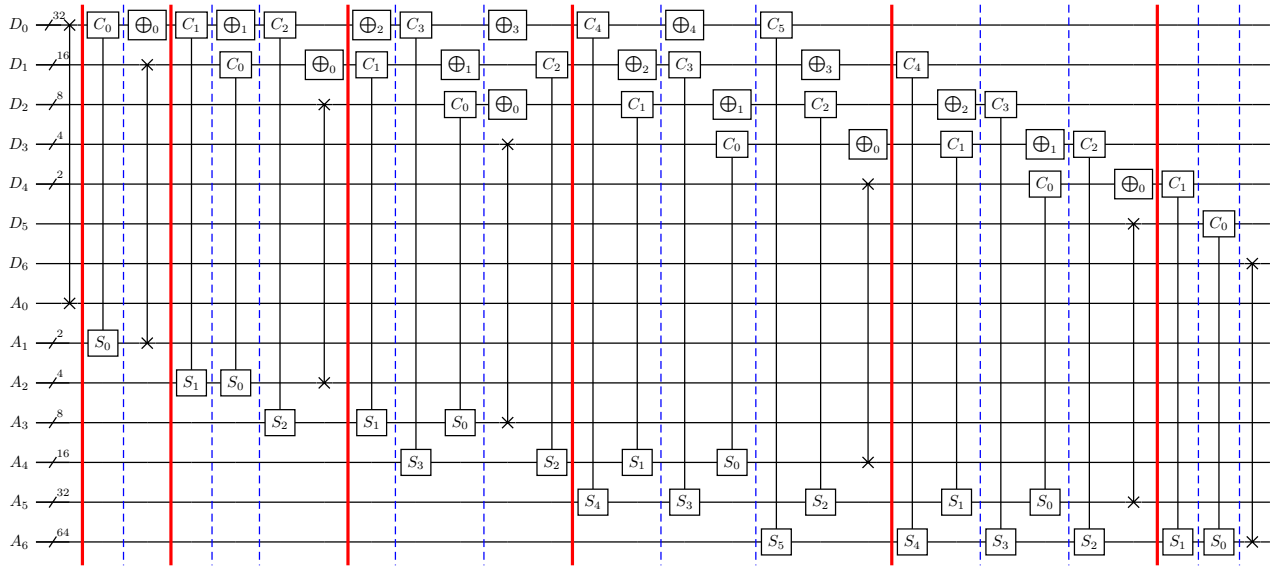
\begin{figure}[h!]
\makebox[\textwidth][c]{
\scalebox{0.6}{
\begin{quantikz}[row sep={2.5em,between origins}, column sep=0.6em, align equals at=1.5]
\lstick{$D_0$} & [1em] \swap{7} \slice[style={solid,line width=0.2em}]{} \qwbundle{32} & \gate{C_0} \vqw{8} \slice[style=blue]{} & \gate{\bigoplus_0} \slice[style={solid,line width=0.2em}]{} & \gate{C_1} \vqw{9} \slice[style=blue]{} & \gate{\bigoplus_1} \slice[style=blue]{} & \gate{C_2} \vqw{10} & \qw \slice[style={solid,line width=0.2em}]{} & \gate{\bigoplus_2} \slice[style=blue]{} & \gate{C_3} \vqw{11} & \qw \slice[style=blue]{} & \gate{\bigoplus_3} & \qw \slice[style={solid,line width=0.2em}]{} & \gate{C_4} \vqw{12} & \qw \slice[style=blue]{} & \gate{\bigoplus_4} & \qw \slice[style=blue]{} & \gate{C_5} \vqw{13} & \qw & \qw \slice[style={solid,line width=0.2em}]{} & \qw & \qw \slice[style=blue]{} & \qw & \qw \slice[style=blue]{} & \qw & \qw \slice[style={solid,line width=0.2em}]{} & \qw \slice[style=blue]{} & \qw \slice[style=blue]{} & \qw & \qw\\
\lstick{$D_1$} & \qw \qwbundle{16} & \qw & \swap{7} & \qw & \gate{C_0} \vqw{8} & \qw & \gate{\bigoplus_0} & \gate{C_1} \vqw{9} & \qw & \gate{\bigoplus_1} & \qw & \gate{C_2} \vqw{10} & \qw & \gate{\bigoplus_2} & \gate{C_3} \vqw{11} & \qw & \qw & \gate{\bigoplus_3} & \qw & \gate{C_4} \vqw{12} & \qw & \qw & \qw & \qw & \qw & \qw  & \qw & \qw & \qw\\
\lstick{$D_2$} & \qw \qwbundle{8}& \qw & \qw & \qw & \qw & \qw & \swap{7} & \qw & \qw & \gate{C_0} \vqw{8} & \gate{\bigoplus_0} & \qw & \qw & \gate{C_1} \vqw{9} & \qw & \gate{\bigoplus_1} & \qw & \gate{C_2} \vqw{10} & \qw & \qw & \gate{\bigoplus_2} & \gate{C_3} \vqw{11} & \qw & \qw & \qw & \qw & \qw & \qw & \qw \\
\lstick{$D_3$} & \qw \qwbundle{4}& \qw & \qw & \qw & \qw & \qw & \qw & \qw & \qw & \qw &  \swap{7} & \qw & \qw & \qw & \qw & \gate{C_0} \vqw{8} & \qw & \qw & \gate{\bigoplus_0} & \qw & \gate{C_1} \vqw{9} & \qw & \gate{\bigoplus_1} & \gate{C_2} \vqw{10} & \qw & \qw & \qw & \qw & \qw\\
\lstick{$D_4$} & \qw \qwbundle{2} & \qw & \qw & \qw & \qw & \qw & \qw & \qw & \qw & \qw & \qw & \qw & \qw & \qw & \qw & \qw & \qw & \qw & \swap{7} & \qw & \qw & \qw & \gate{C_0} \vqw{8} & \qw & \gate{\bigoplus_0} & \gate{C_1} \vqw{9} & \qw & \qw & \qw\\
\lstick{$D_5$} & \qw & \qw & \qw & \qw & \qw & \qw & \qw & \qw & \qw & \qw & \qw & \qw & \qw & \qw & \qw & \qw & \qw & \qw & \qw & \qw & \qw & \qw & \qw & \qw & \swap{7} & \qw & \gate{C_0} \vqw{8} &\qw & \qw\\
\lstick{$D_6$} & \qw & \qw & \qw & \qw &\qw & \qw & \qw & \qw & \qw & \qw & \qw & \qw & \qw & \qw & \qw & \qw & \qw & \qw & \qw & \qw & \qw & \qw & \qw & \qw & \qw & \qw & \qw &\swap{7} & \qw\\
\lstick{$A_0$} & \targX{} & \qw & \qw & \qw & \qw & \qw & \qw & \qw & \qw & \qw & \qw & \qw & \qw & \qw & \qw & \qw & \qw & \qw & \qw & \qw & \qw & \qw & \qw & \qw & \qw & \qw & \qw & \qw & \qw\\
\lstick{$A_1$} & \qw \qwbundle{2} & \gate{S_0} & \swap{} & \qw & \qw & \qw & \qw & \qw & \qw & \qw & \qw & \qw & \qw & \qw & \qw & \qw & \qw & \qw & \qw & \qw & \qw & \qw & \qw & \qw & \qw & \qw & \qw & \qw & \qw\\
\lstick{$A_2$} & \qw \qwbundle{4} & \qw & \qw & \gate{S_1} & \gate{S_0} & \qw & \targX{} & \qw & \qw & \qw & \qw & \qw & \qw & \qw & \qw & \qw & \qw & \qw & \qw & \qw & \qw & \qw & \qw & \qw & \qw & \qw & \qw & \qw & \qw\\
\lstick{$A_3$} & \qw \qwbundle{8} & \qw & \qw & \qw & \qw & \gate{S_2} & \qw & \gate{S_1} & \qw & \gate{S_0} & \targX{} & \qw & \qw & \qw & \qw & \qw & \qw & \qw & \qw & \qw & \qw & \qw & \qw & \qw & \qw & \qw & \qw & \qw & \qw\\
\lstick{$A_4$} & \qw \qwbundle{16} & \qw & \qw & \qw & \qw & \qw & \qw & \qw & \gate{S_3} & \qw & \qw & \gate{S_2} & \qw & \gate{S_1} & \qw & \gate{S_0} & \qw & \qw & \targX{} & \qw & \qw & \qw & \qw & \qw & \qw & \qw & \qw & \qw & \qw\\
\lstick{$A_5$} & \qw \qwbundle{32} & \qw & \qw & \qw & \qw & \qw & \qw & \qw & \qw & \qw & \qw & \qw & \gate{S_4} & \qw & \gate{S_3} & \qw & \qw & \gate{S_2} & \qw & \qw & \gate{S_1} & \qw & \gate{S_0} & \qw & \targX{} & \qw & \qw & \qw & \qw\\
\lstick{$A_6$} & \qw \qwbundle{64} & \qw & \qw & \qw & \qw & \qw & \qw & \qw & \qw & \qw & \qw & \qw & \qw & \qw & \qw & \qw & \gate{S_5} & \qw & \qw & \gate{S_4} & \qw & \gate{S_3} & \qw & \gate{S_2} & \qw & \gate{S_1} & \gate{S_0} &\targX{} & \qw
\end{quantikz}
}
}
\caption{First part of SPF for $m = 7$. The circuit consists of 7 data registers ($D_0$ - $D_6$) and 7 ancilla registers ($A_0$ - $A_6$). Each space between the red line in the circuit represents 1 iteration of the $i$ values in algorithm~\ref{algo:SPF}. There are 3 parallelized gate sequences in each red line space, separated by the blue dashed line. The first parallelized gate sequence corresponds to the first for loop in line 3 - line 9. The second parallelized gate sequence corresponds to the second for loop in line 10 - line 16. The third parallelized gate sequence corresponds to the third for loop in line 17 - line 25. (Notice that we might not have all 3 sequences in all of the boxes based on the condition, for instance, at the beginning of the circuit)}
\label{fig:first_part_SPF}
\end{figure}

\begin{algorithm}[H]
    \caption{SPF Subroutine}\label{algo:SPF}
    \begin{algorithmic}[1]
    \Procedure{SPF}{$D, A, m$}
        
        \Comment{$D$ is the data register of size $m$}

        \Comment{$A = A_0,\ldots,A_{m-1}$ is the angle register where $A_j$ is size $2^j$}
        \For{$i$ \textbf{in} range($m$)} \Comment{Each value of $i$ occupies $O(1)$ depth, shown as one box in Fig.~\ref{fig:first_part_SPF}}
            \For{$q$ \textbf{in} range($i-1$)} \Comment{$1^{\mathrm{st}}$ parallelized sequence in the space between red lines in Fig.~\ref{fig:first_part_SPF}.}
                \If{$(i-q)$ is odd \textbf{and} $\frac{3(i-q-1)}{2}-1 \leq m-q-3$}
                    \State $\bigoplus_{\frac{3(i-q-1)}{2}-1}(D_q)$
                \ElsIf{$(i-q)$ is even \textbf{and} $\frac{3(i-q)}{2}-2 \leq m-q-2$}
                    \State CS$_{\frac{3(i-q)}{2}-2}(D_q, A_{i-1+\frac{i-q}{2}})$
                \EndIf
            \EndFor
            \For{$q$ \textbf{in} range($i$)} \Comment{$2^{\mathrm{nd}}$ parallelized sequence in the space between red lines in Fig.~\ref{fig:first_part_SPF}.}
                \If{$(i-q)$ is odd \textbf{and} $\frac{3(i-q-1)}{2} \leq m-q-2$}
                    \State CS$_{\frac{3(i-q-1)}{2}}(D_q, A_{i+\frac{i-q-1}{2}})$
                \ElsIf{$(i-q)$ is even \textbf{and} $\frac{3(i-q)}{2}-2 \leq m-q-3$}
                    \State $\bigoplus(\frac{3(i-q)}{2}-2, D_q)$
                \EndIf
            \EndFor
            \For{$q$ \textbf{in} range($i+1$)} \Comment{$3^{\mathrm{rd}}$ parallelized sequence in the space between red lines in Fig.~\ref{fig:first_part_SPF}.}
                \If{$i==q$}
                    \State $\mathrm{SWAP}(D_{q,0},A_{i,0})$
                \ElsIf{$(i-q)$ is odd \textbf{and} $\frac{3(i-q-1)}{2} \leq m-q-3$}
                    \State $\bigoplus_{\frac{3(i-q-1)}{2}}(D_q)$
                \ElsIf{$(i-q)$ is even \textbf{and} $\frac{3(i-q)}{2}-1 \leq m-q-2$}
                    \State CS$_{\frac{3(i-q)}{2}-1}(D_q, A_{i+\frac{i-q}{2}})$
                \EndIf
            \EndFor
        \EndFor
        \For{$i$ \textbf{in} reversed(range($1, m$))} \Comment{Each value of $i$ occupies $O(1)$ depth, run in reversed order}
            \For{$q$ \textbf{in} range($i + 1$)} \Comment{All values of $q$ performed in parallel}
                \If{$i == q$}
                    \State continue
                \ElsIf{$(i-q)$ is odd \textbf{and} $\frac{3(i - q - 1)}{2} \leq m - q - 3$}
                    \State $\bigoplus_{\frac{3(i - q - 1)}{2}}(D_q)$
                \ElsIf{$(i-q)$ is even \textbf{and} $\frac{3(i-q)}{2}\leq m-q-2$}
                    \State CS$_{\frac{3(i-q)}{2} - 1}(D_q, A_{i + \frac{i-q}{2}})$
                \EndIf
            \EndFor
            \For{$q$ \textbf{in} range($i$)} \Comment{All values of $q$ performed in parallel}
                \If{$(i-q)$ is odd \textbf{and} $\frac{3(i-q-1)}{2}\leq m-q-2$}
                    \State CS$_{\frac{3(i-q-1)}{2}}(D_q, A_{i + \frac{i-q-1}{2}})$
                \ElsIf{$(i-q)$ is even \textbf{and} $\frac{3(i-q)}{2}-2 \leq m-q-3$}
                    \State $\bigoplus_{\frac{3(i-q)}{2}-2}(D_q)$
                \EndIf
            \EndFor
            \For{$q$ \textbf{in} range($i-1$)} \Comment{All values of $q$ performed in parallel}
                \If{$(i-q)$ is odd \textbf{and} $\frac{3(i-q-1)}{2}-1 \leq m-q-3$}
                    \State $\bigoplus_{\frac{3(i-q-1)}{2}-1}(D_q)$
                \ElsIf{$(i-q)$ is even \textbf{and} $\frac{3(i-q)}{2}-2\leq m-q-2$}
                    \State CS$_{\frac{3(i-q)}{2}-2}(D_q, A_{i-1+\frac{i-q}{2}})$
                \EndIf \Comment{Total D$_{\mathrm{exact}}$: $O(m)$}
            \EndFor \Comment{Total D$_{\mathrm{approx}}$: $O(m)$}
        \EndFor \Comment{Total SA$_{\mathrm{exact}}$: $O(m2^m)$}
    \EndProcedure \Comment{Total SA$_{\mathrm{approx}}$: $O(m2^m)$}
    \end{algorithmic}
\end{algorithm}

\subsection{FLAG}\label{sec:FLAG_circ}

The implementation of the FLAG circuit is similar to that of the SPF circuit, with a few simplifications. Here, the data qubits are only acting as controls, and there is no injection of angles into the data qubits. Thus, we do not need to alternate between copying layers and CSWAP layers; we can simply perform all the copying at the beginning, then all the CSWAPs, and then all the uncopying.  The conceptual idea behind the FLAG implementation is that a flag is set in the $p^{\mathrm{th}}$ flag qubit by first flagging qubit 0 and then swapping qubit 0 into the $p^{\mathrm{th}}$ position using a sequence of CSWAP layers with different data qubits acting as the control. We give an example of FLAG for $m=4$ in Fig.~\ref{fig:FLAG} and pseudocode for FLAG in Algorithm \ref{algo:FLAG}.  

\begin{figure*}[h!]
\centering
\scalebox{0.75}{
\begin{quantikz}[row sep={2.5em,between origins}, column sep=1em, align equals at=1.5]
 & \qw & \ctrlslash{2} & \qw && & \qw & \qw & \mltg{2}{\bigoplus_0} \slice[style=blue]{} & \mltg{2}{\bigoplus_1} \slice[style=blue]{} & \qw & \qw & \gate{C_0}\vqw{7} \slice[style=blue]{} &\qw & \mltg{2}{C_1} \slice[style=blue]{} & \mltg{2}{C_2} \slice[style=blue]{} &\mltg{2}{\bigoplus_1}\slice[style=blue]{} &\mltg{2}{\bigoplus_0} & \qw& \qw& \qw\\
                  & &   & &               && &\lstick{$\ket{0^3}$}&\qw        & \qw  & \qw & \qw & \qw & \qw & \qw \vqw{8} & \qw \vqw{10}&\qw &\qw & \rstick{$\ket{0^3}$} \qw\\
                  & \qw & \ctrlslash{2}  & \qw &               & & \qw & \qw          & \mltg{2}{\bigoplus_0}       & \qw  & \qw  & \gate{C_0} \vqw{7} & \qw & \mltg{2}{C_1} & \qw &\qw & \qw& \mltg{2}{\bigoplus_0} & \qw & \qw& \qw\\
&&&&   &&&\lstick{$\ket{0}$}& \qw           & \qw         & \qw & \qw & \qw & \qw \vqw{8} & \qw &\qw &  \qw  & \qw & \rstick{$\ket{0}$} \qw \\
    & \qw & \ctrlslash{1}   & \qw &  \rb{-3cm}{=}             & & \qw &\qw            &\qw & \qw & \gate{C_0}\vqw{7}                 & \qw & \qw & \qw & \qw & \qw & \qw & \qw& \qw & \qw& \qw\\
       & \qw & \ctrlslash{1}   & \qw &               & & \qw & \qw            & \qw                 & \qw & \qw  & \qw & \qw & \qw & \qw & \qw & \qw& \qw& \qw & \qw& \qw\\
         & \qw & \mltg{7}{\text{FLAG}}   & \qw &               & & \qw & \qw & \gate{X}          & \qw & \qw  & \qw & \qw & \qw & \qw & \qw & \qw & \qw& \qw & \qw& \qw\\
    & \qw & \qw   & \qw &               & & \qw & \qw &\gate{X}         & \qw & \qw  & \qw & \mltg{2}{S_0}  & \qw & \qw & \qw & \qw & \qw& \qw & \qw& \qw\\
  & \qw & \qw   & \qw &               & & \qw & \qw & \qw                & \qw & \qw & \qw & \qw &\qw & \qw & \qw & \qw & \qw& \qw & \qw& \qw\\
   & \qw & \qw   & \qw &               &&\qw  & \qw & \gate{X}         & \qw & \qw  & \mltg{2}{S_0} & \qw & \qw & \mltg{2}{S_1} & \qw & \qw & \qw& \qw & \qw& \qw\\
                 & {\rb{0.1cm}{\shiftright{-0.5cm}{$/^{3}$}}} \qw & \qw   & \qw &               & & \qw & \qw & \qw            & \qw   & \qw  & \qw & \qw & \qw & \qw & \qw & \qw & \qw& \qw & \qw& \qw\\
      & \qw & \qw   & \qw &               &&\qw & \qw & \gate{X}         & \qw & \mltg{2}{S_0}  & \qw & \qw & \mltg{2}{S_1} & \qw & \mltg{2}{S_2} & \qw & \qw& \qw & \qw& \qw\\
                 & {\rb{0.13cm}{\shiftright{-0.5cm}{$/^{7}$}}} \qw & \qw   & \qw &               && \qw & \qw &\qw             & \qw & \qw  & \qw & \qw & \qw & \qw & \qw & \qw & \qw& \qw & \qw& \qw
\end{quantikz}
}
\caption{FLAG circuit for $m=4$. All the gate sequences within each of the two blue dashed lines can be executed in parallel.}
\label{fig:FLAG}
\end{figure*}
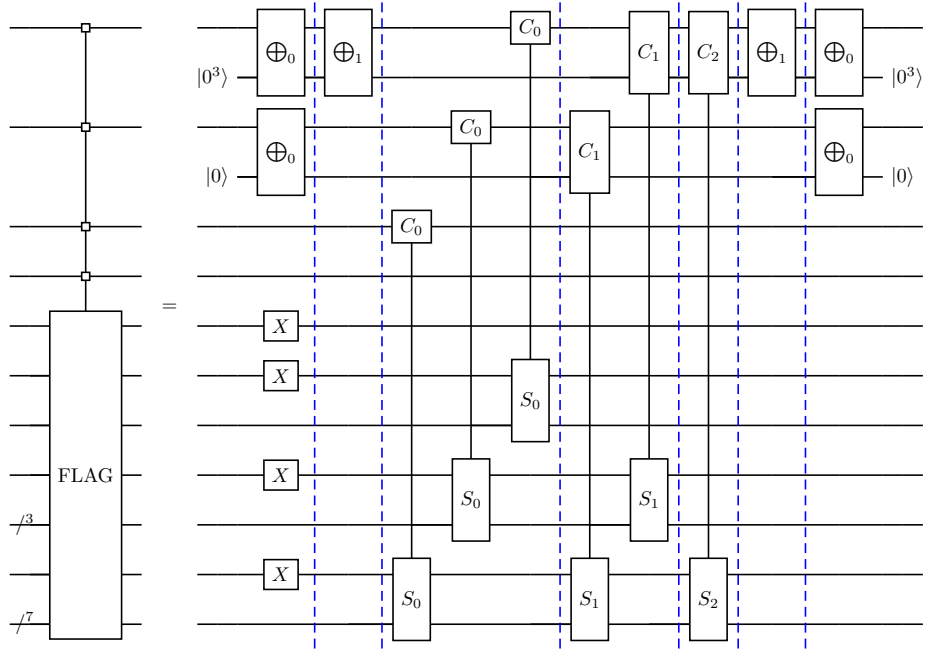

\begin{algorithm}[H]
    \caption{FLAG Operation}\label{algo:FLAG}
    \begin{algorithmic}[1]
    \Procedure{FLAG Subroutine}{$D,F,m$}
    
    \Comment{$D = D_0,\ldots,D_{m-1}$ is the data register, where $D_j$ has size $\max(1,2^{m-j-2})$}
    
    \Comment{$F = F_0,\ldots,F_{m-1}$ is the flag register where $F_j$ has size $2^j$}
    \For{$j$ \textbf{in} range($m$)} \Comment{All values of $j$ performed in parallel}
        \State $\mathrm{X}(F_{j,0})$
    \EndFor
    \For{$i$ \textbf{in} range ($m-1$)} \Comment{All values of $i$ performed in parallel}
        \State \textsc{COPY}($D_i$) \Comment{Occupies maximum of $O(m)$ depth}
    \EndFor
    \For{$i$ \textbf{in} range($m-1$)} \Comment{Each value of $i$ occupies $O(1)$ depth}
        \For{$q$ \textbf{in} range($m-i-1$)}
            \State $\mathrm{CS}_i(D_q, F_{q+1+i})$
        \EndFor
    \EndFor
    \For{$i$ \textbf{in} range ($m-1$)} \Comment{All values of $i$ performed in parallel}
        \State \textsc{COPY}${}^\dagger$($D_i$) \Comment{Occupies maximum of $O(m)$ depth}
    \EndFor
    \EndProcedure
    \end{algorithmic}
    \Comment{Total D$_{\mathrm{exact}}$: $O(m)$}

    \Comment{Total D$_{\mathrm{approx}}$: $O(m)$}
    
    \Comment{Total SA$_{\mathrm{exact}}$: $O(m2^m)$}

    \Comment{Total SA$_{\mathrm{approx}}$: $O(m2^m)$}
\end{algorithm}

\newpage



\subsection{Copy Swap Operation}\label{sec:CopySwap}
To simplify the circuit logic of the LOADF subroutines, we define another subroutine operation $\overline{\mathrm{CopySwap}}$. To define $\overline{\mathrm{CopySwap}}$, let $\ket{k}$ be an $m$-qubit computational basis state, written in binary as $k_{m-1}k_{m-2}\ldots k_0$, so that $k_0$ represents the least significant bit. Let $\ket{\xi}$ be an arbitrary single-qubit state. The operation $\overline{\mathrm{CopySwap}}$ enacts the isometry from an $m+1$ qubit state to a $2M-1$ qubit space
\begin{equation}
    \overline{\mathrm{CopySwap}}\left(\ket{k}\ket{\xi}\right) = \ket{k_{m-1}}^{\otimes2^{m-1}}\ket{k_{m-2}}^{\otimes 2^{m-2}} \ldots \ket{k_1}^{\otimes 2} \ket{k_0}\ket{0^k}\ket{\xi}\ket{0^{M-k-1}}
\end{equation}
for any $m$-qubit computational basis state $\ket{k}$ and any arbitrary single-qubit state $\ket{\xi}$. Implementing $\overline{\mathrm{CopySwap}}$ efficiently will involve a combination of copying layers and layers of parallel CSWAPs, depicted in Fig.~\ref{fig:CopySwap} using the $\bigoplus_t$ and CS$_t$ subroutines defined in App.~\ref{sec:data_copy_circ}. Similar to previous sections, the goal of performing many $\bigoplus_t$ operations before the corresponding CS$_t$ sequences at the earliest time is to ensure the CS$_t$ sequences can be maximally parallelized under 1 step.

\begin{figure*}[h!]
\centering
\begin{minipage}{.5\textwidth}
  \centering
  \begin{quantikz}[row sep={1.5em,between origins}, column sep=1em, align equals at=2.5]
        \lstick{$\ket{k}$} &  \qwbundle{m} &\gate[2,nwires={2}]{\overline{\mathrm{Copy}}}\vqw{2} & [0.5em] \qwbundle{m} & [1.5em]\rstick[2]{$ \ket{k_{m-1}}^{\otimes 2^{m-1}}\ldots \ket{k_0}$} \qw\\
        & & &\qwbundle{M-m-1} & \qw \\[1.5em]
        \lstick{$\ket{\xi}$} & \qw &\gate[2,nwires={2}]{\overline{\mathrm{Swap}}} & \qw  & \rstick[2]{$\ket{0^k}\ket{\xi}\ket{0^{M-k-1}}$} \qw \\
        &&& \qwbundle{M-1} & \qw 
    \end{quantikz}
  \label{fig:test1}
\end{minipage}%
\begin{minipage}{.5\textwidth}
  \centering
  \scalebox{0.8}{
  \begin{quantikz}[row sep={1.8em,between origins}, column sep=2.2em, align equals at=1]
    \lstick{$\overline{R_4}$} & \qw&\gate{\bigoplus_0} & \gate{\bigoplus_1} \qwbundle{2} & \gate{\bigoplus_3}\qwbundle{4} & \gate{\bigoplus_4} \qwbundle{8}  & \gate{C_4} \vqw{5} \qwbundle{16} & \qw \\[0.5em]
     \lstick{$\overline{R_3}$} & \qw&\gate{\bigoplus_0} & \gate{\bigoplus_1} \qwbundle{2} & \gate{\bigoplus_3}\qwbundle{4} & \gate{C_3} \vqw{4} \qwbundle{8}  & \qw & \qw \\[0.5em]
    \lstick{$\overline{R_2}$} & \qw&\gate{\bigoplus_0} & \gate{\bigoplus_1} \qwbundle{2} & \gate{C_2} \vqw{3} \qwbundle{4} & \qw  & \qw & \qw \\[0.5em]
    \lstick{$\overline{R_1}$} & \qw&\gate{\bigoplus_0} & \gate{C_1}\vqw{2}\qwbundle{2} & \qw & \qw  & \qw & \qw \\[0.5em]
    \lstick{$\overline{R_0}$} & \qw&\gate{C_0} \vqw{1}& \qw & \qw & \qw  & \qw & \qw \\[1.5em]
    \qw & \qw & \gate[2]{S_0}& \gate[3]{S_1}  & \gate[4]{S_2} & \gate[5]{S_3} & \gate[6]{S_4} & \qw \\
        & \lstick{$\ket{0}$} & & \qw & \qw & \qw & \qw & \qw \\
        &  & \lstick{$\ket{0^2}$}& \qwbundle{2} & \qwbundle{2} & \qwbundle{2} & \qwbundle{2} & \qwbundle{2}\\
        \lstick{$\overline{S}$}  & & & \lstick{$\ket{0^4}$} & \qwbundle{4} & \qwbundle{4} & \qwbundle{4} & \qwbundle{4} \\
        &  & &  & \lstick{$\ket{0^8}$} & \qwbundle{8} & \qwbundle{8} & \qwbundle{8} \\
        &  & &  & &  \lstick{$\ket{0^{16}}$} & \qwbundle{16} & \qwbundle{16} 
    \end{quantikz}
    }
\end{minipage}
\caption{\label{fig:CopySwap}Left: action of 
$\overline{\mathrm{CopySwap}}$ operation, which simultaneously copies $m$ control bits set to $\ket{k}$ and moves a target register into the $k^{\mathrm{th}}$ position of a target register. The input is $m+1$ qubits and the output is $2M-1$ qubits, as several fresh ancillae are introduced during the protocol. Right: implementation of $\overline{\mathrm{CopySwap}}$ operation using copying layers and swap operations for $m=5$. The total depth is $m$. At layer $t=0,\ldots,m-1$, there are $(m+(2^t-1)(m-t)+2^{t+1})$ active qubits so the total spacetime allocation is $\sum_{t=0}^{m-1}(m+(2^t-1)(m-t)+2^{t+1}) = O(2^m)$.
}
\end{figure*}
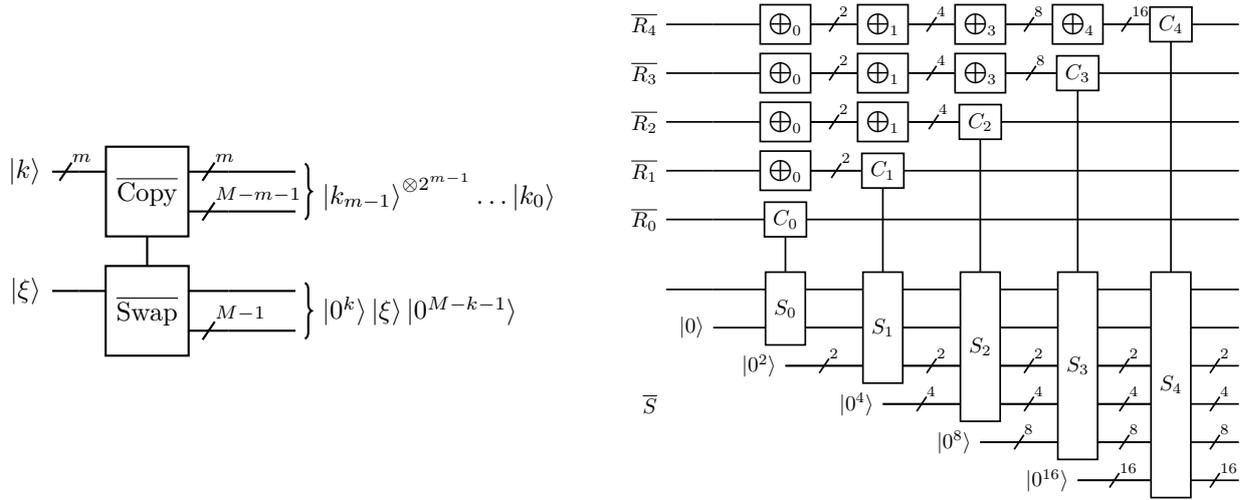

\begin{algorithm}[H]
    \caption{Copy Swap Operation}\label{algo:copy_swap}
    \begin{algorithmic}[1]
    \Procedure{$\overline{\mathrm{CopySwap}}$ Subroutine}{$\overline{R},\overline{S},m$}

    \Comment{$m$ is the number of control bits}
    
    \Comment{$\overline{R}$ = $\overline{R_0}$$,\ldots,$$\overline{R_{m-1}}$ is the control register, where $\overline{R_j}$ has size $\max(1,2^{m-j-2})$}
    
    \Comment{$\overline{S}$ is the target register of size $2^m$}
    \For{$i$ \textbf{in} range($m$)} \Comment{$O(m)$}
        \For{j in range(1, $m$ - $i$)}\Comment{All values of $j$ performed in parallel with the CS operation}
            \State$\bigoplus_i(\overline{R_j})$
        \EndFor
        \State CS$_i(\overline{R_i}, \overline{S})$
    \EndFor
    \EndProcedure
    \end{algorithmic}
    \Comment{Total D$_{\mathrm{exact}}$: $O(m)$}

    \Comment{Total D$_{\mathrm{approx}}$: $O(m)$}
    
    \Comment{Total SA$_{\mathrm{exact}}$: $O(m2^m)$}

    \Comment{Total SA$_{\mathrm{approx}}$: $O(m2^m)$}
\end{algorithm}

\newpage

\subsection{LOADF}\label{sec:app_LOADF}

The LOADF circuit is perhaps the key to making the whole protocol work. It has two control registers, a data register of size $m$ and a flag register of size $\tilde{B}:= N/M-1$. If the data register is set to $\ket{k}$, LOADF reads in the $\tilde{B}$ angle states needed to create the associated $(n-m)$-qubit state, but the angle associated with indices $(s,p)$ is only read in if the flag in position $(s,p)$ is set to 1. Mathematically, LOADF achieves the following operation:
\begin{equation}\label{eq:LOADF2}
    \mathrm{LOADF}\left(\ket{k}\ket{0^{N/M-1}}\left[\bigotimes_{s=0}^{m-n-1}\bigotimes_{p=0}^{2^s-1} \ket{f_{s,p}}\right]\right)
    = \ket{k}\left[\bigotimes_{s=0}^{m-n-1}\bigotimes_{p=0}^{2^s-1} \ket{f_{s,p}\theta^{(k)}_{s,p}}\right]\left[\bigotimes_{s=0}^{m-n-1}\bigotimes_{p=0}^{2^s-1} \ket{f_{s,p}}\right]
\end{equation}

One key feature of LOADF is that all of the single-qubit rotation gates occur in the same layer of the circuit. This allows it to achieve $O(\log(N) + \log(1/\epsilon))$ depth in the approximate setting. Another feature is that the $O(N)$ ancilla qubits are needed only very briefly to act as copied controls for doubly controlled rotation gates. As soon as the rotation gates are finished, the ancilla qubits can be rapidly deallocated, allowing the protocol to achieve $O(N)$ spacetime allocation. The protocol for LOADF is depicted in Fig.~\ref{fig:LOADF}, with the corresponding pseudocode written in Algorithm \ref{algo:LOADF}.

\begin{figure}[h]
\begin{subfigure}[b]{\textwidth}
\centering
\scalebox{0.7}{
\begin{quantikz}[row sep={2.5em,between origins}, column sep=0.5em, align equals at=2]
    \lstick{$D_0:\ket{k}$} & [1em]\qwbundle{m} & [1em] \ctrlslash{1} & \rstick{$\ket{k}$}\qw \\
    \lstick{$B_0:\ket{0^{\tilde{B}}}$} & \qwbundle{\tilde{B}} & \gate{\substack{\text{LOADF} \\ \{\theta^{(k)}\}}} & \rstick{$\bigotimes_{s,p} \mathrm{R}_y(f_{s,p}\theta_{s,p}^{(k)})\ket{0}$} \qw \\
    \lstick{$F_0:\bigotimes_{s,p}\ket{f_{s,p}}$} & \qwbundle{\tilde{B}} & \ctrlslash{-1} & \rstick{$\bigotimes_{s,p}\ket{f_{s,p}}$}\qw
\end{quantikz}
\qquad\;=
\begin{quantikz}[row sep={2.5em,between origins}, column sep=0.5em, align equals at=2]
    \lstick{$D:\ket{k}\otimes\ket{0^{O(N)}}$} & \qw & \gate[4]{\mathrm{Setup}} &\qw & \gate[4]{\mathrm{Setup^{\dagger}}} &\qw &\rstick{$\ket{k}\otimes\ket{0^{O(N)}}$} \\
    \lstick{$B:\ket{0^{\tilde{B}}}\otimes \ket{\varphi^{O(N)}}$} & \qw & \qw & \gate[3]{\mathrm{CCRy}} & \qw &\qw \rstick{$\bigotimes_{s,p} \mathrm{R}_y(f_{s,p}\theta_{s,p}^{(k)})\ket{0} \otimes \ket{\varphi^{O(N)}}$} \\
    \lstick{$A:\ket{0^{O(N)}}$} & \qw & \qw & \qw &\qw &\rstick{$\ket{0^{O(N)}}$}\qw \\
    \lstick{$F:\bigotimes_{s,p}\ket{f_{s,p}}\otimes\ket{0^{O(N)}}$} & \qw & \qw & \qw &\qw & \rstick{$\bigotimes_{s,p}\ket{f_{s,p}}\otimes\ket{0^{O(N)}}$}\qw
\end{quantikz}
}
\end{subfigure}

\vspace{2.5em}
\begin{subfigure}[b]{\textwidth}
\makebox[\textwidth][c]{
\centering
\scalebox{0.63}{
\begin{quantikz}[row sep={1.6em,between origins}, column sep=0.5em, align equals at=1.5]
\lstick{$D_0: \ket{k}$} &\qwbundle{m}&\qw&\qw&\gate[3,nwires={2}]{\mathrm{COPY}_{\tilde{B}+1}^{\otimes m}} &[2em] \qwbundle{m}&\gate[2]{\overline{\mathrm{Copy}}} \vqw{6}              &[4em]\qw       &  \qw   &[2em]\qwbundle{m}  &[3em] \qw   
&\qw          &[3em]\qw          & \gate[2]{\overline{\mathrm{Copy}}^\dagger} \vqw{6}    & \gate[3,nwires={2}]{\mathrm{COPY}_{\tilde{B}+1}^{\dagger \otimes m}}  & \qw  & \qwbundle{m}                 & \qw  &  \rstick{ $\ket{k}$}               \\
%
   &&&&& \lstick{$D_1$}  &             
&\qw & \qw  & \qwbundle{M-m-1} &\qw & \qw & \qw & \qw & &                &            &      \\[1em]
%
   &&&\lstick{$D_2$}&&\qwbundle{\tilde{B}m}&  \qw               
& \qw &\gate[2]{\overline{\mathrm{Copy}}^{\otimes \tilde{B}}} \vqw{2}  & \qwbundle{\tilde{B}m} &\qw & \gate[2]{\overline{\mathrm{Copy}}^{\dagger\otimes \tilde{B}}} \vqw{2} & \qwbundle{\tilde{B}m} &\qw & &&&      \\
%
   &&&&&&                 
&  \lstick{$D_3$}& & \qwbundle{\tilde{B}(M-m-1)} &\qw &  & & & &                &             &      \\[1em]
%
    \lstick{$B_0: \ket{0^{\tilde{B}}}$}\qw&\qwbundle{\tilde{B}}&\qw&\qw&\qw                   & \qw    & \qw               &\qw      & \mltg{2}{\overline{\mathrm{Swap}}^{\otimes \tilde{B}}}            &\qwbundle{\tilde{B}}  & \mltg{2}{\bigotimes_{s,p}\bigotimes_{j=0}^{M-1}\mathrm{R}_y(\theta^{(j)}_{s,p})}       & \mltg{2}{\overline{\mathrm{Swap}}^{\dagger\otimes \tilde{B}}}  &\qwbundle{\tilde{B}}               &\qw         & \qw \rstick{$\bigotimes_{s,p} \mathrm{R}_y(f_{s,p}\theta_{s,p}^{(k)})\ket{0}$}   \\ [1em]
%
     && && &   & & \lstick{$B_1$} & & \qwbundle{(M-1)\tilde{B}} &      
&                                 &&        &&&& \\ [1em]
%
                  &&&&\lstick{$A_0: \ket{0}$}                   &\gate{X}& \gate[2]{\overline{\mathrm{Swap}}}  &\qw& \gate[3] {\mathrm{COPY}_{\tilde{B}}^{\otimes M}} & \qw    & \ctrlslash{-1}
&\mltg{3}{\mathrm{COPY}_{\tilde{B}}^{\dagger\otimes M}}  & \qw& \gate[2]{\overline{\mathrm{Swap}}^\dagger}      & \gate{X}                                 & \rstick{$\ket{0}$}\qw    &&               \\
%
                  &&&&                   &   \lstick{$A_1$} &  &\qwbundle{M-1}&   & \qwbundle{M-1} &\ctrlslash{-1} 
&     & \qwbundle{M-1} &  \\ [1em]
%
                   &&&&                   &&   &\lstick{$A_2$}&   & \qwbundle{M(\tilde{B}-1)} &\ctrlslash{-1} 
&     & &                                &     &&&               \\ [1em]
%
%
%
\lstick{$F_0: \bigotimes_{s,p}\ket{f_{s,p}}$} \qw&\qwbundle{\tilde{B}}&\qw&\qw&\qw                   & \qw    & \qw      &\qw &\gate[2]{\mathrm{COPY}_M^{\otimes \tilde{B}}}                           & \qwbundle{\tilde{B}} & \ctrlslash{-1} 
& \gate[2]{\mathrm{COPY}_M^{\dagger\otimes \tilde{B}}} &\qwbundle{\tilde{B}}                                 & \qw    &\qw \rstick{$\bigotimes_{s,p}\ket{f_{s,p}}$}                     \\
%
     &&&&                  & &                &\lstick{$F_1$}& &   \qwbundle{(M-1)\tilde{B}}            & \ctrlslash{-1} 
&         & &   & & &                             &   
\end{quantikz}
}
}
\end{subfigure}
\caption{Top: Modular abstraction of LOADF. Setup is used to prepare the expanded address based on the address $\ket{k}$. Then CCRy gates pump in the angles using the expanded address in parallel. Lastly, Setup$^{\dagger}$ uncompute all the expanded address. Bottom: Circuit implementing LOADF. The numerical value $\tilde{B}$ is defined as $\tilde{B}:= N/M-1$. The layer of rotation gates represents $M\tilde{B} = N-M$ parallel $\mathrm{CCR}_y$ gates by different angles $\theta_{s,p}^{(k)}$, with $k=0,\ldots,M-1$, $s=0,\ldots,n-m-1$, and $p=0,\ldots,2^s-1$. For each $\mathrm{CCR}_y$ gate, one qubit from the final two registers (corresponding to the choice of $(s,p)$) and one qubit from the three registers above the final two registers (corresponding to the choice of $k$) acts as a control. Copying layers before the $\mathrm{CCR}_y$ sequence is necessary such that these can all be performed in parallel. After the $\mathrm{CCR}_y$ sequence, we perform $\tilde{B}$ $\overline{\mathrm{CopySwap}}^{\dagger}$ operations in parallel in order to load the $N - M$ angles into the buffer ancilla of size $\tilde{B}$. It can be verified that LOADF achieves $O(N)$ spacetime allocation: the only steps where $O(N)$ ancillae are active are the rotation gates (depth $O(1)$) and the uncopying step that follows. However, the uncopying step for $O(N)$ qubits was seen in Figs.~\ref{fig:CopySwap} and \ref{fig:State_Copying_Subroutine} to occupy only $O(N)$ spacetime allocation. Unless labeled otherwise, all ancilla qubits begin and end with the $\ket{0}$ states. We also note that register B can be mostly allocated with dirty qubits $\ket{\varphi^{O(N)}}$ \cite{low2018trading} instead of qubits in the $\ket{0}$ states.}
\label{fig:LOADF}
\end{figure}
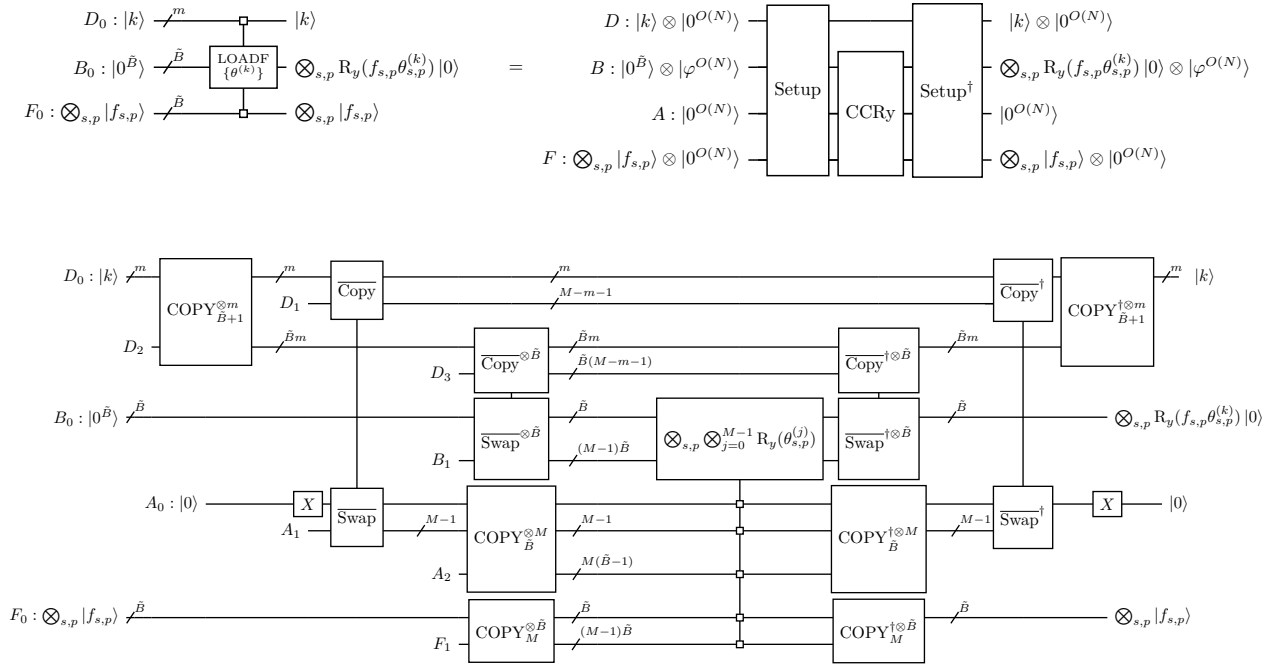

\begin{algorithm}[H]
    \caption{LOADF Subroutine}\label{algo:LOADF}
    \begin{algorithmic}[1]
    \Procedure{LOADF Subroutine}{$D,B,F$}

    \Comment{This subroutine will also involve additional O(N) ancilla qubits. They do not need to be allocated in each ancilla group until individually called on. Once uncomputed, they can be deallocated (i.e., able to participate in a different computing task).}
        
    \Comment{$D_0$ is the address register of size $m$. $D_1$, $D_2$, $D_3$ contain $O(N)$ additional ancillae needed to complete all the parallelized CSWAP and CCRy gates.}
        
    \Comment{$B_0$ is the ``buffer" ancilla register of size $\tilde{B}$ to hold the angles for SPF to pump into the data registers. $B_1$ contains $O(N)$ additional ancillae needed to complete all the parallelized CSWAP and CCRy gates.}
        
    \Comment{$F_0$ is the FLAG register of size $2^{n-m}-1$ (only required for LOADF$^{\dagger}$, optional for LOADF. Since the FLAG bits are all in the $\ket{1}$ states and thus guaranteed the control requirement for the following CCRy gates, we can effectively treat them as CRy gates only controlled on the A register.) $F_1$ contains $O(N)$ additional ancillae needed to complete all the parallelized CSWAP and CCRy gates.}

    \Comment{$A$ is the register of size $M - N$ to hold the $M - N$ expanded angle address (Not shown in Fig.~\ref{fig:CSP}).}

    \Comment{For the purpose of illustration, we ignore the particular indices in the subroutine parameters. We provide exact circuit-level implementations in \ref{sec:code} with all the indices specified.} 

    \State \textbf{Setup} Subroutine
    
    \For{$i$ in range($N - M$)} \Comment{All $i$ in parallel, to load $N - M$ angles into $N - M$ qubits}
        \State CCRy($F, A, B, \mathrm{\theta}$) \Comment{$O(1) (O(\log(n/\epsilon)))$ depth}
    \EndFor
    \State \textbf{Setup$^{\dagger}$} Subroutine
    \EndProcedure

    \Procedure{$\mathrm{Setup}$ Subroutine}{D, B, A, F}
    \State  $\mathrm{X(A)}$ \Comment{$O(1)$ depth}
        
    \For{$i$ in range(m)} \Comment{all $i$ in parallel}
        \State COPY$^{\dagger}_{\tilde{B} + 1}(D)$ \Comment{$O(n - m)$ depth}
    \EndFor
        
    \State $\overline{\mathrm{CopySwap}}(D, A, m)$ \Comment{$O(m)$ depth}
        
    \Comment{The following 3 for loops can be done in parallel \hspace{9.5cm}}
        
    \For{$i$ in range($\tilde{B}$)} \Comment{all $i$ in parallel}
        \State $\overline{\mathrm{CopySwap}}(D, B, n - m)$ \Comment{$O(n - m)$ depth}
    \EndFor
    \For{$i$ in range($M$)} \Comment{all $i$ in parallel, to prepare control bits in the CCRy layer's control on the A registers}
        \State COPY$_{\tilde{B}}$($A$) \Comment{$O(n - m)$ depth}
    \EndFor
    \For{$i$ in range($\tilde{B}$)} \Comment{all $i$ in parallel, to prepare control bits in the CCRy layer's control on the F register}
        \State COPY$_M$($F$) \Comment{$O(m)$ depth}
    \EndFor 
    \EndProcedure \Comment{Total D$_{\mathrm{exact}}$: $O(n)$}
    
    \end{algorithmic}

    \Comment{Total D$_{\mathrm{approx}}$: $O(n + \log(1/\epsilon))$}
    
    \Comment{Total SA$_{\mathrm{exact}}$: $O(2^n)$}

    \Comment{Total SA$_{\mathrm{approx}}$: $O(2^n \log(n/\epsilon))$}
\end{algorithm}

When compiling the CSP circuit into the $\{\mathrm{H,S,T,CNOT}\}$ gate set, a significant quantum resource consumption ($O(\log(n/\epsilon))$ depth and $O(2^n\log(n/\epsilon))$ spacetime allocation) is contributed by the parallelized CCRy gates. As illustrated in Fig.~\ref{fig:LOADF}, a majority of the qubits in B that performs the parallelized Ry gates can be dirty (See proofing \href{https://algassert.com/quirk#circuit={%22cols%22:[[1,1,1,{%22id%22:%22Ryft%22,%22arg%22:%221%22}],[%22Swap%22,1,1,%22Swap%22],[%22%E2%80%A6%22,%22%E2%80%A6%22,%22%E2%80%A6%22,%22%E2%80%A6%22,{%22id%22:%22Ryft%22,%22arg%22:%221.5pi%20t%22},1,1,{%22id%22:%22Ryft%22,%22arg%22:%222.5pi%20t%22},{%22id%22:%22Ryft%22,%22arg%22:%223.5pi%20t%22},%22X%22,1,%22X%22,%22X%22,%22X%22],[],[%22%E2%80%A2%22,1,1,1,1,1,1,1,1,%22Swap%22,%22Swap%22],[%22%E2%80%A2%22,1,1,%22Swap%22,%22Swap%22],[%22%E2%80%A2%22,1,1,1,1,%22Swap%22,1,%22Swap%22],[%22%E2%80%A2%22,1,1,1,1,1,%22Swap%22,1,%22Swap%22],[1,1,1,{%22id%22:%22Ryft%22,%22arg%22:%223%22},1,1,1,1,1,%22%E2%80%A2%22,1,%22%E2%80%A2%22],[1,1,1,1,{%22id%22:%22Ryft%22,%22arg%22:%222%22},1,1,1,1,1,%22%E2%80%A2%22,%22%E2%80%A2%22],[1,1,1,1,1,{%22id%22:%22Ryft%22,%22arg%22:%227%22},1,1,1,%22%E2%80%A2%22,1,1,%22%E2%80%A2%22],[1,1,1,1,1,1,{%22id%22:%22Ryft%22,%22arg%22:%226%22},1,1,%22%E2%80%A2%22,1,1,1,%22%E2%80%A2%22],[1,1,1,1,1,1,1,{%22id%22:%22Ryft%22,%22arg%22:%225%22},1,1,%22%E2%80%A2%22,1,%22%E2%80%A2%22],[1,1,1,1,1,1,1,1,{%22id%22:%22Ryft%22,%22arg%22:%224%22},1,%22%E2%80%A2%22,1,1,%22%E2%80%A2%22],[%22%E2%80%A2%22,1,1,%22Swap%22,%22Swap%22],[%22%E2%80%A2%22,1,1,1,1,%22Swap%22,1,%22Swap%22],[%22%E2%80%A2%22,1,1,1,1,1,%22Swap%22,1,%22Swap%22],[%22%E2%80%A2%22,1,1,1,1,1,1,1,1,%22Swap%22,%22Swap%22],[1,1,1,1,{%22id%22:%22Ryft%22,%22arg%22:%22-1.5pi%20t%22},1,1,{%22id%22:%22Ryft%22,%22arg%22:%22-2.5pi%20t%22},{%22id%22:%22Ryft%22,%22arg%22:%22-3.5pi%20t%22}],[1,1,1,1,1,1,1,1,1,%22X%22,1,%22X%22,%22X%22,%22X%22],[1,%22Swap%22,1,%22Swap%22],[1,%22%E2%80%A2%22,1,1,1,%22Swap%22,%22Swap%22],[1,1,%22Swap%22,1,1,%22Swap%22],[1,%22%E2%80%A2%22,1,1,1,%22Swap%22,%22Swap%22],[1,1,1,1,1,1,1,1,1,1,1,%22X%22,%22X%22,%22X%22],[1,1,1,1,{%22id%22:%22Ryft%22,%22arg%22:%22pi%20t%22},1,1,{%22id%22:%22Ryft%22,%22arg%22:%222pi%20t%22},{%22id%22:%22Ryft%22,%22arg%22:%223pi%20t%22}],[1,1,1,1,1,1,1,1,1,1,1,%22X%22,%22X%22],[1,%22%E2%80%A2%22,1,1,1,1,1,1,1,1,1,1,%22Swap%22,%22Swap%22],[1,1,1,1,1,1,1,1,1,%22X%22],[%22%E2%80%A2%22,1,1,1,1,1,1,1,1,%22Swap%22,%22Swap%22],[%22%E2%80%A2%22,1,1,1,1,1,%22Swap%22,1,%22Swap%22],[%22%E2%80%A2%22,1,1,1,1,%22Swap%22,1,%22Swap%22],[%22%E2%80%A2%22,1,1,%22Swap%22,%22Swap%22],[1,1,1,{%22id%22:%22Ryft%22,%22arg%22:%22-3%22},1,1,1,1,1,%22%E2%80%A2%22,1,%22%E2%80%A2%22],[1,1,1,1,{%22id%22:%22Ryft%22,%22arg%22:%22-2%22},1,1,1,1,1,%22%E2%80%A2%22,%22%E2%80%A2%22],[1,1,1,1,1,{%22id%22:%22Ryft%22,%22arg%22:%22-7%22},1,1,1,%22%E2%80%A2%22,1,1,%22%E2%80%A2%22],[1,1,1,1,1,1,{%22id%22:%22Ryft%22,%22arg%22:%22-6%22},1,1,%22%E2%80%A2%22,1,1,1,%22%E2%80%A2%22],[1,1,1,1,1,1,1,{%22id%22:%22Ryft%22,%22arg%22:%22-5%22},1,1,%22%E2%80%A2%22,1,%22%E2%80%A2%22],[1,1,1,1,1,1,1,1,{%22id%22:%22Ryft%22,%22arg%22:%22-4%22},1,%22%E2%80%A2%22,1,1,%22%E2%80%A2%22],[%22%E2%80%A2%22,1,1,1,1,1,%22Swap%22,1,%22Swap%22],[%22%E2%80%A2%22,1,1,1,1,%22Swap%22,1,%22Swap%22],[%22%E2%80%A2%22,1,1,%22Swap%22,%22Swap%22],[%22%E2%80%A2%22,1,1,1,1,1,1,1,1,%22Swap%22,%22Swap%22],[1,1,1,1,1,1,1,1,1,%22X%22],[1,%22%E2%80%A2%22,1,1,1,1,1,1,1,1,1,1,%22Swap%22,%22Swap%22],[1,1,1,1,{%22id%22:%22Ryft%22,%22arg%22:%22-pi%20t%22},1,1,{%22id%22:%22Ryft%22,%22arg%22:%22-2pi%20t%22},{%22id%22:%22Ryft%22,%22arg%22:%22-3pi%20t%22},1,1,%22X%22,%22X%22],[1,1,1,1,1,1,1,1,1,1,1,%22X%22,%22X%22,%22X%22]]}}{code} 
for a state preparation task ($m = 1$, $n = 3$) where we use time-dependent gates before and after each segment to emulate the dirty qubits in the $B_1$ register), which would be of abundant supply in fault-tolerant algorithms such as LCU \cite{berry2015simulating}. The ancilla qubits required to perform the Toffoli operations (i.e., $A$ and $F$ registers) can be immediately uncomputed using COPY$^{\dagger}$ once the Toffoli gates are done if the ratio between $\log(1/\epsilon)$ and $n$ is large, thus making the required spacetime allocation of clean qubits $O(N/M\log(\log(N)/\epsilon) + N)$ rather than $O(N\log(\log(N)/\epsilon))$. Future work can also aim to modify the circuit structure to allow, e.g., $D$, $A$, and/or $F$ registers to be dirty.

When considering the number of clean qubits required to perform Ry gates in the context of the whole circuit (both SP and CSP circuits require parallelized Ry gates), we can see that we essentially need $O(M)$ clean qubits for the SP circuit and $O(N/M)$ clean qubits for the CSP circuit. Given enough supply of dirty qubits in the case when we need to reach high quantum state precisions using the $\{\mathrm{H,S,T,CNOT}\}$ gate set, we can let $M = \sqrt{N}$ (in other words, $m = n/2$, which also satisfies Eq.~\eqref{eq:mn_relation}), so the total clean qubit count will now be upper bounded to $O(\sqrt{N})$, thus resulting in $O(\sqrt{N}\log(\log(N)/\epsilon) + N)$ spacetime allocation for the clean qubits.


We illustrate the early ancilla free-up advantage over previous work in Fig.~\ref{fig:clader_imp} and Fig.~\ref{fig:Usp+LOADF} using Quirk \cite{quirk}. As mentioned at the end of Sec.~\ref{sec:intro}, previous methods (e.g., Clader et al., shown in Fig.~\ref{fig:clader_imp}, see \href{https://algassert.com/quirk#circuit={%22cols%22:[[1,1,1,{%22id%22:%22Ryft%22,%22arg%22:%221%22},{%22id%22:%22Ryft%22,%22arg%22:%223%22},{%22id%22:%22Ryft%22,%22arg%22:%222%22},{%22id%22:%22Ryft%22,%22arg%22:%227%22},{%22id%22:%22Ryft%22,%22arg%22:%226%22},{%22id%22:%22Ryft%22,%22arg%22:%225%22},{%22id%22:%22Ryft%22,%22arg%22:%224%22}],[%22Swap%22,1,1,%22Swap%22],[%22%E2%80%A2%22,1,1,1,%22Swap%22,%22Swap%22],[%22%E2%80%A2%22,1,1,1,1,1,%22Swap%22,1,%22Swap%22],[%22%E2%80%A2%22,1,1,1,1,1,1,%22Swap%22,1,%22Swap%22],[1,%22Swap%22,1,1,%22Swap%22],[1,%22%E2%80%A2%22,1,1,1,1,%22Swap%22,%22Swap%22],[1,1,%22Swap%22,1,1,1,%22Swap%22],[1,%22%E2%80%A2%22,1,1,1,1,%22Swap%22,%22Swap%22],[%22%E2%80%A2%22,1,1,1,1,1,1,%22Swap%22,1,%22Swap%22],[%22%E2%80%A2%22,1,1,1,1,1,%22Swap%22,1,%22Swap%22],[%22%E2%80%A2%22,1,1,1,%22Swap%22,%22Swap%22]]}}{code}) 
require $O(N)$ ancilla qubits to be entangled with the $O(n)$ data qubits for $O(n)$ depth, leading to $O(n2^n)$ spacetime allocation. This work (illustrated in Fig.~\ref{fig:Usp+LOADF}, see \href{https://algassert.com/quirk#circuit={%22cols%22:[[1,1,1,{%22id%22:%22Ryft%22,%22arg%22:%221%22}],[%22Swap%22,1,1,%22Swap%22],[%22%E2%80%A2%22,%22X%22],[%22%E2%80%A2%22,1,%22X%22],[1,1,1,1,1,1,1,1,1,1,%22X%22],[%22%E2%80%A2%22,1,1,1,1,1,1,1,1,1,%22Swap%22,%22Swap%22],[%22%E2%80%A2%22,1,1,1,%22Swap%22,%22Swap%22],[1,%22%E2%80%A2%22,1,1,1,1,%22Swap%22,1,%22Swap%22],[1,1,%22%E2%80%A2%22,1,1,1,1,%22Swap%22,1,%22Swap%22],[1,1,1,1,1,1,1,1,1,1,%22%E2%80%A2%22,1,%22X%22],[1,1,1,1,1,1,1,1,1,1,%22%E2%80%A2%22,1,1,%22X%22],[1,1,1,1,1,1,1,1,1,1,1,%22%E2%80%A2%22,1,1,%22X%22],[1,1,1,1,1,1,1,1,1,1,1,%22%E2%80%A2%22,1,1,1,%22X%22],[1,1,1,1,{%22id%22:%22Ryft%22,%22arg%22:%223%22},1,1,1,1,1,%22%E2%80%A2%22],[1,1,1,1,1,{%22id%22:%22Ryft%22,%22arg%22:%222%22},1,1,1,1,1,%22%E2%80%A2%22],[1,1,1,1,1,1,{%22id%22:%22Ryft%22,%22arg%22:%227%22},1,1,1,1,1,%22%E2%80%A2%22],[1,1,1,1,1,1,1,{%22id%22:%22Ryft%22,%22arg%22:%226%22},1,1,1,1,1,%22%E2%80%A2%22],[1,1,1,1,1,1,1,1,{%22id%22:%22Ryft%22,%22arg%22:%225%22},1,1,1,1,1,%22%E2%80%A2%22],[1,1,1,1,1,1,1,1,1,{%22id%22:%22Ryft%22,%22arg%22:%224%22},1,1,1,1,1,%22%E2%80%A2%22],[1,1,1,1,1,1,1,1,1,1,1,%22%E2%80%A2%22,1,1,1,%22X%22],[1,1,1,1,1,1,1,1,1,1,1,%22%E2%80%A2%22,1,1,%22X%22],[1,1,1,1,1,1,1,1,1,1,%22%E2%80%A2%22,1,1,%22X%22],[1,1,1,1,1,1,1,1,1,1,%22%E2%80%A2%22,1,%22X%22],[1,1,%22%E2%80%A2%22,1,1,1,1,%22Swap%22,1,%22Swap%22],[1,%22%E2%80%A2%22,1,1,1,1,%22Swap%22,1,%22Swap%22],[%22%E2%80%A2%22,1,1,1,%22Swap%22,%22Swap%22],[%22%E2%80%A2%22,1,1,1,1,1,1,1,1,1,%22Swap%22,%22Swap%22],[1,1,1,1,1,1,1,1,1,1,%22X%22],[%22%E2%80%A2%22,1,%22X%22],[%22%E2%80%A2%22,%22X%22]]}}{code}, 
and also on the right part of Fig.~\ref{fig:SA_comps}), only requires most of the $O(N)$ qubits to be allocated briefly.
\begin{figure}[h!]
    \centering
    \includegraphics[width=0.65\textwidth]{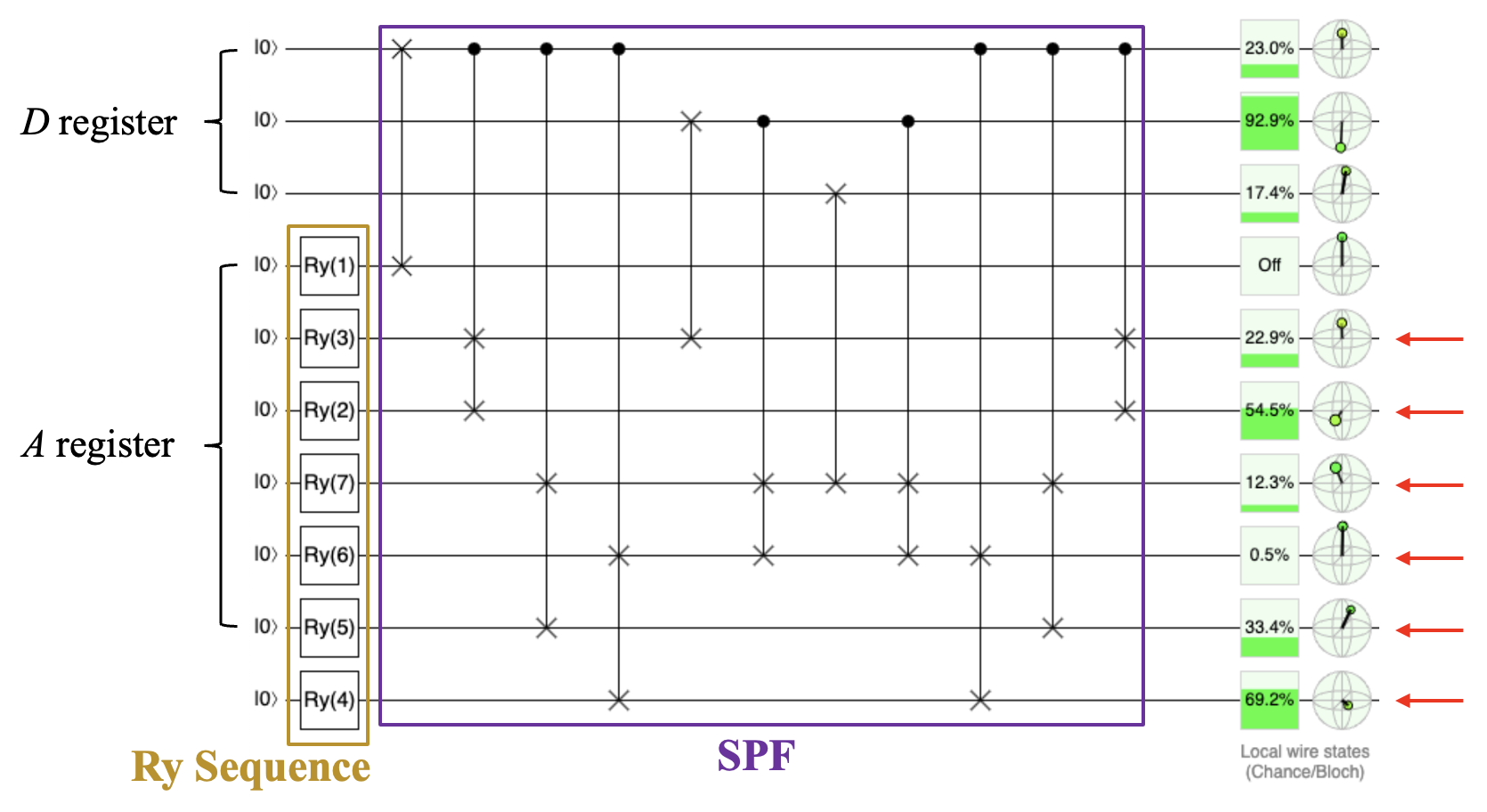}
    \caption{First part of Clader et al.~\cite{clader2023quantum} (aka $U_{\mathrm{SP}}$ without COPY). Note that all the ancilla qubits in the $A$ register are left \textbf{entangled} with the data qubit in the $D$ register (pointed by the \textcolor{red}{red} arrows). These $O(N)$ ancilla qubits are not returned to $\ket{0}$ until after the execution of the depth-$O(n)$ FLAG subroutine.}
\label{fig:clader_imp}
\end{figure}

\begin{figure}[h!]
\centering
\makebox[\textwidth][c]{
\includegraphics[width=1.0\textwidth]{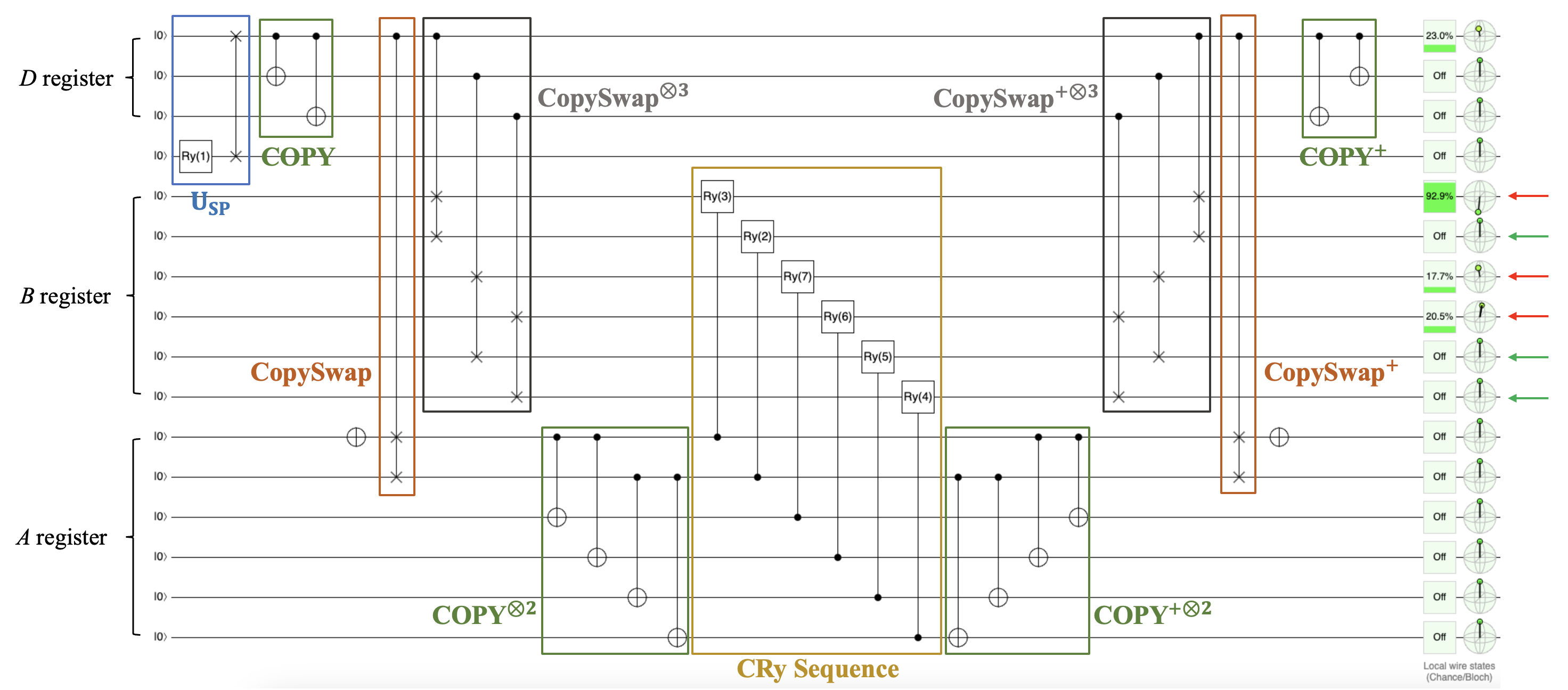}
}
\caption{$U_{\mathrm{SP}}$ + LOADF part of the circuit developed in this work, on an example with $m=1$ and $n=3$. Note that most of the ancilla qubits in the $B$ register are freed up (pointed by the \textcolor{green}{green} arrows). The remaining $O(\frac{M}{N})$ ancilla qubits (pointed by the \textcolor{red}{red} arrows) are freed up after the following SPF circuit that only takes $O(m)$ depth rather than $O(n)$ depth in the previous case. All other ancilla qubits in $D$, $B$, and $A$ registers are freed up almost immediately after initialization. Note that as $N$ grows larger and if the relation in Eq.~\eqref{eq:mn_relation} is satisfied, we can see a much larger ratio between the qubits labeled with the \textcolor{red}{red} arrows and qubits labeled with the \textcolor{green}{green} arrows, which shows the full advantage of this work's approach.}
\label{fig:Usp+LOADF}
\end{figure}

\newpage

\section{Action of our protocol when the input state is not \texorpdfstring{$\ket{0^n}$}{TEXT}}\label{sec:input_not_0}

Our state preparation procedure acts on $n$ data qubits and uses a number $\ell = O(N)$ of ancilla qubits. It implements a unitary operation $U$ on the $n+\ell$ qubits that sends $\ket{0^n}\ket{0^\ell} \mapsto \ket{\psi}\ket{0^\ell}$, that is, all the ancillae begin and end in $\ket{0}$. However, if the same unitary is performed on a state $\ket{z}\ket{0^\ell}$ for some $z \neq 0^n$ (or more generally, a state which is a superposition of all $z \neq 0^n$), it will no longer be the case that all ancillas are reset to $\ket{0}$. This results from the SPF circuit structure, described in Sec.~\ref{sec:SPF_circ}, which enacts rotations by swapping the $n$ data qubits into the ancilla register. It exploits the fact that the data qubits are known to start in $\ket{0}$ in order to guarantee that the ancillae are left in the state $\ket{0}$.

This feature could be problematic in applications where the state preparation unitary is a part of a larger algorithm, and does not always act on the state $\ket{0^n}$. Is the procedure still useful in these other cases, and does it still achieve efficient spacetime allocation? 
Here we argue that the favorable properties of our procedure, in particular its optimal depth and spacetime allocation, extend to other common situations where state preparation appears, for example, implementing reflections and projections. 

\subsection{Implementing reflections about arbitrary states}
A common use of state preparation within larger algorithms is to perform a reflection about a particular state $\ket{\psi}$, that is, the operation 
\begin{equation}
    R = I_n - 2 \ket{\psi}\bra{\psi}\,,
\end{equation}
where $I_n$ denotes the identity operator on $n$ qubits.  

We now discuss how to implement this operator using $U$. First, we make a distinction between different ancilla registers in our state-preparation protocol. For the SP portion of the protocol, we have ancilla registers $A$ and $F$, each of size $M-1$, depicted in Fig.~\ref{fig:SP}. For the CSP portion, we have an additional register $B$ of size $N/M-1$ (register $F$ is the same as the one from the SP portion). The figures depict the registers $A$ and $B$ beginning and ending in $\ket{0}$, but this is only the case because some of the data registers have input $\ket{0}$. On the other hand, the $F$ register has the property that it begins in $\ket{0}$ if and only if it ends in $\ket{0}$, regardless of the state of the other registers. Additionally, the implementation of SPF, FLAG, and LOADF introduce additional ancilla registers, as depicted in Figs.~\ref{fig:SPF}, \ref{fig:FLAG}, \ref{fig:LOADF}, but these ancilla registers are similar to the F register: they begin in $\ket{0}$ if and only if they end in $\ket{0}$, regardless of the state elsewhere in the circuit. Thus, we separate the $\ell$ ancillae into two groups, the group of $\ell' = M+N/M-2$ ancillae in registers $A$ and $B$, and the other $\ell'' = \ell - \ell'$ ancillae. 

With this in mind, we can express
\begin{equation}
    R \otimes \ket{0^{\ell'}} \otimes \ket{0^{\ell''}} = U \Big[I_{n}\otimes I_{\ell'} \otimes I_{\ell''} - 2\ket{0^n}\bra{0^n}\otimes \ket{0^{\ell'}}\bra{0^{\ell'}}\otimes I_{\ell''}\Big] U^\dagger\Big[I_n \otimes \ket{0^{\ell'}}\otimes \ket{0^{\ell''}}\Big]\,.
\end{equation}
Let us verify the above formula. When the input is $\ket{\psi}$, the action of $U^\dagger$ yields $\ket{0^n}\ket{0^\ell}$, a sign is applied, and application of $U$ outputs the state $-\ket{\psi}\ket{0^\ell}$, as expected. When the input is $\ket{\perp}$ orthogonal to $\ket{\psi}$, application of $U^\dagger$ yields a state $\ket{\perp'}\ket{0^{\ell''}}$, where all we can guarantee is that $\ket{\perp'}$ is orthogonal to $\ket{0^n}\ket{0^{\ell'}}$. The reflection operation does not apply a sign, and subsequent application of $U$ outputs $\ket{\perp}\ket{0^\ell}$, as expected.  Crucially, this requires that we perform a reflection about both the $n$ data qubits \emph{and} the $\ell'$ ancillae that make up registers $A$ and $B$ all simultaneously being in $\ket{0}$ (but not the other $\ell''$ ancillae as they are already guaranteed to be in $\ket{0}$). A reflection about $t$ qubits being in the state $\ket{0}$ can be implemented in depth $O(\log(t))$ using $O(t)$ ancillae and $O(t)$ Toffoli gates by computing whether all qubits are set to $\ket{0}$ in a tree-like fashion, and applying a phase if the result is 1. Here $t = n + N/M + M -2$, so the depth is $O(n)$. The spacetime allocation is upper bounded by the depth times the number of active qubits, i.e.~$O(n \cdot \max(N/M,M))$. Since we choose $M$ such that $\max(N/M,M) = O(N/n)$, this spacetime allocation is at most $O(N)$. In conclusion, our state preparation method can be used to perform reflections about arbitrary states on $n$ qubits in depth $O(n)$ and spacetime allocation $O(N)$. 

\newpage

\subsection{Implementing projections and block-encodings involving projections}

Another operator where state-preparation is relevant is the projection onto the complement of an arbitrary state $\ket{\psi}$, that is 
\begin{equation}
    P = I - \ket{\psi}\bra{\psi}\,.
\end{equation}
This operator appears, for example, in the matrices that form the adiabatic path used in query-optimal quantum linear systems solvers \cite{costa2022optimal}. In that context, Appendix F of Ref.~\cite{costa2022optimal} explains how a block-encoding of the relevant matrix can be formed given a state preparation unitary that maps $\ket{0} \mapsto \ket{\psi}$. This larger block-encoding involves constructing a block-encoding of the projector $P$, which is reduced to implementing a reflection about $\ket{\psi}$ (see Figs.~1--6 of Ref.~\cite{dalzell2022socp}, especially Fig.~4, for a quantum circuit interpretation of the discussion in Appendix F of Ref.~\cite{costa2022optimal}). As illustrated above, our method can perform reflections with optimal depth and spacetime allocation, and thus can also be applied in this application.

\section{Complex amplitudes}\label{app:complex-amplitudes}

The constructions presented in the main text can prepare arbitrary states with real non-negative coefficients. Extending the construction to work for arbitrary coefficients is straightforward. The SP portion of the SP+CSP protocol is unchanged, as the state $\ket{\phi}$, as defined in Eq.~\eqref{eq:phi_def} is insensitive to any phases in the vector $\mathbf{x}$ (all $y_i$ are positive). On the other hand, the CSP portion of the protocol must be slightly modified. For each $k =0,\ldots,M-1$, $j = 0,\ldots, N/M-1$, define the phase  \begin{equation}
    e^{i\varphi^{(k)}_j} = 
    \frac{x_{j+kN/M}}{\lvert x_{j+kN/M} \rvert}
\end{equation}
(if $x_{j+kN/M}=0$, then the phase  can be defined arbitrarily). 

Next, we redefine the state $\ket{\theta^{(k)}_{s,p}}$:
\begin{equation}
    \ket{\theta^{(k)}_{s,p}} = \begin{cases}
    \cos(\theta^{(k)}_{s,p}/2)\ket{0} + \sin(\theta^{(k)}_{s,p}/2)\ket{1}& \text{if } s < n-m-1\\
    e^{i\varphi^{(k)}_{2p}}\cos(\theta^{(k)}_{s,p}/2)\ket{0} + e^{i\varphi^{(k)}_{2p+1}}\sin(\theta^{(k)}_{s,p}/2)\ket{1}& \text{if } s = n-m-1
    \end{cases}
\end{equation}

To prepare $\ket{\theta^{(k)}_{s,p}}$ we need to apply a single-qubit rotation about the $y$-axis, and then,  if $s=n-m-1$, we must apply a single qubit rotation about the $z$-axis, as well as a global phase (the global phase will be important when we add controls). To prepare a state with complex amplitudes, the only aspect of the protocol that has to change is the LOADF operation, which appears twice in the circuit for $U_{\mathrm{CSP}}$ in Fig.~\ref{fig:CSP}. As seen in Fig.~\ref{fig:LOADF}, implementing LOADF involves a doubly controlled $R_y$ gate. To account for complex amplitudes, we must augment this step with a doubly controlled $R_z$ gate, as well as a doubly controlled global phase (which is equivalent to a singly-controlled $R_z$ gate). These additions lead to at most a constant factor increase in the depth and spacetime allocation. In the approximate $\{\mathrm{H,S,T,CNOT}\}$ gate set, all rotations need only be synthesized to error $O(\epsilon/n)$ to achieve overall error $\epsilon$ in the state preparation protocol, requiring only $O(\log(n/\epsilon))$ gates per rotation. None of the complexity statements in the paper are impacted. 

\end{document}